\documentclass[12pt]{article}

\usepackage{paper_style}

\title{Solving the $n$-color ice model}

\author{Patrick Addona, Ethan Bockenhauer, Ben Brubaker, \\ Michael Cauthorn, Cianan Conefrey-Shinozaki, David Donze, \\ William Dudarov, Jessamyn Dukes, Andrew Hardt, Cindy Li, \\ Jigang Li, Yanli Liu, Neelima Puthanveetil, Zain Qudsi, \\ Jordan Simons, Joseph Sullivan, and Autumn Young}

\date{\today}

\begin{document} \maketitle

\begin{abstract}
Given an arbitrary choice of two sets of nonzero Boltzmann weights for $n$-color lattice models, we provide explicit algebraic conditions on these Boltzmann weights which guarantee a solution (i.e., a third set of weights) to the Yang-Baxter equation. Furthermore we provide an explicit one-dimensional parametrization of all solutions in this case. These $n$-color lattice models are so named because their admissible vertices have adjacent edges labeled by one of $n$ colors with additional restrictions. The two-colored case specializes to the six-vertex model, in which case our results recover the familiar quadric condition of Baxter for solvability. The general $n$-color case includes important solutions to the Yang-Baxter equation like the evaluation modules for the quantum affine Lie algebra $U_q(\hat{\mathfrak{sl}}_n)$. Finally, we demonstrate the invariance of this class of solutions under natural transformations, including those associated with Drinfeld twisting.
\end{abstract}

\section{Introduction}

Lattice models are discrete dynamical systems on two-dimensional lattices arising in statistical mechanics \cite{Baxter-book}. Local interactions at each \emph{vertex} in the lattice are described in terms of \emph{Boltzmann weights}, and these interactions can be combined into a global weight for each admissible configuration on the lattice by taking the product of weights over all vertices. The weighted sum over all admissible configurations with fixed boundary conditions is called the \emph{partition function} of the model. From the perspective of statistical mechanics, the partition function is related to important physical quantities, such as the energy of the system and the associated Gibbs measure. 

A lattice model is called \emph{solvable} (or sometimes \emph{integrable}) if its Boltzmann weights admit a solution to the \emph{(quantum) Yang-Baxter} equation. In this remarkable case, the partition function of the lattice model satisfies symmetries, or sometimes recursion relations, that lead to closed form expressions of the partition function (hence the term \emph{solvable}). Such solutions to the Yang-Baxter equation are difficult to find, but highly prized, as the associated partition functions describe important special functions in many areas of mathematics, including knot theory, integrable probability, Schubert calculus, orthogonal polynomials, and $p$-adic representation theory. 

A solution to the Yang-Baxter equation is often expressed in terms of an $R$-matrix. In its algebraic formulation, the Yang-Baxter equation is an identity of endomorphisms. Given vector spaces $U,V$, and $W$ and endomorphisms $R\in\End(U\otimes V)$, $S\in\End(V\otimes W)$, and $T\in\End(U\otimes W)$, the Yang-Baxter equation is the relationship \[RST = TSR \hspace{10pt} \text{as elements of} \hspace{10pt} \End(U\otimes V\otimes W),\] where each of $R,S,T$ acts on the appropriate tensor factors (and as the identity on the third factor). We are often given $S$ and $T$ and asked to solve for the matrix $R$. 

Partly motivated by an effort to find sources for Yang-Baxter equations, Drinfeld \cite{Drinfeld-hopf-quantum, Drinfeld-ICM} and Jimbo \cite{Jimbo-q-difference, Jimbo-Hecke-algebra} were led to define and study quantum enveloping algebras, also known as \emph{quantum groups}. These include $q$-deformations of universal enveloping algebras of Lie algebras and are examples of quasi-triangular Hopf algebras whose associated module category is braided. It is from this structure that we obtain solutions to the Yang-Baxter equation from modules of quantum groups. There's a general heuristic linking solutions to the Yang-Baxter equation to algebraic structures like the modules of quasitriangular Hopf algebras. A precise version of this connection is given by the Faddeev-Reshetikhin-Takhtajian construction (see \cite[VIII.6]{Kassel-quantum-groups}).

Absent this quantum group structure, much work has been done, particularly by physicists, on solutions to the Yang-Baxter equation for general classes of models with a fixed collection of admissible vertices. Given such a fixed set of admissible vertices, one can ask for the following:
\begin{itemize}
    \item Necessary and sufficient conditions on the Boltzmann weights of $S$ and $T$ such that the lattice model is solvable.
    \item A parametrization of all such Yang-Baxter solutions.
\end{itemize}
Providing answers to these questions is often referred to as \emph{solving the model}. Some authors reserve this term for the associated explicit expression for the partition function of the lattice model, though this typically follows in a straightforward way from the existence of Yang-Baxter equations.

The prototypical example is Baxter's solution of the six-vertex model on the square lattice, where every vertex has four adjacent edges \cite{Baxter-inversion-relation, Baxter-book}. This model has six admissible vertices, where adjacent edges are labelled with a $0$ or $1$ in a pattern which must follow the \emph{ice rule} (see next section). Figure~\ref{example-state}, has an example state for the six-vertex model, in which each of the possible six vertices satisfying the ice rule appear.

\begin{figure}[h]
\begin{center}
\scalebox{0.8}{
\begin{tikzpicture}
  \coordinate (ab) at (1,0);
  \coordinate (ad) at (3,0);
  \coordinate (af) at (5,0);
  \coordinate (ah) at (7,0);
  \coordinate (ba) at (0,1);
  \coordinate (bc) at (2,1);
  \coordinate (be) at (4,1);
  \coordinate (bg) at (6,1);
  \coordinate (bi) at (8,1);
  \coordinate (cb) at (1,2);
  \coordinate (cd) at (3,2);
  \coordinate (cf) at (5,2);
  \coordinate (ch) at (7,2);
  \coordinate (da) at (0,3);
  \coordinate (dc) at (2,3);
  \coordinate (de) at (4,3);
  \coordinate (dg) at (6,3);
  \coordinate (di) at (8,3);
  \coordinate (eb) at (1,4);
  \coordinate (ed) at (3,4);
  \coordinate (ef) at (5,4);
  \coordinate (eh) at (7,4);
  \coordinate (fa) at (0,5);
  \coordinate (fc) at (2,5);
  \coordinate (fe) at (4,5);
  \coordinate (fg) at (6,5);
  \coordinate (fi) at (8,5);
  \coordinate (gb) at (1,6);
  \coordinate (gd) at (3,6);
  \coordinate (gf) at (5,6);
  \coordinate (gh) at (7,6);
  \coordinate (bb) at (1,1);
  \coordinate (bd) at (3,1);
  \coordinate (bf) at (5,1);
  \coordinate (bh) at (7,1);
  \coordinate (db) at (1,3);
  \coordinate (dd) at (3,3);
  \coordinate (df) at (5,3);
  \coordinate (dh) at (7,3);
  \coordinate (fb) at (1,5);
  \coordinate (fd) at (3,5);
  \coordinate (ff) at (5,5);
  \coordinate (fh) at (7,5);
  \coordinate (bax) at (0,1.5);
  \coordinate (bcx) at (2,1.5);
  \coordinate (bex) at (4,1.5);
  \coordinate (bgx) at (6,1.5);
  \coordinate (bix) at (8,1.5);
  \coordinate (dax) at (0,3.5);
  \coordinate (dcx) at (2,3.5);
  \coordinate (dex) at (4,3.5);
  \coordinate (dgx) at (6,3.5);
  \coordinate (dix) at (8,3.5);
  \coordinate (fax) at (0,5.5);
  \coordinate (fcx) at (2,5.5);
  \coordinate (fex) at (4,5.5);
  \coordinate (fgx) at (6,5.5);
  \coordinate (fix) at (8,5.5);
  \draw (ab)--(gb);
  \draw (ad)--(gd);
  \draw (af)--(gf);
  \draw (ah)--(gh);
  \draw (ba)--(bi);
  \draw (da)--(di);
  \draw (fa)--(fi);
  \draw[line width=0.5mm,blue] (fa)--(fi);
  \draw[line width=0.5mm,blue] (gb)--(db)--(dd)--(bd)--(bi);
  \draw[line width=0.5mm,blue] (gf)--(af);
  \draw[line width=0.5mm,red] (ba)--(bd)--(ad);
  \draw[line width=0.5mm,red] (gd)--(dd)--(dh)--(ah);
  \draw[line width=0.5mm,red] (da)--(db)--(ab);
  \draw[line width=0.5mm,red] (gh)--(dh)--(di);
  \draw[line width=0.5mm,red,fill=white] (ab) circle (.25);
  \draw[line width=0.5mm,red,fill=white] (ad) circle (.25);
  \draw[line width=0.5mm,blue,fill=white] (af) circle (.25);
  \draw[line width=0.5mm,red,fill=white] (ah) circle (.25);
  \draw[line width=0.5mm,red,fill=white] (ba) circle (.25);
  \draw[line width=0.5mm,red,fill=white] (bc) circle (.25);
  \draw[line width=0.5mm,blue,fill=white] (be) circle (.25);
  \draw[line width=0.5mm,blue,fill=white] (bg) circle (.25);
  \draw[line width=0.5mm,blue,fill=white] (bi) circle (.25);
  \draw[line width=0.5mm,red,fill=white] (cb) circle (.25);
  \draw[line width=0.5mm,blue,fill=white] (cd) circle (.25);
  \draw[line width=0.5mm,blue,fill=white] (cf) circle (.25);
  \draw[line width=0.5mm,red,fill=white] (ch) circle (.25);
  \draw[line width=0.5mm,red,fill=white] (da) circle (.25);
  \draw[line width=0.5mm,blue,fill=white] (dc) circle (.25);
  \draw[line width=0.5mm,red,fill=white] (de) circle (.25);
  \draw[line width=0.5mm,red,fill=white] (dg) circle (.25);
  \draw[line width=0.5mm,red,fill=white] (di) circle (.25);
  \draw[line width=0.5mm,blue,fill=white] (eb) circle (.25);
  \draw[line width=0.5mm,red,fill=white] (ed) circle (.25);
  \draw[line width=0.5mm,blue,fill=white] (ef) circle (.25);
  \draw[line width=0.5mm,red,fill=white] (eh) circle (.25);
  \draw[line width=0.5mm,blue,fill=white] (fa) circle (.25);
  \draw[line width=0.5mm,blue,fill=white] (fc) circle (.25);
  \draw[line width=0.5mm,blue,fill=white] (fe) circle (.25);
  \draw[line width=0.5mm,blue,fill=white] (fg) circle (.25);
  \draw[line width=0.5mm,blue,fill=white] (fi) circle (.25);
  \draw[line width=0.5mm,blue,fill=white] (gb) circle (.25);
  \draw[line width=0.5mm,red,fill=white] (gd) circle (.25);
  \draw[line width=0.5mm,blue,fill=white] (gf) circle (.25);
  \draw[line width=0.5mm,red,fill=white] (gh) circle (.25);
  \node at (gb) {$1$};
  \node at (gd) {$0$};
  \node at (gf) {$1$};
  \node at (gh) {$0$};
  \node at (fa) {$1$};
  \node at (fc) {$1$};
  \node at (fe) {$1$};
  \node at (fg) {$1$};
  \node at (fi) {$1$};
  \node at (eb) {$1$};
  \node at (ed) {$0$};
  \node at (ef) {$1$};
  \node at (eh) {$0$};
  \node at (da) {$0$};
  \node at (dc) {$1$};
  \node at (de) {$0$};
  \node at (dg) {$0$};
  \node at (di) {$0$};
  \node at (cb) {$0$};
  \node at (cd) {$1$};
  \node at (cf) {$1$};
  \node at (ch) {$0$};
  \node at (ba) {$0$};
  \node at (bc) {$0$};
  \node at (be) {$1$};
  \node at (bg) {$1$};
  \node at (bi) {$1$};
  \node at (ab) {$0$};
  \node at (ad) {$0$};
  \node at (af) {$1$};
  \node at (ah) {$0$};
\end{tikzpicture}}
\end{center}
\caption{An example lattice model state for the ice-type six-vertex model. Here, we have two labels, 0 and 1. Note that the edges with a given label form paths that propagate through the grid.}
\label{example-state}
\end{figure}

In his treatment of the six-vertex model, Baxter requires the weights to be {\em symmetric}, that is invariant when swapping labels $0$ and $1$. This is sometimes referred to as the ``field-free'' case. Baxter found that for these lattice models, solvability is governed by a quadric $\Delta$ in the Boltzmann weights of vertex type $S$ and $T$ -- the model is solvable if and only if $\Delta(S) = \Delta(T)$. Moreover $\Delta$ plays an important role in the physical properties of the model.

In this special case on the square lattice, the Yang-Baxter equation is expressible as an identity of lattice model partition functions. The matrices $R, S$, and $T$ can be viewed as three different types of vertex, each with its own set of Boltzmann weights. In pictorial form, the Yang-Baxter equation becomes an equality of partition functions of the following two lattices for every choice of the six boundary edge labels $E_1,E_2,E_3,F_1,F_2,F_3$:
\begin{align}\label{generic-YBE-diagram}
\begin{array}{c}
\scalebox{.85}{
\begin{tikzpicture}
  \draw (0,0)--(2,0);
  \draw (0,2)--(2,2);
  \draw (1,-1)--(1,3);
  \coordinate (a1) at (-2,0);
  \coordinate (c1) at (0,2);
  \coordinate (a2) at (-2,2);
  \coordinate (c2) at (0,0);
  \draw (a1) to [out=0,in=180] (c1);
  \draw (a2) to [out=0,in=180] (c2);
  \draw[fill=white] (-2,0) circle (.35);
  \draw[fill=white] (-2,2) circle (.35);
  \draw[fill=white](1,3) circle(.35);
  \draw[fill=white] (2,2) circle (.35);
  \draw[fill=white] (2,0) circle (.35);
  \draw[fill=white](1,-1) circle (.35);
  \draw[fill=white] (0,2) circle (.35);
  \draw[fill=white] (0,0) circle (.35);
  \draw[fill=white](1,1) circle (.35);
  \node at (-2,0) {$E_1$};
  \node at (-2,2) {$E_2$};
  \node at (1,3) {$E_3$};
  \node at (2,2) {$F_1$};
  \node at (2,0) {$F_2$};
  \node at (1,-1) {$F_3$};
  \path[fill=white] (-1,1) circle (.35);
  \path[fill=white] (1,2) circle (.35);
  \path[fill=white] (1,0) circle (.35);
  \node at (1,2) {$S$};
  \node at (1,0) {$T$};
  \node at (-1,1){$R$};
  \node at (3,1) {{\Large $=$}};
  \end{tikzpicture}}
\hspace{1em}
\scalebox{.85}{\begin{tikzpicture}
  \draw (0,0)--(-2,0); 
  \draw (0,2)--(-2,2);
  \draw (-1,-1)--(-1,3);
  \coordinate (a1) at (2,0);
  \coordinate (c1) at (0,2);
  \coordinate (a2) at (2,2);
  \coordinate (c2) at (0,0);
  \draw (a1) to [out=180,in=0] (c1);
  \draw (a2) to [out=180,in=0] (c2);
  \draw[fill=white] (-2,0) circle (.35);
  \draw[fill=white] (-2,2) circle (.35);
  \draw[fill=white](-1,3) circle(.35);
  \draw[fill=white] (2,2) circle (.35);
  \draw[fill=white] (2,0) circle (.35);
  \draw[fill=white](-1,-1) circle (.35);
  \draw[fill=white] (0,2) circle (.35);
  \draw[fill=white] (0,0) circle (.35);
  \draw[fill=white](-1,1) circle (.35);
  \node at (-2,0) {$E_1$};
  \node at (-2,2) {$E_2$};
  \node at (-1,3) {$E_3$};
  \node at (2,2) {$F_1$};
  \node at (2,0) {$F_2$};
  \node at (-1,-1) {$F_3$};
  \path[fill=white] (1,1) circle (.35);
  \path[fill=white] (-1,2) circle (.35);
  \path[fill=white] (-1,0) circle (.35);
  \node at (-1,2) {$T$};
  \node at (-1,0) {$S$};
  \node at (1,1) {$R$};
  \end{tikzpicture}}\end{array}.
\end{align}
We've used a common shorthand here, writing that the lattice configurations are equal when we mean an equality of their partition functions. That partition function is typically a sum over a very small set, as there are just three internal edges to be prescribed in each configuration. A precise definition in terms of lattice models is presented in the next section.

Baxter's result was extended by Korepin, Bogoliubov, and Izergin \cite{KBI} and by Brubaker, Bump, and Friedberg~\cite{BBF-Schur-polynomials} to obtain the solution to the six-vertex model in full generality. In this non-field-free setting, $\Delta$ is replaced by a pair of invariants.

In this paper, we {\em solve} a natural generalization of the six-vertex model -- the $n$-color lattice model, for any fixed positive integer $n$. We may view the six-vertex model as a two-color lattice model (as depicted in Figure~\ref{example-state}) whose admissible vertices are selected from the 16 possible vertex configurations so that the resulting states form colored paths moving downward and rightward through the model, any of which may cross or not at a vertex. That choice agrees with the set of admissible vertices picked out by matching Boltzmann weights with entries of the $R$-matrix for the standard module of $U_q(\mathfrak{sl}_2)$ -- a four-by-four matrix supported at six entries. (See \cite[Section 7.5]{ChariPressley} for the association of lattice model Boltzmann weights to $R$-matrix entries. The paper \cite{BBF-Schur-polynomials} also matches our pictorial approach very closely.) Our $n$-colored models have admissible states determined similarly. Its admissible vertices combine to make colored paths traveling downward and rightward through the lattice with no restrictions on crossings. An example with $n=4$ is given in Figure~\ref{example-color-state} and the admissible vertex types are stated precisely in Figure~\ref{rect-vertices}. Again, these precisely match the admissible vertices picked out by the non-zero entries of the $R$-matrix for the standard module of $U_q(\mathfrak{sl}_n)$.\footnote{The reason that admissible vertices corresponding to non-zero $R$-matrix entries result in the path dynamics described above is straightforward to explain. In the associated quantum group module, each weight space is one-dimensional and can be assigned a unique color. The fact that the Cartan subalgebra, under comultiplication, commutes with the $R$-matrix ensures that weight spaces are preserved. At the level of lattice model vertices, this then implies that the colored edges coming into a vertex, from above and left, must match the colored paths exiting the vertex below and to the right, hence forming paths.}

Other solutions for classes of colored lattice models have been obtained by Perk and Schultz \cite{PerkSchultz-1, PerkSchultz-2}, by Sun, Wang, Wu, and collaborators \cite{SunWangWu, SunWangWu-eight-vertex, RenWangWu}, and by many others e.g. \cite{Khachatryan, Vieira, VieiraLima-Santos, Hardt-Hamiltonians}. We remark that the Perk-Schultz use of color is similar to ours--this relationship is discussed below--while the Sun-Wang-Wu use of color is substantially different.

Our main results appear in Section~\ref{sec:n-label}, and we combine results here to provide a version of our main theorem -- conditions for the solvability of the $n$-color model.

\begin{mainthm}[Theorems \ref{nlabelSol} and \ref{nlabelCond}] \label{main-theorem}
The non-degenerate $n$-color ice-type lattice model has nonzero solutions to the Yang-Baxter equation precisely when the $4n^3-11n^2+7n$ conditions given in \emph{(\ref{n-color-conditions})} are satisfied. In this case, the solution $R$ is unique up to scalar multiple and is given in \emph{(\ref{n-color-parametrization})} in terms of the Boltzmann weights of $S$ and $T$.
\end{mainthm}

One interesting aspect of the conditions (\ref{n-color-conditions}) is the appearance of analogues of Baxter's $\Delta$. Namely, there are $n(n-1)$ quadrics $\Delta_{ij}$, one for every ordered pair of distinct colors $i$ and $j$, such that $\Delta_{ij}(S) = \Delta_{ij}(T)$ is a necessary condition for solvability. However, there are additional conditions for solvability which are not easily expressed as invariants of the model; that is, the conditions can't be easily separated into algebraic relations involving only the Boltzmann weights for $S$ versus that of $T$. It is an interesting open question whether these conditions can be rephrased in those terms.

It is then natural to consider $n$-color lattice models where (some of) the quadrics $\Delta_{ij}$ vanish. In the six-vertex model, this is the well-known \emph{free fermion point}. Baxter~\cite{Baxter-book} showed for the symmetric six-vertex model that when $S$ and $T$ are free fermionic, $R$ is also free fermionic, and this was later generalized to all six-vertex models~\cite{KBI,BBF-Schur-polynomials}. In the $n$-color ice model, the vanishing of $\Delta_{ij}$ can be considered independently for each pair of colors. It turns out that results like the above hold \emph{pairwise}: if for any labels $i,j$, $\Delta_{ij}(S) = \Delta_{ij}(T) = 0$, then $\Delta_{ij}(R) = 0$ where it is defined. See Proposition \ref{ff-prop}.

The free fermion point corresponds to the center of the disordered regime~\cite[Section~8.10]{Baxter-book} where all motion is entropic -- the Boltzmann weights, which express the energy of the configuration, provide no energetic penalty nor reward for paths to touch or collide. Partition functions of free fermionic six-vertex models can be expressed as determinants \cite{AggarwalBorodinPetrovWheeler, Naprienko-ff}, and as $\tau$-functions of discrete-time Hamiltonian operators \cite{Hardt-Hamiltonians}. In fact, the latter interpretation generalizes to ice-type lattice models with ``charge'' (see \cite[Section~8]{Hardt-Hamiltonians}). The generalization of discrete-time Hamiltonian operators to involve ``color'' is an interesting open problem, and it is conceivable that these objects might correspond with solvable $n$-color lattice models with $\Delta_{ij}=0$.

Specific choices of Boltzmann weights for the $n$-color ice model have featured prominently in the literature. For example, the $R$-matrix for evaluation representations of $U_q(\mathfrak{gl}_n)$ found by Jimbo~\cite{Jimbo-Uqsln} and of $U_q(\mathfrak{gl}_{1|n})$ as in \cite{Kojima-superalgebra} are supported on the set of $n$-color admissible vertices, as alluded to above. The use of color here to denote different labels (or equivalently basis elements in quantum group modules) is due to Borodin and Wheeler \cite{BorodinWheeler-bosonic}, and solvable lattice models using color in various ways have been recently used to study functions such as LLT polynomials \cite{AggarwalBorodinWheeler-fermionic, CGKM-LLT, CFYZZ-super-LLT}, Grothendieck polynomials \cite{FrozenPipes, BuciumasScrimshaw-grothendieck}, and Iwahori Whittaker functions \cite{BBBG-Iwahori, BBBG-metahori, BumpNaprienko-bosonic}, each with respective interesting connections to quantum affine algebra and superalgebra modules. Most relevant to the present work, Perk and Schultz \cite{PerkSchultz-1, PerkSchultz-2, Schultz-thesis} found a large class of solvable $n$-color lattice models, which were later generalized by Perk and Au-Yang \cite{Perk-Au-Yang} to include more spectral parameters. The resulting class of solutions is very large; in fact, while some of these solutions are ice type, others are $n$-color generalizations of the eight-vertex model. It is an open question to determine the extent of the overlap between their solutions and ours. Let us briefly explain why this is so, drawing contrasts between our approach and others along the way.

One approach to solvable lattice models, which could be called \emph{generative}, is to find a joint parametrization of sets of $R,S$, and $T$ weights such that these weights together satisfy the Yang-Baxter equation. This approach was taken by Perk and Schultz, and has the advantage of exhibiting concrete solutions (for them, parametrized as families of hyperbolic trigonometric functions). By contrast, our approach is much closer to that of Baxter in the six-vertex model, and could be deemed \emph{prescriptive}. We view the $S$ and $T$ weights as fixed and prove a precise criterion for when there exists a set of $R$ weights that makes the model solvable. When this happens, we give a formula for the $R$-weights in terms of the $S$ and $T$ weights. One of the main advantages to our approach is that we can determine solvability and the resulting $R$-weights without any advance knowledge of what the $R$ weights might look like. It is a common occurrence that one has weights $S$ and $T$ resulting in partition functions that {\em conjecturally} match certain special functions, and such conjectures typically follow if the model is shown to be solvable. If one has a given set of prospective $R$ weights, it is a straightforward calculation to check whether or not the Yang-Baxter equation is satisfied. It is much harder to determine whether there exists any set of $R$ weights that makes the model solvable, especially if we vary over the number of colors $n$. Our prescriptive solution here allows one to simply check the conditions (often uniformly in $n$) and generate the corresponding $R$ matrix solution.

There are some other differences between our work and that of Perk and Schultz. They assume that the vertex-dependence of a set Boltzmann weights is given by a single parameter (often called a \emph{spectral parameter}). This is a well-motivated assumption in the (then-unknown) context of quantum group modules; however, we require no such restriction. (Perk and Au-Yang \cite{Perk-Au-Yang} have multiple spectral parameters; however, different choices of their parameters may sometimes give the same weights). A more minor difference is that Perk and Schultz assume cylindrical boundary conditions and find families of weights that cause the transfer matrices to commute. We have no such restriction on the boundary conditions; solvability under our definition leads to commutation relations between transfer matrices for any boundary conditions. Finally, the $n$-color Yang-Baxter equation is a set of polynomial equations (see Proposition \ref{genEnum}); our solutions are manifestly Laurent polynomials in the Boltzmann weights (which can be normalized to produce polynomials), while the Perk-Schultz solutions involve the aforementioned transcendental functions. As such, our solutions may be better suited to algebraic combinatorics, and perhaps even commutative algebra.

When a solution to the Yang-Baxter equation arises from a quantum group module, the associated Boltzmann weights may be deformed according to a certain constrained procedure known as \emph{Drinfeld-Reshetikhin twisting} which preserves the solvability of the weights (see Section~\ref{sec:qgroups-solutions} for details). This transformation, originally defined by Drinfeld, preserves the algebra structure of the quantum group, but modifies the coalgebra structure. Reshetikhin \cite[\S~2]{Reshetikhin-twist} found a class of explicit examples where the quasitriangular Hopf algebra structure is preserved, and so new $R$-matrices are produced. When applied to the standard $U_q(\mathfrak{gl}_n)$ $R$-matrix, this twist has a nice combinatorial description in terms of the Boltzmann weights. We show (Corollary~\ref{twist-corollary}) that in fact a similar transformation holds for \emph{all} solutions to the nondegenerate $n$-color ice-type model.

In this way, we may partition the set of solutions into families up to twisting, and suggest that each such class may have a natural algebraic origin through a quantum group module or related object. It would be interesting to explore various constructions for building ``quantum objects'' from solvable lattice models (e.g., using the Yang-Baxter algebra) or from solutions to the Yang-Baxter equation using the FRT construction \cite{FRT-construction}, though we don't pursue these in the present work.

\begin{figure}[h]
\begin{center}
\scalebox{0.8}{
\begin{tikzpicture}
  \coordinate (ab) at (1,0);
  \coordinate (ad) at (3,0);
  \coordinate (af) at (5,0);
  \coordinate (ah) at (7,0);
  \coordinate (ba) at (0,1);
  \coordinate (bc) at (2,1);
  \coordinate (be) at (4,1);
  \coordinate (bg) at (6,1);
  \coordinate (bi) at (8,1);
  \coordinate (cb) at (1,2);
  \coordinate (cd) at (3,2);
  \coordinate (cf) at (5,2);
  \coordinate (ch) at (7,2);
  \coordinate (da) at (0,3);
  \coordinate (dc) at (2,3);
  \coordinate (de) at (4,3);
  \coordinate (dg) at (6,3);
  \coordinate (di) at (8,3);
  \coordinate (eb) at (1,4);
  \coordinate (ed) at (3,4);
  \coordinate (ef) at (5,4);
  \coordinate (eh) at (7,4);
  \coordinate (fa) at (0,5);
  \coordinate (fc) at (2,5);
  \coordinate (fe) at (4,5);
  \coordinate (fg) at (6,5);
  \coordinate (fi) at (8,5);
  \coordinate (gb) at (1,6);
  \coordinate (gd) at (3,6);
  \coordinate (gf) at (5,6);
  \coordinate (gh) at (7,6);
  \coordinate (bb) at (1,1);
  \coordinate (bd) at (3,1);
  \coordinate (bf) at (5,1);
  \coordinate (bh) at (7,1);
  \coordinate (db) at (1,3);
  \coordinate (dd) at (3,3);
  \coordinate (df) at (5,3);
  \coordinate (dh) at (7,3);
  \coordinate (fb) at (1,5);
  \coordinate (fd) at (3,5);
  \coordinate (ff) at (5,5);
  \coordinate (fh) at (7,5);
  \coordinate (bax) at (0,1.5);
  \coordinate (bcx) at (2,1.5);
  \coordinate (bex) at (4,1.5);
  \coordinate (bgx) at (6,1.5);
  \coordinate (bix) at (8,1.5);
  \coordinate (dax) at (0,3.5);
  \coordinate (dcx) at (2,3.5);
  \coordinate (dex) at (4,3.5);
  \coordinate (dgx) at (6,3.5);
  \coordinate (dix) at (8,3.5);
  \coordinate (fax) at (0,5.5);
  \coordinate (fcx) at (2,5.5);
  \coordinate (fex) at (4,5.5);
  \coordinate (fgx) at (6,5.5);
  \coordinate (fix) at (8,5.5);
  \draw (ab)--(gb);
  \draw (ad)--(gd);
  \draw (af)--(gf);
  \draw (ah)--(gh);
  \draw (ba)--(bi);
  \draw (da)--(di);
  \draw (fa)--(fi);
  \draw[line width=0.5mm,red] (fa)--(fi);
  \draw[line width=0.5mm,blue] (gb)--(db)--(dd)--(bd)--(bi);
  \draw[line width=0.5mm,blue] (gf)--(af);
  \draw[line width=0.5mm,green] (ba)--(bd)--(ad);
  \draw[line width=0.5mm,green] (gd)--(dd)--(dh)--(ah);
  \draw[line width=0.5mm,violet] (da)--(db)--(ab);
  \draw[line width=0.5mm,violet] (gh)--(dh)--(di);
  \draw[line width=0.5mm,violet,fill=white] (ab) circle (.25);
  \draw[line width=0.5mm,green,fill=white] (ad) circle (.25);
  \draw[line width=0.5mm,blue,fill=white] (af) circle (.25);
  \draw[line width=0.5mm,green,fill=white] (ah) circle (.25);
  \draw[line width=0.5mm,green,fill=white] (ba) circle (.25);
  \draw[line width=0.5mm,green,fill=white] (bc) circle (.25);
  \draw[line width=0.5mm,blue,fill=white] (be) circle (.25);
  \draw[line width=0.5mm,blue,fill=white] (bg) circle (.25);
  \draw[line width=0.5mm,blue,fill=white] (bi) circle (.25);
  \draw[line width=0.5mm,violet,fill=white] (cb) circle (.25);
  \draw[line width=0.5mm,blue,fill=white] (cd) circle (.25);
  \draw[line width=0.5mm,blue,fill=white] (cf) circle (.25);
  \draw[line width=0.5mm,green,fill=white] (ch) circle (.25);
  \draw[line width=0.5mm,violet,fill=white] (da) circle (.25);
  \draw[line width=0.5mm,blue,fill=white] (dc) circle (.25);
  \draw[line width=0.5mm,green,fill=white] (de) circle (.25);
  \draw[line width=0.5mm,green,fill=white] (dg) circle (.25);
  \draw[line width=0.5mm,violet,fill=white] (di) circle (.25);
  \draw[line width=0.5mm,blue,fill=white] (eb) circle (.25);
  \draw[line width=0.5mm,green,fill=white] (ed) circle (.25);
  \draw[line width=0.5mm,blue,fill=white] (ef) circle (.25);
  \draw[line width=0.5mm,violet,fill=white] (eh) circle (.25);
  \draw[line width=0.5mm,red,fill=white] (fa) circle (.25);
  \draw[line width=0.5mm,red,fill=white] (fc) circle (.25);
  \draw[line width=0.5mm,red,fill=white] (fe) circle (.25);
  \draw[line width=0.5mm,red,fill=white] (fg) circle (.25);
  \draw[line width=0.5mm,red,fill=white] (fi) circle (.25);
  \draw[line width=0.5mm,blue,fill=white] (gb) circle (.25);
  \draw[line width=0.5mm,green,fill=white] (gd) circle (.25);
  \draw[line width=0.5mm,blue,fill=white] (gf) circle (.25);
  \draw[line width=0.5mm,violet,fill=white] (gh) circle (.25);
  \node at (gb) {$1$};
  \node at (gd) {$2$};
  \node at (gf) {$1$};
  \node at (gh) {$3$};
  \node at (fa) {$0$};
  \node at (fc) {$0$};
  \node at (fe) {$0$};
  \node at (fg) {$0$};
  \node at (fi) {$0$};
  \node at (eb) {$1$};
  \node at (ed) {$2$};
  \node at (ef) {$1$};
  \node at (eh) {$3$};
  \node at (da) {$3$};
  \node at (dc) {$1$};
  \node at (de) {$2$};
  \node at (dg) {$2$};
  \node at (di) {$3$};
  \node at (cb) {$3$};
  \node at (cd) {$1$};
  \node at (cf) {$1$};
  \node at (ch) {$2$};
  \node at (ba) {$2$};
  \node at (bc) {$2$};
  \node at (be) {$1$};
  \node at (bg) {$1$};
  \node at (bi) {$1$};
  \node at (ab) {$3$};
  \node at (ad) {$2$};
  \node at (af) {$1$};
  \node at (ah) {$2$};
\end{tikzpicture}}
\end{center}
\caption{An example state for the $n$-color ice-type lattice model. Here, we have four labels, 0, 1, 2, and 3. Paths of each color propagate down and to the right.}
\label{example-color-state}
\end{figure}

We conclude the introduction with an outline of the remaining sections. Section \ref{background-section} introduces the ice-type $n$-color lattice model and the Yang-Baxter equation, and outlines the process for obtaining a set of polynomials from the Yang-Baxter equation. The Yang-Baxter equation is equivalent to the vanishing of these polynomials, which we call \emph{Yang-Baxter polynomials}. Section \ref{enumeration-section} then enumerates these polynomials.

Sections \ref{3-label-section} and \ref{n-label-section} together contain the proof of the Main Theorem. Section \ref{3-label-section} treats the $n=2$ and $n=3$ cases, in turn. Section \ref{n-label-section} obtains necessary and sufficient conditions for an ice-type $n$-color lattice model to be solvable (Theorem \ref{nlabelCond}, as well as a parametrization of the solutions when they exist (Theorem \ref{nlabelSol}. The proof relies on considering 3-label subsystems, lattice models obtained from the general case by allowing only vertices with a particular size-3 subset of the edge labels. The solvability conditions for general $n$ turn out to be equivalent to the union of the solvability conditions for all of the 3-label subsystems. Hence, the $n=3$ case is paramount in our solution of the general model. Section \ref{n-label-section} also considers the case where $\Delta_{ij} = 0$ and then analyzes when all $R$-weights are nonzero.

Finally, Section \ref{sec:qgroups-solutions} explores transformations of the Boltzmann weights that preserve solvability. After a discussion of Drinfeld-Reshetikhin twisting, we prove Corollary \ref{twist-corollary} and give a second transformation that also preserves solvability. This transformation does not yet have an algebraic interpretation, but we suspect that it may be related to a change of basis.

\textbf{Acknowledgements:} We would like to thank Daniel Bump for helpful conversations and Jacques Perk for providing numerous useful references and context related to his work on the $n$-color ice model. Much of the work for this paper was done as part of the 2022 Polymath Jr.~program. We would like to thank the tireless and committed organizers, as well as the entire community who took part in the program. This research was partially supported by NSF awards DMS-2101392 (Brubaker), DMS-1937241 (Hardt, as part of an RTG grant), and DMS-2218374 (Polymath Jr.).

\section{Background} \label{background-section}

\subsection{Lattice models}

We work with ``ice-type'' lattice models, which are finite grids of intersecting lines. The points where grid lines cross are called \emph{vertices}, and the line segments connecting any two vertices are called \emph{interior edges}. We use the term \emph{boundary edge} for the half-edge connecting a vertex to the outside of the grid.

Fix a positive integer $n$. All edges (both interior and boundary), can be decorated with one of $n$ labels $c_0,\ldots,c_{n-1}$, also called \emph{spins} or \emph{colors}. For brevity, we will usually refer to each color simply by its index $0,\ldots,n-1$. In some texts it is customary to distinguish no color or the uncolor with its own label; here we consider it to be same as any other color or label. 

The power of lattice models comes from considering the ``local'' assignments of spins around each individual vertex, and using this data to build up ``global'' statistics of the lattice model. Every vertex $\mathfrak{v}$, also called a \emph{rectangular vertex}, is surrounded by four edges, on the North, West, South, and East sides of the vertex. We call the edges to the North and West of the vertex \emph{incoming} edges, and the edges to the South and East of the vertex \emph{outgoing edges}.

We impose the following \emph{generalized ice rule} on the colors of these edges: \begin{align*}&\hspace{20pt}\textit{The number of incoming edges of color $i$ at $\mathfrak{v}$} \\&\textit{ must equal the number of outgoing edges of color $i$ at $\mathfrak{v}$.}\end{align*}

More concretely, the edges surrounding each rectangular vertex must match one of the configurations in Figure \ref{rect-vertices}. When the edges surrounding $\mathfrak{v}$ match one of these configurations, we call $\mathfrak{v}$ an \emph{admissible vertex}.

\begin{rmk}
The reason for the term \emph{generalized ice rule} is the connection between these lattice models and the classical \emph{six-vertex model} or \emph{ice model}. In the case $n=2$, let the label $c_0$ be written as $+$ and the label $c_1$ be written as $-$. Then Figure \ref{rect-vertices} consists of six admissible vertices, which match those of the classical six-vertex model (see \cite{Baxter-book}).

Further interpret each vertex as an oxygen atom. A $+$ spin on an incoming edge and a $-$ spin on an outgoing edge correspond to a hydrogen atom closely bonded to the oxygen atom, while the other spins correspond to a hydrogen atom weakly bonded to the oxygen atom. Then, the six admissible vertices are precisely the six ways of choosing two of the four hydrogen atoms to be closely bonded to the oxygen atom. The similarities between this set-up and the structure of square ice motivate the terms \emph{ice model} and \emph{ice rule}, which our rule generalizes.
\end{rmk}

\begin{figure}[h]
\centering
\scalebox{1.05}{$
\begin{array}{c@{\hspace{10pt}}c@{\hspace{10pt}}c}
\toprule
\vx{a_i} & \vx{b_{ij}} & \vx{c_{ij}} \\
\midrule
\begin{tikzpicture}
    \recvert{$i$}{$i$}{$i$}{$i$}
\end{tikzpicture}
&
\begin{tikzpicture}
    \recvert{$i$}{$j$}{$i$}{$j$}
\end{tikzpicture}
&
\begin{tikzpicture}
    \recvert{$i$}{$j$}{$j$}{$i$}
\end{tikzpicture}
\\

   \bottomrule
\end{array}$}
\caption{The admissible vertices for the $n$-color ice-type lattice model. Here, $i,j\in [0,n)$.}
    \label{rect-vertices}
\end{figure}

Next, we assign each configuration in Figure \ref{rect-vertices} a \emph{Boltzmann weight}. These Boltzmann weights can be elements of $\mathbb{C}$, or functions in some number of indeterminates. Boltzmann weights can depend on the position of a vertex in its lattice model, so in the most general setting, every vertex has its own set of Boltzmann weights. For our purpose, we only need two sets of weights (plus an additional set defined in the next subsection which we treat slightly differently): the \emph{$S$-weights}, denoted $a_i(S), b_{ij}(S), c_{ij}(S)$, and the \emph{$T$-weights}, denoted $a_i(T), b_{ij}(T), c_{ij}(T)$. Both $a_i(S)$ and $a_i(T)$ are Boltzmann weights for a vertex $\mathfrak{v}$ of configuration $\vx{a_i}$, but the former will be used when $\mathfrak{v}$ is associated to the $S$-weights, and the latter will be used when $\mathfrak{v}$ is associated to the $T$-weights.

Let the spins on the boundary edges of $\mathfrak{S}$ be fixed, and consider an assignment of spins to the interior edges. When every vertex in a lattice model $\mathfrak{S}$ is admissible, we say that the resulting global configuration $\mathfrak{s}$ is an \emph{admissible state}. The Boltzmann weight of an admissible state is simply the product of the Boltzmann weights of each of its vertices. Finally, define the \emph{partition function} $Z(\mathfrak{S})$ to be the sum of the Boltzmann weights of all admissible states: \[Z(\mathfrak{S}) := \sum_{\text{state } \mathfrak{s}} \wt(\mathfrak{s}) = \sum_{\text{state } \mathfrak{s}} \prod_{\text{vertex } \mathfrak{v}} \wt(\mathfrak{v}).\]

\subsection{The Yang-Baxter equation}

In this subsection, we review the Yang-Baxter equation, which is the main focus of this paper. Solutions to the Yang-Baxter equation are highly prized in the study of lattice models and other integrable systems. To describe the Yang-Baxter equation, we introduce a third set of Boltzmann weights, called the \emph{$R$-weights}. Instead of the rectangular vertices from Figure \ref{rect-vertices}, these weights are associated to the \emph{$R$-vertices} vertices displayed in Figure \ref{R-vertices}. The only different between the two is that the $R$-vertices are rotated $45^\circ$ counterclockwise. In fact it is possible to treat the $R,S$, and $T$ weights all on the same footing, but as our goal is to solve for the $R$-weights in terms of the $S$- and $T$-weights, the slightly asymmetrical rendering is apt. The $R$-vertices and $R$-weights use capital letters, and since there is no possibility of confusion, we will usually leave off the $R$ from the notation when writing specific $R$-weights (e.g. $A_i$ instead of $A_i(R)$).

\begin{figure}[h]
\centering
\scalebox{1.05}{$
\begin{array}{c@{\hspace{10pt}}c@{\hspace{10pt}}c}
\toprule
A_i & B_{ij} & C_{ij} \\
\midrule
\begin{tikzpicture}[scale=0.7]
    \diavert{$i$}{$i$}{$i$}{$i$}
\end{tikzpicture}
&
\begin{tikzpicture}[scale=0.7]
    \diavert{$i$}{$j$}{$i$}{$j$}
\end{tikzpicture}
&
\begin{tikzpicture}[scale=0.7]
    \diavert{$i$}{$j$}{$j$}{$i$}
\end{tikzpicture}
\\
   \bottomrule
\end{array}$}
\caption{The $R$-vertices for the $n$-color ice-type lattice model. Here, $i,j\in [1,n]$.}
    \label{R-vertices}
\end{figure}

The Yang-Baxter equation is the following equality of partition functions:

\begin{equation} \label{YBE-diagram}
    Z\left(
        \begin{tikzpicture}[scale=0.9,baseline=0.7cm]
            \lybd{E_1}{E_2}{E_3}{F_1}{F_2}{F_3}{}{}{}
        \end{tikzpicture}\right)
        \quad=\quad Z\left(\begin{tikzpicture}[scale=0.9,baseline=0.7cm]
            \rybd{E_1}{E_2}{E_3}{F_1}{F_2}{F_3}{}{}{}
        \end{tikzpicture}\right),
\end{equation}

for every choice of boundary conditions $E_1, E_2, E_3, F_1, F_2, F_3 \in [0,n)$. We call the diagram on the left of this equation a \emph{left Yang-Baxter diagram}, and the diagram on the right of this equation a \emph{right Yang-Baxter diagram}. Each choice of boundary conditions $E_1, E_2, E_3, F_1, F_2, F_3$ produces an equation describing a relationship between the Boltzmann weights of $R$, $S$, and $T$. This equation is equivalent to the vanishing of a polynomial that we call a \emph{Yang-Baxter polynomial}. Many Yang-Baxter polynomials are trivial, either because both Yang-Baxter diagrams have no admissible states, or because their partition functions are manifestly equal. However, certain boundary conditions lead to nontrivial relations. We enumerate the resulting polynomials in the next section.

\begin{eg}
We construct the Yang-Baxter polynomial $X^{010}_{001}$, defined in the next section. This polynomials corresponds to the following equality of partition functions:

\begin{equation*}
    Z\left(
        \begin{tikzpicture}[scale=0.85,baseline=0.7cm]
            \lybd{0}{1}{0}{0}{0}{1}{}{}{}
        \end{tikzpicture}\right)
        \quad=\quad Z\left(\begin{tikzpicture}[scale=0.85,baseline=0.7cm]
            \rybd{0}{1}{0}{0}{0}{1}{}{}{}
        \end{tikzpicture}\right),
\end{equation*}
The left hand Yang-Baxter diagram has two admissible states, whereas the right only has one admissible state. Expanding both partition functions as sums over their states:
\begin{equation*}
    \text{wt}\left(
        \begin{tikzpicture}[scale=0.625,baseline=0.7cm]
            \lybd{0}{1}{0}{0}{0}{1}{1}{1}{0}
        \end{tikzpicture}\right)
        \quad+\quad 
        \text{wt}\left(\begin{tikzpicture}[scale=0.625,baseline=0.7cm]
            \lybd{0}{1}{0}{0}{0}{1}{0}{0}{1}
        \end{tikzpicture}\right)
        \quad=\quad \text{wt}\left(\begin{tikzpicture}[scale=0.625,baseline=0.7cm]
            \rybd{0}{1}{0}{0}{0}{1}{1}{0}{0}
        \end{tikzpicture}\right),
\end{equation*}

and writing each state weight as a product of its vertex weights gives \[B_{01} a_{0}(S) c_{10}(T) + C_{01} b_{01}(T) c_{10}(S) = A_{0} b_{01}(S) c_{10}(T).\] Moving every term onto one side of the equation, the vanishing of the Yang-Baxter polynomial \[X^{010}_{001} = B_{01} a_{0}(S) c_{10}(T) + C_{01} b_{01}(T) c_{10}(S) - A_{0} b_{01}(S) c_{10}(T)\] is equivalent to the partition function equality with these boundary conditions.
\end{eg}

Notice that the polynomial $X^{010}_{001}$ in the previous example enjoys a nice structure, shared by all Yang-Baxter polynomials. It is homogeneous of degree 3; moreover, it is homogeneous of degree 1 in each set ($R,S,T$) of Boltzmann weights. Therefore, given a choice of $R$-weights causing $X^{010}_{001}$ to vanish, any scalar multiple of these weights will do the same, and setting all $R$-weights to zero always causes $X^{010}_{001}$ to vanish.

Because the terms are homogeneous of degree 1, we  also know every Yang-Baxter polynomial will have zero as a solution. As such, we are interested in the nonzero solutions.

Here is a much more trivial example.

\begin{eg}
Let $i,j\in [0,n)$ be distinct. Consider the following equality of partition functions:

\begin{equation*}
    Z\left(
        \begin{tikzpicture}[scale=0.85,baseline=0.7cm]
            \lybd{i}{i}{j}{i}{j}{j}{}{}{}
        \end{tikzpicture}\right)
        \quad=\quad Z\left(\begin{tikzpicture}[scale=0.85,baseline=0.7cm]
            \rybd{i}{j}{j}{i}{j}{j}{}{}{}
        \end{tikzpicture}\right),
\end{equation*}

Both the left and right partition functions must be zero, as the boundary conditions do not satisfy the generalized ice rule; thus the equation is trivially satisfied.
\end{eg}

This leads to the following simple result.

\begin{prop} \label{conserveCor}
Given $E_1, E_2, E_3, F_1, F_2, F_3 \in [0, n)$, the equation (\ref{YBE-diagram}) with these boundary conditions is identically zero (and thus trivially satisfied) unless $\{E_1, E_2, E_3\}$ and $\{F_1, F_2, F_3\}$ are equal as multisets.\end{prop}

\section{Enumeration of Yang-Baxter polynomials} \label{enumeration-section}

Our objective is to find conditions on the $S$ and $T$ weights such that the Yang-Baxter equation is satisfied, and then parameterize the resulting $R$-weights in terms of the $S$- and $T$-weights. Let $K$ be a field. Throughout the rest of the paper, we make the following assumption: \[\text{\emph{Every $S$ and $T$ weight is an element of $K^\times$.}}\] This nonzeroness assumption is vital for our results, as we will often need to invert the $S$ and $T$ Boltzmann weights. Therefore, we leave out some important special cases, such as the five-vertex model. However, the nonzeroness assumption only restricts us to a dense open set of the full choice of weights, and for other lattice models, it has often been true that taking a judicious limit allows one to consider to consider $S$ and $T$ vertices of weight zero. We will \emph{not} make any assumptions on the $R$-weights.

In this section, we will make explicit the equations that constitute the realization of the Yang-Baxter equation in the $n$-label ice-type model. To do so, we need to determine the number of admissible vertices in an $n$-label lattice model. As shown in Figure \ref{R-vertices}, $A$-vertices have exactly one label and $B$ and $C$-vertices have exactly two labels. So there are $n$ possible $A_i$ vertices, $2\binom{n}{2}$ possible $B_{i,j}$ vertices, and $2\binom{n}{2}$ possible $C_{i,j}$ vertices ($n(2n-1)$ in total).

Using these vertices, we can construct the left and right Yang-Baxter diagrams for each boundary condition. As described in the previous section, from these diagrams the Yang-Baxter equation induces a set of polynomials where the equation is satisfied if and only if each polynomial is identically zero. We seek to enumerate all such polynomials.
\begin{defn} \label{r_notation}
    For the boundary conditions $E_1,E_2,E_3,F_1,F_2,F_3$, we use $L^{E_1,E_2,E_3}_{F_1,F_2,F_3}$ to denote the left Yang-Baxter diagram and $R^{E_1,E_2,E_3}_{F_1,F_2,F_3}$ to denote the right Yang-Baxter diagram. The associated {\it Yang-Baxter polynomial} is: \[X^{E_1,E_2,E_3}_{F_1,F_2,F_3} := Z(L^{E_1,E_2,E_3}_{F_1,F_2,F_3}) - Z(R^{E_1,E_2,E_3}_{F_1,F_2,F_3})\]
\end{defn}

Since $n$ is finite, there are finitely many sets of boundary conditions of Yang-Baxter diagrams for a given $n$. We may define an equivalence relation on these boundary conditions. Two tuples of boundary conditions are ``permutation equivalent'' when some permutation of the label set $[0,n)$ takes one to the other. Since each tuple of boundary conditions corresponds to a Yang-Baxter polynomial, the equivalence on Yang-Baxter diagrams gives a natural equivalence relation on the set of polynomials. Explicitly,
\begin{defn} \label{YBPequiv} We say two equations are \textit{permutation equivalent} and write
    $X^{E_1,E_2,E_3}_{F_1,F_2,F_3} \equiv X^{e_1,e_2,e_3}_{f_1,f_2,f_3}$ if there exists some permutation $\sigma\in S_n$ acting on the set of labels $[0,n)$ such that $E_i=\sigma(e_i)$ and $F_i=\sigma(f_i)$ for $i=1,2,3$.
\end{defn}

Based on the above notation, two such polynomials are in the same equivalence class if and only if they are identical up to relabeling. In this sense, we may generate all our polynomials by such relabelings on a set of representatives of equivalence classes. In is \emph{not} the case that permutation equivalent polynomials are equal, but their structure is similar.

By Proposition \ref{conserveCor}, each label $E_i$ in the entrance set must equal some label $F_j$ in the exit set. This means that any set of boundary conditions with four or more distinct labels is trivial. Thus, all potentially nontrivial Yang-Baxter polynomials may be enumerated (up to permutation equivalence) by the cases where there are 1, 2, or 3 distinct labels in the corresponding boundary conditions. We do so in the \ref{Appendix} and see that there are $\binom{3}{0}^2+\binom{3}{1}^2+\binom{3}{3}3! = 16$ equivalence classes of polynomials.

\begin{eg}
Explicit Yang-Baxter diagrams and corresponding Yang-Baxter polynomials' equivalence classes can be found in the \ref{Appendix}.
\end{eg}

The enumeration in the previous example shows there is one equivalence class with only one distinct label, 9 with 2, and 6 with 3. Naively counting all these polynomials yields $1 \binom{n}{1} + 9 \binom{n}{2} + 6 \binom{n}{3}$ that need to be satisfied. Fortunately, we can throw out vacuous cases where $Z(L^{E_1,E_2,E_3}_{F_1,F_2,F_3}) = Z(R^{E_1,E_2,E_3}_{F_1,F_2,F_3})$ and $X^{E_1,E_2,E_3}_{F_1,F_2,F_3} = 0$. Further, we use our equivalence from Definition \ref{YBPequiv} to represent these by a fixed subset. This brings us to the following result:
\begin{prop} \label{genEnum}
    For the n-label lattice model, there are exactly $5n^3-8n^2+3n$ nonzero Yang-Baxter polynomials. Specifically, these polynomials are:
    $$ X^{iij}_{iji},\ X^{iij}_{jii},\ X^{iji}_{iij},\ X^{iji}_{iji},\ X^{iji}_{jii},\ X^{ijj}_{jij},\ X^{ijj}_{jji},\ X^{ijk}_{ikj},\ X^{ijk}_{jik},\ X^{ijk}_{kij},\ X^{ijk}_{jki},\ X^{ijk}_{kji} $$
    for any choice of distinct labels $i$, $j$, $k$ from $[0,n)$.
\end{prop}

\begin{proof}
    As mentioned before, we use Proposition \ref{conserveCor} to see that we need only inspect cases with 1, 2, or 3 distinct labels on the boundary. First consider the cases with 1 or 2 labels together and say that they are 0 and 1. For a set of entry and exit conditions $(e_1,e_2,e_3)$ and $(f_1,f_2,f_3)$, $(f_1,f_2,f_3)$ is a permutation of $(e_1,e_2,e_3)$ by Proposition \ref{conserveCor}. Represent the boundary conditions as a string $e_1e_2e_3f_1f_2f_3$. Then, starting with the string $s=000111$, any such string is a unique selection of three characters in $s$ where a selected character is changed from 0 to 1 or 1 to 0. There are $\binom{6}{3}=20$ ways to do so. Among the boundary conditions generated, those pairs in which swapping the places of 0 and 1 in one gives the other are permutation equivalent. These are actually the only equivalences in this set, so there are exactly 10 classes of conditions. In the \ref{Appendix}, we calculate their Yang-Baxter polynomials. The nontrivial polynomials are precisely those we have listed.

    In the case where there are exactly three labels on the boundary, we exactly describe the boundary conditions satisfying Proposition \ref{conserveCor} as those where the exit set is a permutation of the entry set. So $e_1e_2e_3f_1f_2f_3=e_1e_2e_3\sigma(e_1e_2e_3)$ for some permutation $\sigma$. On a set of size 3, there are 6 distinct permutations: again, we calculate their Yang-Baxter polynomials in the \ref{Appendix} and list those that are nontrivial.
\end{proof}

Essential to the rest of our work is the fact that, when we analyze the $n$-label case, we need only examine 3-label ``subsystems'' and the compatibilities between them. A 3-label subsystem refers some to choice of three colors and all the Yang-Baxter polynomials with those colors. For labels 0, 1, and 2, the next corollary describes the 3-label case ($n=3$). 

\begin{cor} \label{enumeration}
For $n = 3$, Proposition \ref{genEnum} gives us 72 polynomials in the 15 $R$-vertex weights that must simultaneously vanish. Up to permutation of $\{0,1,2\}$, they are:
    \begin{equation}
        \begin{aligned}
        X^{001}_{010} = Y_{1}(0,1) := & A_{0} b_{01}(S) c_{01}(T) - B_{01} a_{0}(S) c_{01}(T) - C_{10} b_{01}(T) c_{01}(S) \\
        X^{001}_{100} = Y_2(0,1) := & A_{0} a_{0}(T) c_{01}(S) - B_{10} b_{01}(T) c_{01}(S) - C_{01} a_{0}(S) c_{01}(T) \\
        X^{010}_{001} = Y_3(0,1) := & B_{01} a_{0}(S) c_{10}(T) + C_{01} b_{01}(T) c_{10}(S) - A_{0} b_{01}(S) c_{10}(T) \\
        X^{010}_{010} = Y_4(0,1) := & C_{01} c_{01}(T) c_{10}(S) - C_{10} c_{01}(S) c_{10}(T) \\
        X^{010}_{100} = Y_5(0,1) := & C_{01} a_{0}(T) b_{10}(S) - B_{10} c_{01}(S) c_{10}(T) - C_{01} a_{0}(S) b_{10}(T) \\
        X^{011}_{101} = Y_6(0,1) := & B_{01} c_{01}(S) c_{10}(T) + C_{01} a_{1}(S) b_{01}(T) - C_{01} a_{1}(T) b_{01}(S) \\
        X^{011}_{110} = Y_7(0,1) := & B_{01} b_{10}(T) c_{01}(S) + C_{01} a_{1}(S) c_{01}(T) - A_{1} a_{1}(T) c_{01}(S) \\
        X^{012}_{021} = Y_8(0,1,2) := & B_{01} b_{02}(S) c_{12}(T) - B_{02} b_{01}(S) c_{12}(T) \\
        X^{012}_{102} = Y_9(0,1,2) := & C_{01} b_{12}(S) b_{02}(T) - C_{01} b_{02}(S) b_{12}(T) \\
        X^{012}_{201} = Y_{10}(0,1,2) := & C_{01} c_{12}(S) b_{01}(T) + B_{01} c_{02}(S) c_{10}(T) - C_{02} b_{01}(S) c_{12}(T) \\
        X^{012}_{120} = Y_{11}(0,1,2) := & C_{01} b_{12}(S) c_{02}(T) - C_{21} c_{02}(S) b_{12}(T) - B_{12} c_{01}(S) c_{12}(T) \\
        X^{012}_{210} = Y_{12}(0,1,2) := & C_{01} c_{12}(S) c_{01}(T) + B_{01} c_{02}(S) b_{10}(T) \\&\hspace{85pt}- C_{12} c_{01}(S) c_{12}(T) - B_{21} c_{02}(S) b_{12}(T)
        \end{aligned} \label{Yang-Baxter-polynomials}
    \end{equation}
\end{cor}

We will use the notation $Y_m(i,j)$ and $Y_m(i,j,k)$ more generally to represent the polynomials (\ref{Yang-Baxter-polynomials}) with $0, 1$, and $2$ replaced by $i,j$, and $k$ respectively. We avoid confusion by adopting the convention that any $Y_m(i,j), Y_m(i,j,k)$, or indeed any quantity, with repeated colors is taken to be zero.

\section{Solving the 3-label Yang-Baxter equation} \label{3-label-section}

The key to solving the general $n$-color Yang-Baxter equation is the case $n=3$. The reason for this is that the polynomials $Y_m, 1\le m\le 12$ depend on at most three labels. Thus, as we show in Section \ref{n-label-section}, solutions to the Yang-Baxter equation for $n$-color ice-type lattice model systems can be expressed in terms of its 3-color subsystems.

In fact, looking at the Yang-Baxter polynomials using subsets of the label set will prove useful throughout. Define the following sets of polynomials: \[P_{ij} := \{Y_m(i,j)\;|\; 1\le m\le 7\}, \hspace{20pt} P_{\{i,j\}} := P_{ij}\cup P_{ji} \hspace{20pt} i,j\in [0,n),\] \[P_{ijk} := \{Y_m(i,j,k) \;|\; 8\le m\le 12\}, \hspace{20pt} P_{\{i,j,k\}} := \bigcup_{i',j',k'\in \{i,j,k\}} P_{i'j'k'}, \hspace{20pt} i,j,k\in [0,n),\] \[Q_I := \left(\bigcup_{i,j\in I} P_{ij}\right) \cup \left(\bigcup_{i,j,k\in I} P_{ijk}\right), \hspace{20pt} I\subset [0,n).\] In the sets $P_{ij}$ and $P_{ijk}$, the order of $i,j$, and $k$ matters: they consist of all relevant Yang-Baxter polynomials in \ref{Yang-Baxter-polynomials} where $0$ must be replaced with $i$, $1$ with $j$, and $2$ with $k$. By contrast, $Q_I$ consists of all Yang-Baxter polynomials with labels in $I$, so that $Q := Q_{[0,n)}$ is the full set of nonzero Yang-Baxter polynomials given by Proposition \ref{genEnum}.

\begin{eg} \label{poly-set-example}
Consider the case $n=2$. Since we only have two labels, 0 and 1, none of the three-label polynomials $Y_m$ with $m\ge 8$ will appear. We have \[P_{01} = \{Y_1(0,1), Y_2(0,1), Y_3(0,1), Y_4(0,1), Y_5(0,1), Y_6(0,1), Y_7(0,1)\}\] and \[P_{10} = \{Y_1(1,0), Y_2(1,0), Y_3(1,0), Y_4(1,0), Y_5(1,0), Y_6(1,0), Y_7(1,0)\}.\] Meanwhile, \[Q_{\{0,1\}} = P_{\{0,1\}} = P_{01} \cup P_{10}\] is the full set of all 14 polynomials that arise from the 2-color Yang-Baxter equation, and $P_{\{0,1,2\}}=\emptyset$.
\end{eg}

Let $d = n(2n-1)$ be the number of $R$-vertices in the $n$-color ice-type lattice model. Assuming the $S$ and $T$ weights to be fixed, a solution to the Yang-Baxter equation is an element of $K^d$, where $K$ is our field. Given a set $J\subset Q$ of Yang-Baxter polynomials, let \[V(J) = \{A_i,B_{ij},C_{ij}, i,j\in [0,n) \;|\; p(A_i,B_{ij},C_{ij})=0 \text{ for all } p\in J\} \subset{K^d},\] and abbreviate $V_{ij} := V(P_{ij}), V_{ijk} := V(P_{ijk}), V_I := V(Q_I)$. In general, $V(J)$ is the set of all possible sets of $R$-weights such that the equations corresponding to elements of $J$ are satisfied. Elements of $V(Q)$ are therefore solutions to the (full) $n$-color Yang-Baxter equation. Since the Yang-Baxter polynomials are homogeneous of degree 1 in the R-weights, the zero solution (where all the R-weights are zero) is always an element of $V(J)$. Thus, we will say that $V(J)$ is \emph{nonzero} if it contains any element other than the zero solution.

\begin{rmk} \label{notation-remark}
We will abuse notation in one important way: since the polynomials in $P_{ij}$ only involve labels $i$ and $j$, the solution set $V_{ij}$ will have no restrictions on any R-vertex whose labels are not both in $\{i,j\}$. Therefore, we'll often consider elements of $V_{ij}$ to be choices only of the Boltzmann weights $\{A_i,A_j, B_{ij}, B_{ji}, C_{ij}, C_{ji}\}$, and therefore elements of $K^6$, even though technically they are elements of $K^d$. The other cases, $V_{ijk}$ and $V_I$, will be treated similarly. This abuse of notation allows us to identify $V_I$ where $|I|=n'<n$ with the solutions of the $n'$-color Yang-Baxter equation.
\end{rmk}

The reader may notice that $V(J)$ is indeed the variety associated to the ideal of polynomials generated by $J$. Aside from the suggestive use of notation, we won't need to pursue this angle, although geometric interpretations of the Yang-Baxter equation could be interesting.

\subsection{The 2-label subcase}

    We begin with the 2-label case, which is the classical six-vertex model. As we saw in Example \ref{poly-set-example}, $Q_{\{0,1\}}$ is the set $\{Y_m(0,1), Y_m(1,0) \;|\; 1\le m\le 7\}$, so the polynomials $Y_8,\ldots,Y_{12}$ do not affect this case. Brubaker, Bump and Friedberg \cite[Theorem~1]{BBF-Schur-polynomials} have given a combinatorial solution to the Yang-Baxter equation here, using ideas from Baxter \cite{Baxter-book}. However, our proof of this result will also be a first step towards the 3-color case, and also gives slightly refined information about which conditions follow from which pieces of the Yang-Baxter equation.
    
    Recall Baxter's invariant $\Delta_{ij}$ associated to a set of six-vertex weights which determines conditions for the solvability of the two-color model:
	$$ \Delta_{ij}(x) = \frac{a_{i}(x)a_j(x) + b_{ij}(x)b_{ji}(x) - c_{ij}(x)c_{ji}(x)}{a_{i}(x)b_{ij}(x)}$$
	for $x$ being either $S$ or $T$. The following results give a condition involving $\Delta_{01}$ for solution to the Yang-Baxter equation.
	
	\begin{prop} \label{deltaProp}
        The solution set $V_{01}$ is nonzero if and only if $\Delta_{01}(S) = \Delta_{01}(T)$. When this holds, the solution is unique up to scalar multiple (i.e. one-dimensional).
	\end{prop}
    
    This proposition involves the vanishing of seven Yang-Baxter polynomials: \begin{equation} Y_1(0,1) = Y_2(0,1) = Y_3(0,1) = Y_4(0,1) = Y_5(0,1) = Y_6(0,1) = Y_7(0,1) = 0. \label{two-color-necessary-equation}\end{equation} We can simplify these equations using the following quantities:
	
	\begin{equation} \tau_{ij} := \frac{c_{ij}(T)c_{ji}(S)}{c_{ij}(S)c_{ji}(T)}, \hspace{30pt} \beta_{ij} := \frac{a_j(T)b_{ij}(S)-a_j(S)b_{ij}(T)}{c_{ij}(S)c_{ji}(T)}, \hspace{30pt} i,j\in [0,n). \label{tau-beta-def}\end{equation}
	
	Note that $\tau_{ij}$ and $\beta_{ij}$ are always defined, since all $S$ and $T$-weights are nonzero. As we will see, these quantities come up frequently throughout our proof of the $n$-label case.
	
	\begin{proof}
    Using (\ref{tau-beta-def}), we can simplify the equations in (\ref{two-color-necessary-equation}). In particular, $Y_4(0,1)$, $Y_5(0,1)$, and $Y_6(0,1)$ may be rewritten as follows:
    \begin{align*}
        Y_4(0,1) &= c_{01}(S)c_{10}(T)\cdot\left(\frac{c_{01}(T)c_{10}(S)}{c_{01}(S)c_{10}(T)}C_{01} - C_{10}\right) = c_{01}(S)c_{10}(T)\cdot(\tau_{01}C_{01} - C_{10}) \\
        Y_5(0,1) &= c_{01}(S)c_{10}(T)\cdot\left(\frac{a_{0}(T)b_{10}(S) - a_{0}(S)b_{10}(T)}{c_{01}(S)c_{10}(T)}C_{01} - B_{10}\right) \\
        &= c_{01}(S)c_{10}(T)\cdot(\beta_{10}\tau_{01}C_{01} - B_{10}) \\
        Y_6(0,1) &= c_{01}(S)c_{10}(T)\cdot\left(B_{01} - \frac{ a_{1}(T)b_{01}(S) - a_{1}(S)c_{01}(T)}{c_{01}(S)c_{10}(T)}C_{01} \right) \\
        &= c_{01}(S)c_{10}(T)\cdot(B_{01} - \beta_{01}C_{01})
    \end{align*}
    
    Notice that they are all 0 if and only if $C_{10}=\tau_{01}C_{01}$, $B_{10}=\beta_{10}\tau_{01}C_{01}$ and $B_{01}=\beta_{01}C_{01}$. We can thus express $B_{01}, B_{10}$, and $C_{10}$ in terms of $C_{01}$. We can use this observation to obtain, from (\ref{two-color-necessary-equation}):
	\begin{align*}
		0=Y_1(0,1) &= A_{0}b_{01}(S)c_{01}(T) - C_{01}(\beta_{01}a_{0}(S)c_{01}(T) + \tau_{01}b_{01}(T)c_{01}(S)) \\
		0=Y_2(0,1) &= A_{0}a_{0}(T)c_{01}(S) - C_{01}(\beta_{10}\tau_{01}b_{01}(T)c_{01}(S) + a_{0}(S)c_{01}(T)) \\
		0=Y_3(0,1) &= A_{0}b_{01}(S)c_{10}(T) - C_{01}(\beta_{01}a_{0}(S)c_{10}(T) + b_{01}(T)c_{10}(S)) \\
		0=Y_7(0,1) &= A_{1}a_{1}(T)c_{01}(S) - C_{01}(\beta_{01}b_{10}(T)c_{01}(S) + a_{1}(S)c_{01}(T))
	\end{align*}
 
	Each of these equations expressed precisely one other R-weight in terms of $C_{01}$. Thus, given any choice of $C_{01} \ne 0$, we have uniquely determined the values of the remaining R-weights, assuming that the first three equations are consistent with each other. This requires that:
	\begin{align}
		A_0
		&= C_{01}(\beta_{01}a_{0}(S)c_{01}(T) + \tau_{01}b_{01}(T)c_{01}(S))\ /\ (b_{01}(S)c_{01}(T)) \label{a0-first-eqn} \\
		&= C_{01}(\beta_{10}\tau_{01}b_{01}(T)c_{01}(S) + a_{0}(S)c_{01}(T))\ /\ (a_{0}(T)c_{01}(S)) \label{a0-second-eqn} \\
		&= C_{01}(\beta_{01}a_{0}(S)c_{10}(T) + b_{01}(T)c_{10}(S))\ /\ (b_{01}(S)c_{10}(T)) \label{a0-third-eqn}
	\end{align}
 
	It remains to check that these three expressions are equal if and only if $\Delta_{01}(S) = \Delta_{01}(T)$. The expressions (\ref{a0-first-eqn}) and (\ref{a0-third-eqn}) are equal by the definition of $\tau_{01}$, so we need only consider the equality of (\ref{a0-second-eqn}) and (\ref{a0-third-eqn}). Using the definitions (\ref{tau-beta-def}), this equality becomes: \begin{align*}&\left(\frac{a_0(T)b_{10}(S) - a_0(S)b_{10}(T)}{c_{01}(S)c_{10}(T)} b_{01}(T)c_{01}(S) + a_{0}(S)c_{01}(T)\right) \cdot b_{01}(S)c_{10}(T) \\&= \left(\frac{a_1(T)b_{01}(S) - a_1(S)b_{01}(T)}{c_{01}(S)c_{10}(T)} a_{0}(S)c_{10}(T) + b_{01}(T)c_{10}(S)\right) \cdot a_{0}(T)c_{01}(S),\end{align*} which after simplification is precisely the condition $\Delta_{01}(S)=\Delta_{01}(T)$.
	\end{proof}

    \begin{cor}\cite[Theorem~1]{BBF-Schur-polynomials} \label{deltaTranspose}
        The solution set $V_{\{0,1\}}$ for the 2-color Yang-Baxter equation is nontrivial if and only if $\Delta_{01}(S) = \Delta_{01}(T)$ and $\Delta_{10}(S)=\Delta_{10}(T)$. When these conditions hold, $V_{\{0,1\}} = V_{01} = V_{10}$ is one-dimensional.
    \end{cor}
    
    \begin{proof}
        As in Example \ref{poly-set-example}, $Q_{\{0,1\}} = P_{01}\cup P_{10}$, so $V_{\{0,1\}}=V_{01}\cap V_{10}$. By using Proposition \ref{deltaProp} twice, once where we swap the roles of 0 and 1, we see that these equations can only have a nonzero solution if $\Delta_{01}(S) = \Delta_{01}(T)$ and $\Delta_{10}(S) = \Delta_{10}(T)$. Fixing $C_{10}$, the solution given by each use of Proposition \ref{deltaProp} is unique, so to check that the intersection $V_{01}\cap V_{10}$ is nonzero we must show that these two solutions are identical. 

        It is straightforward to check, using the definitions of $\tau_{01}$ and $\tau_{10}$, that the values of $B_{01}$, $B_{10}$, and $C_{10}$ are consistent between the solutions in $V_{01}$ and $V_{10}$. Checking $A_0$ and $A_1$ are slightly more complicated. For $A_0$, we have three equivalent expressions (\ref{a0-first-eqn}, \ref{a0-second-eqn}, \ref{a0-third-eqn}) arising from $V_{01}$, and from $V_{10}$, we have a new expression for $A_0$ arising from the equation $Y_7(1,0)=0$: \begin{align} A_0 = (C_{10}\beta_{10}b_{01}(T)c_{10}(S) + C_{10}a_{0}(S)c_{10}(T))\ /\ (a_0(T)c_{10}(S)). \label{a0-fourth-equation}\end{align} Substituting the equation $C_{10} = \tau_{01}C_{01}$, we have 
        \begin{align*} A_0 &= C_{01}\tau_{01}\cdot(\beta_{10}b_{01}(T)c_{10}(S) + a_{0}(S)c_{10}(T))\ /\ (a_0(T)c_{10}(S)) \\&= C_{01}(\beta_{10}\tau_{01}b_{01}(T)c_{01}(S) + a_{0}(S)c_{01}(T))\ /\ (a_{0}(T)c_{01}(S)),\end{align*} and the right side matches the right side of (\ref{a0-second-eqn}).
        
        Therefore, $A_0$ has the same value in both solutions, and by a similar argument, $A_1$ does too. Thus, when $\Delta_{01}(S) = \Delta_{01}(T)$ and $\Delta_{10}(S) = \Delta_{10}(T)$, $V_{\{0,1\}}$ is nontrivial. In particular, since both $V_{01}$ and $V_{10}$ are one-dimensional, we have $V_{\{0,1\}} = V_{01} = V_{10}$.
    \end{proof}
    
    This observations in the previous proofs give an immediate simplification for the full $n$-color lattice model. 

    \begin{cor} \label{reduceCor}
        For general $n$, if for all distinct labels $i,j\in [0,n)$, $\Delta_{ij}(S)=\Delta_{ij}(T)$ and $Y_3(i,j)=0$, then also $Y_1(i,j)=Y_2(i,j)=Y_7(i,j)=0$ for all $i,j$.
    \end{cor}

    \begin{proof}
        As we have shown in the proofs of Proposition \ref{deltaProp} and Corollary \ref{deltaTranspose}, for a given $i\in [0,n)$ the conditions $\Delta_{ij}(S) = \Delta_{ji}(T)$ for any $j\in [0,n)$ ensure that the expressions on the right sides of (\ref{a0-first-eqn}, \ref{a0-second-eqn}, \ref{a0-third-eqn}, \ref{a0-fourth-equation}), replacing $0$ with $i$ and $1$ with $j$, are all equivalent. These four expressions are the values of $A_i$ when $Y_1(i,j) = 0$, $Y_2(i,j) = 0$, $Y_3(i,j) = 0$, and $Y_7(j,i) = 0$, respectively, so the fact that $A_i$ must equal either all of them or none of them means either all four of $Y_1(i,j), Y_2(i,j), Y_3(i,j)$, and $Y_7(j,i)$ are zero, or all are nonzero.
    \end{proof}

    Moving forward, Corollary \ref{reduceCor} allows us to ignore $Y_1$, $Y_2$, and $Y_7$, since $Y_3=0$ ensures they are too.

\subsection{Solving the 3-label case}

    Broadening our view to the case $n=3$, there are three different 2-color subsystems, with solution sets $V_{\{0,1\}}$, $V_{\{0,2\}}$, and $V_{\{1,2\}}$ computed by Corollary \ref{deltaTranspose}. Immediately, we have \[V_{\{0,1,2\}} = V_{\{0,1\}}\cap V_{\{0,2\}}\cap V_{\{1,2\}}\cap V(P_{\{1,2,3\}})\] (see Remark \ref{notation-remark}). So that the three $V_{\{i,j\}}$ are nontrivial, we assume henceforth that $\Delta_{ij}(S) = \Delta_{ij}(T)$ for all $i,j\in [0,3)$.
    
    In this subsection, we will focus on necessary conditions for a solution to the 3-color Yang-Baxter equation, putting off work on sufficiency until the general $n$-color result. Based on our work in the previous subsection, we introduce another quantity:
    \begin{align}\alpha_{ij} := \frac{\beta_{ij}a_{i}(S)c_{ji}(T) + b_{ij}(T)c_{ji}(S)}{b_{ij}(S)c_{ji}(T)} \label{alpha-def}\end{align}

    Note that $Y_3(i,j)=0$ is equivalent to the equation $A_i = \alpha_{ij}C_{ij}$. Using Corollary \ref{deltaTranspose} allows for a simple parameterization of the 2-color solution $V_{\{i,j\}}$ in the parameter $C_{ij}$:
    \begin{equation} \begin{split} V_{\{i,j\}} = &\{A_i, A_j, B_{ij}, B_{ji}, C_{ij}, C_{ji} \\&\;|\; A_{i}=\alpha_{ij}C_{ij}, A_{j}=\alpha_{ji}\tau_{ij}C_{ij}, B_{ij}=\beta_{ij}C_{ij}, B_{ji}=\beta_{ji}\tau_{ij}C_{ij}, C_{ji}=\tau_{ij}C_{ij}\}. \end{split} \label{two-color-parametrization}\end{equation} Note that we also have the analogous parametrization in terms of $C_{ji}$. 

    Therefore, we may relate a pair of two-color solutions by comparing their parametrizations in terms of the same $C_{ij}$. For this, we define the quantity \begin{align}\gamma_{ij} := \frac{b_{ij}(T)c_{ji}(S)}{b_{ij}(S)c_{ji}(T)} = \alpha_{ij}-\frac{a_{i}(S)}{b_{ij}(S)}\beta_{ij} \label{gamma-def}\end{align}

    \begin{lem} \label{gammaLemma}
        In any solution to the 3-color Yang-Baxter equation, \begin{align}\gamma_{ij}C_{ij}=\gamma_{ik}C_{ik} \hspace{20pt} \text{for all} \hspace{20pt} i,j,k\in [0,3). \label{c-weights-transitive-equation}\end{align}
    \end{lem}
    
    \begin{proof}
        By (\ref{two-color-parametrization}), we have $A_{i}=\alpha_{ij}C_{ij}$ and $A_{i}=\alpha_{ik}C_{ik}$, so $\alpha_{ij}C_{ij}=\alpha_{ik}C_{ik}$. Substituting (\ref{gamma-def}) gives \[ \gamma_{ij}C_{ij} + \frac{a_{i}(S)}{b_{ij}(S)}\beta_{ij}C_{ij} = \gamma_{ik}C_{ik} + \frac{a_{i}(S)}{b_{ik}(S)}\beta_{ik}C_{ik}\] and noting again the parametrization (\ref{two-color-parametrization}), \begin{align}\gamma_{ij}C_{ij} + \frac{a_{i}(S)}{b_{ij}(S)}B_{ij} = \gamma_{ik}C_{ik} + \frac{a_{i}(S)}{b_{ik}(S)}B_{ik}. \label{three-color-lemma-rearrange-equation}\end{align}

        The vanishing of $Y_8(i,j,k)$, defined in (\ref{Yang-Baxter-polynomials}), gives a relation between $B_{ij}$ and $B_{ik}$: \[B_{ij}b_{ik}(S)c_{jk}(T) = B_{ik}b_{ij}(S)c_{jk}(T).\] Using this fact along with (\ref{three-color-lemma-rearrange-equation}) gives \[\gamma_{ij}C_{ij}-\gamma_{ik}C_{ik} = a_{i}(S)\left(\frac{1}{b_{ik}(S)}B_{ik} - \frac{1}{b_{ij}(S)}B_{ij} \right) = 0,\] and the lemma follows.
    \end{proof}

    At this point, we have established the following necessary relations between the R-weights $C_{ij}$ in any solution:
    \begin{enumerate}
        \item $\gamma_{ij}C_{ij}=\gamma_{ik}C_{ik}$, for all $i,j,k\in [0,3)$, by (\ref{c-weights-transitive-equation})
        \item $C_{ji} = \tau_{ij} C_{ij}$, for $i < j$, by (\ref{two-color-parametrization}).
    \end{enumerate}
    The compatibility of these relations can be realized as the commutativity of a diagram in the following sense. Consider a graph consisting of \( \left\{ C_{ij} \vert i,j \in [1,n] \right\} \), and draw an arrow \( C_{ij} \xrightarrow{x} C_{i' j'} \) if there is a relation \( C_{i'j'} = x C_{ij} \) for some quantity $x$. Compatibility is realized as the commutativity of cycles in the associated graph.
    
    If the commutativity of such a diagram is ensured, we can produce a parametrization of the $C_{ij}$ in terms of one distinguished $C_{i_0 j_0}$ with the following procedure:
    
    \begin{itemize}
        \item Find a spanning tree of the undirected graph.
        \item for any $C_{ij}$, there is a unique undirected path between $C_{ij}$ and $C_{i_0 j_0}$ in the spanning tree. Each step in the graph gives a ratio between the corresponding Boltzmann weights, and so one can use these ratios to parameterize $C_{ij}$ in terms of $C_{i_0 j_0}$.
    \end{itemize}
    
    By considering the spanning tree in Figure \ref{spanning-tree-figure}, we can now establish the general solution to the 3-label case, $V_{\{012\}}$. For convenience, we set $\tau_{ii}=\gamma_{ii}=1$ for any $i\in [0,3)$.

    As an example, from the unique path from \(C_{01} \to C_{21}\) given by \(C_{01} \to C_{02} \to C_{20} \to C_{21} \), we have the relationship \( C_{21} = \frac{C_{21}}{C_{20}} \frac{C_{20}}{C_{02}} \frac{C_{02}}{C_{01}} C_{01} = \frac{\gamma_{20}}{\gamma_{21}} \tau_{02} \frac{\gamma_{01}}{\gamma_{02}} C_{01} = \frac{\tau_{02} \gamma_{01} \gamma_{20}}{\gamma_{21} \gamma_{02}} C_{01} \).

    \begin{figure}[h]
    \begin{center}
    \begin{tikzcd}
	{C_{20}} & {C_{02}} && {C_{12}} & {C_{21}} \\
	&& {C_{01}} \\
	&& {C_{10}}
	\arrow[color={rgb,255:red,92;green,92;blue,214}, from=2-3, to=3-3, "\tau_{01}", swap]
	\arrow[color={rgb,255:red,92;green,92;blue,214}, from=1-4, to=1-5, "\tau_{12}", swap]
	\arrow[from=2-3, to=1-2, "\frac{\gamma_{01}}{\gamma_{02}}"]
	\arrow[from=3-3, to=1-4, dotted, "\frac{\gamma_{10}}{\gamma_{12}}", swap]
	\arrow[color={rgb,255:red,92;green,92;blue,214}, from=1-2, to=1-1, "\tau_{02}"]
	\arrow[bend left = 20, from=1-1, to=1-5, "\frac{\gamma_{20}}{\gamma_{21}}"]
    \end{tikzcd}
    \end{center}
    \caption{A directed graph representing the relationships between the $C_{ij}$ weights in solutions to the 3-color Yang-Baxter equation. We remove the dashed edge to create a spanning tree, and the resulting figure allows us to obtain a parametrization of $V_{\{0,1,2\}}$ in terms of $C_{01}$.}
    \label{spanning-tree-figure}
    \end{figure}

    \begin{lem} \label{3labelSol}
        Any nonzero solution to the 3-color Yang-Baxter equation, if one exists, is unique up to scalar multiple. This solution can be parametrized by $C_{01}$, as follows:
        \[V_{\{0,1,2\}} \subset \left\{ A_i,B_{ij},C_{ij}, i,j\in [0,3) \;\left|\; A_{i} = \alpha_{ij}C_{ij},\ B_{ij} = \beta_{ij}C_{ij},\ C_{ij} = \frac{\tau_{0i}\gamma_{01}\gamma_{i0}}{\gamma_{ij}\gamma_{0i}}\cdot C_{01} \right.\right\}\]
    \end{lem}

    \begin{proof}
        The expressions for $A_{i}$ and $B_{ij}$ arise from the parametrization (\ref{two-color-parametrization}) of each of the 2-label subcases. To show that these potential solutions are well-defined, we must make sure that every way to express a given vertex weight as a multiple of $C_{01}$ is equivalent. This is trivial for the weights $B_{ij}$ and $C_{ij}$, but for $A_i$ to be well-defined we must have
        \[
            \alpha_{ij} C_{ij} = \alpha_{ik} C_{ik} \hspace{20pt} \text{for all } i,j,k\in[0.3).
        \]
        
        By (\ref{gamma-def}) and the vanishing of $Y_6(i,j)$ and $Y_8(i,j,k)$ from~(\ref{Yang-Baxter-polynomials}), we have $\alpha_{ij}=\gamma_{ij}+\frac{a_i(S)}{b_{ij}(S)}\beta_{ij}$, $B_{ij}=\beta_{ij}C_{ij}$, and $\frac{1}{b_{ij}(S)}B_{ij} = \frac{1}{b_{ik}(S)}B_{ik}$. Combining the latter two reveals that $\frac{1}{b_{ij}(S)}\beta_{ij}C_{ij} = \frac{1}{b_{ik}(S)}\beta_{ik}C_{ik}$. Indeed, by Lemma \ref{gammaLemma}, we find that: $$\alpha_{ij}C_{ij} = \gamma_{ij}C_{ij} + \frac{a_i(S)}{b_{ij}(S)}\beta_{ij}C_{ij} = \gamma_{ik}C_{ik} + \frac{a_i(S)}{b_{ik}(S)}\beta_{ik}C_{ik} = \alpha_{ik}C_{ik}$$
        Finally, (\ref{two-color-parametrization}) and (\ref{c-weights-transitive-equation}) require that in any solution $C_{ji} = \tau_{ij}C_{ij}$ and $\gamma_{ij}C_{ij} = \gamma_{ik}C_{ik}$. This implies that $C_{0i} = \frac{\gamma_{01}}{\gamma_{0i}}\cdot C_{01}$, so $C_{i0} = \frac{\tau_{0i}\gamma_{01}}{\gamma_{0i}}\cdot C_{01}$, and therefore $C_{ij} = \frac{\tau_{0i}\gamma_{01}\gamma_{i0}}{\gamma_{ij}\gamma_{0i}}\cdot C_{01}$, as desired. Since all weights are determined uniquely given a fixed $C_{01}$, such a solution is unique up to scalar multiple.
    \end{proof}

    The existence of an explicit parametrization allows us to simply substitute the $R$-weights in Lemma \ref{3labelSol} into the Yang-Baxter polynomials (\ref{Yang-Baxter-polynomials}), yielding polynomials in the $S$ and $T$ weights which must vanish. By simplifying the resulting expressions, we obtain the following conditions for solvability.
    
    \begin{prop} \label{3labelCond}
        The 3-label Yang-Baxter equation has a nontrivial solution if and only if, for all $i,j,k\in [0,n)$, the following equations hold:
        $$
        \Delta_{ij}(S)=\Delta_{ij}(T)\ \hspace{3em}
        \frac{\beta_{ij}}{\gamma_{ij}b_{ij}(S)} = \frac{\beta_{ik}}{\gamma_{ik}b_{ik}(S)} \hspace{3em}
        \frac{b_{ik}(S)}{b_{ik}(T)}=\frac{b_{jk}(S)}{b_{jk}(T)}
        $$\vspace{-1.5em}
        \begin{align*}
            \gamma_{ik}c_{jk}(S)b_{ij}(T) + \beta_{ij}\gamma_{ik}c_{ik}(S)c_{ji}(T) &= \gamma_{ij}b_{ij}(S)c_{jk}(T) \\
	    	\gamma_{jk}b_{jk}(S)c_{ik}(T) &= \tau_{ij}\tau_{jk}\gamma_{ji}c_{ik}(S)b_{jk}(T) + \tau_{ij}\beta_{jk}\gamma_{ji}c_{ij}(S)c_{jk}(T) \\
	    	\gamma_{jk}c_{jk}(S)c_{ij}(T) + \beta_{ij}\gamma_{jk}c_{ik}(S)b_{ji}(T) &= \tau_{ij}\gamma_{ji}c_{ij}(S)c_{jk}(T) + \tau_{ij}\tau_{jk}\beta_{kj}\gamma_{ji}c_{ik}(S)b_{jk}(T)
        \end{align*}
        When these conditions hold, the solution is unique up to scalar multiple, and is given by Lemma \ref{3labelSol}.
    \end{prop}

    \begin{proof}
        The necessity of these conditions follows from the parametrization in Lemma \ref{3labelSol}. This calculation is lengthy, but straightforward.

        As an example, we prove that the condition \((\ast)\) is necessary. The other equations follow a similar approach. In any solution in $V_{\{0,1,2\}}$, the vanishing of the Yang-Baxter polynomial \(Y_{10}(i,j,k)\) is the equation \( C_{ij} c_{jk}(S) b_{ij}(T) + B_{ij} c_{ik}(S) c_{ji}(T) - C_{ik} b_{ij}(S) c_{jk}(T) = 0 \). Therefore, by the parametrization in Lemma \ref{3labelSol},
        \[
        \frac{\tau_{0i} \gamma_{01} \gamma_{i_0}}{\gamma_{ij} \gamma_{0i}} C_{01} c_{jk}(S) b_{ij}(T) + \beta_{ij} \frac{\tau_{0i} \gamma_{01} \gamma_{i_0}}{\gamma_{ij} \gamma_{0i}} C_{01} c_{ik}(S) c_{ji}(T) - \frac{\tau_{0i} \gamma_{01} \gamma_{i_0}}{\gamma_{ik} \gamma_{0i}} C_{01} b_{ij}(S) c_{jk}(T) = 0.
        \]
        Because the $S$ and $T$-weights are nonzero and since in a nonzero solution $C_{01}\ne 0$, we can cancel \( \frac{\gamma_{0i} \gamma_{01} \gamma_{i_0}}{\gamma_{0i}} C_{01} \) from every term, which leaves us with \( (*) \) as a necessary condition.
        
        Sufficiency can also be done as a direct check, but we postpone the proof to the next section, as it is lengthy and the $n$-color case is no different than the 3-color case.
    \end{proof}
    
\section{Solving the $n$-label Yang-Baxter equation} \label{n-label-section}

In this section, we extend the 3-color case to the $n$-color case. Since every Yang-Baxter polynomial involves at most three colors, the set of $n$-color polynomials is just the union of the polynomials for each of the 3-color subsystems. That is, \[Q = Q_{[0,n)} = \bigcup_{I\subset [0,n), |I|=3} Q_I,\] and so  \[V := V_{[0,n)} = \bigcap_{I\subset [0,n), |I|=3} V_I.\]

This leads us to consider the compatibility of different 3-color parametrization in Lemma \ref{3labelSol}. Given a fixed $C_{ij}$, in order for $V_{\{i,j,k\}}\cap V_{\{i,j,\ell\}}$ to be nonzero, the $R$-weights $A_i, A_j, B_{ij}, B_{ji}$, and $C_{ji}$ must be equal in both subsystems.  

For the $n$-color Yang-Baxter equation, as we did in the 3-color case, we can construct a directed graph where the nodes of the graph represent the Boltzmann weights $C_{ij}$ and the directed edges denote relations between them arising from (\ref{two-color-parametrization}) and (\ref{c-weights-transitive-equation}). The label on each directed edge gives us a factor to multiply by, and the directionality of the edge tells us how that factor gives a relationship between the relevant Boltzmann weights.

As in the $n=3$ case, we will look only at a spanning tree of this graph, which we use to parametrize proposed solutions to the Yang-Baxter equation. Instead of proving that all the relations in the original graph hold, we will then prove that the Yang-Baxter equation holds directly.

We draw the graph using two types of edges:
\begin{description}
\item Between vertices $C_{ij}$ and $C_{ji}$, we have $C_{ij}\xrightarrow{\tau_{ji}}C_{ji}$.
\item Between vertices $C_{ij}$ and $C_{ik}$, we have $C_{ij}\xrightarrow{\frac{\gamma_{ij}}{\gamma_{ik}}}C_{ik}$.
\end{description}
where the weight describes the factor by which the tail is multiplied to obtain the head.

To choose the edges of the spanning (out)-tree, we first root the tree at $C_{01}$. Then, to find a (unique) path to $C_{ij}$, we consider the cases $i=0$ and $i\ne0$ and proceed as in Lemma \ref{3labelSol}. The spanning tree generated by this is shown in Figure \ref{large-spanning-tree-figure}.

\begin{figure}[h]
\begin{center}
\begin{tikzcd}
	\cdots \\
	{C_{03}} \\
	{C_{02}} & {C_{30}} & {C_{31}} & {C_{32}} & {C_{34}} & \cdots \\
	{C_{01}} & {C_{20}} & {C_{21}} & {C_{23}} & {C_{24}} & \cdots \\
	& {C_{10}} & {C_{12}} & {C_{13}} & {C_{14}} & \cdots
	\arrow[from=5-2, to=5-3]
	\arrow[bend left = 9, from=5-2, to=5-4]
	\arrow[bend left = 14, from=5-2, to=5-5]
	\arrow[from=4-1, to=5-2]
	\arrow[from=4-1, to=3-1]
	\arrow[bend left = 20, from=4-1, to=2-1]
	\arrow[from=3-1, to=4-2]
	\arrow[from=4-2, to=4-3]
	\arrow[bend left = 9, from=4-2, to=4-4]
	\arrow[bend left = 14, from=4-2, to=4-5]
	\arrow[from=2-1, to=3-2]
	\arrow[from=3-2, to=3-3]
    \arrow[from=4-1, to=1-1, bend left = 24]
    \arrow[from=2-1, to=1-1, dotted, no head]
    \arrow[from=3-1, to=1-1, bend right = 18, dotted, no head]
    \arrow[from=3-1, to=2-1, dotted, no head]
    \arrow[from=3-3, to=3-5, bend right = 14, dotted, no head]
    \arrow[from=4-3, to=4-5, bend right = 14, dotted, no head]
    \arrow[from=5-3, to=5-5, bend right = 14, dotted, no head]
    \arrow[from=4-3, to=4-4, dotted, no head]
    \arrow[from=4-4, to=4-5, dotted, no head]
	\arrow[bend left = 9, from=3-2, to=3-4]
	\arrow[bend left = 14, from=3-2, to=3-5]
	\arrow[dotted, no head, from=3-3, to=3-4]
	\arrow[dotted, no head, from=3-4, to=3-5]
	\arrow[dotted, no head, from=5-3, to=5-4]
	\arrow[dotted, no head, from=5-4, to=5-5]
	\arrow[bend left = 18, from=3-2, to=3-6]
	\arrow[bend left = 18, from=4-2, to=4-6]
	\arrow[bend left = 18, from=5-2, to=5-6]
\end{tikzcd}
\end{center}
\caption{A spanning out-tree of the directed graph representing the relationships between the $C_{ij}$ weights in solutions to the $n$-color Yang-Baxter equation. The dashed edges represent some of the edges of the graph that don't appear in the spanning tree., and the resulting figure allows us to obtain a parametrization of $V_{[0,n)}$ in terms of $C_{01}$.}
\label{large-spanning-tree-figure}
\end{figure}

\begin{eg}
As an example, we use Figure \ref{large-spanning-tree-figure} to compute $C_{31}$ in terms of $C_{01}$. Using the edges of the spanning tree, the unique path from $C_{01}$ to $C_{31}$ is $C_{01}\to C_{03}\to C_{30}\to C_{31}$, and this gives the relationship
    \[C_{31} = \frac{C_{31}}{C_{30}} \frac{C_{30}}{C_{03}} \frac{C_{03}}{C_{01}} C_{01}
    = \frac{\gamma_{30}}{\gamma_{31}} \tau_{03} \frac{\gamma_{01}}{\gamma_{03}} C_{01},\] so $C_{31} = \frac{\gamma_{30} \tau_{03}}{\gamma_{31} \gamma_{03}} \gamma_{01} C_{01}$.
\end{eg}

This approach was our tool for obtaining a parametrization, and it shows conceptually how 3-color subsystems fit together in the $n$-color Yang-Baxter equation. However, analyzing the graph is not strictly necessary to prove the Yang-Baxter equation holds. Instead of supplementing this approach with more details, we'll give an explicit, uniform parametrization directly.

\subsection{Solving the $n$-label case\label{sec:n-label}}

    Similarly to the 3-color case of the previous section, we can use the spanning tree in Figure \ref{large-spanning-tree-figure}, to obtain a parametrization of the (potential) solutions to the $n$-color Yang-Baxter equation in terms of the $S$ and $T$ weights and $C_{01}$. However, there is nothing special about the weight $C_{01}$, so by choosing a particularly nice weight for $C_{01}$, we instead give a particular solution that treats all colors equally and is unique up to scalar. Each $R$-vertex weight depends only on $S$ and $T$-vertex weights in its own labels and some arbitrary additional label uniform across all weights.
    
    \begin{thm} \label{nlabelSol}
        For any tuple of distinct labels $i$, $j$, $k$, any solution to the $n$-label Yang-Baxter equation can be written as a scalar multiple of the $d$-tuple with components:
        \begin{equation}
            A_{i}=\alpha_{ij}\frac{\gamma_{ik}\tau_{ki}}{\gamma_{ij}\gamma_{ki}}, \hspace{2em}
            B_{ij}=\beta_{ij}\frac{\gamma_{ik}\tau_{ki}}{\gamma_{ij}\gamma_{ki}}, \hspace{2em}
            C_{ij}=\frac{\gamma_{ik}\tau_{ki}}{\gamma_{ij}\gamma_{ki}}. \label{n-color-parametrization}
        \end{equation}
    \end{thm}
    
    \begin{proof}
        By the parametrization of the 3-color case (Lemma \ref{3labelSol}), we must have $C_{ij} = \frac{\gamma_{01}\gamma_{i0}\tau_{0i}}{\gamma_{0i}\gamma_{ij}}\cdot C_{01}$ for all $i,j\in [0,n)$. Replacing 0 and 1 by $k$ and $l$ gives $C_{ij} = \frac{\gamma_{kl}\gamma_{ik}\tau_{ki}}{\gamma_{ki}\gamma_{ij}}\cdot C_{kl}$. Instead replacing $i$ and $j$ in the original expression with $k$ and $l$ gives $C_{kl} = \frac{\gamma_{01}\gamma_{k0}\tau_{0k}}{\gamma_{0k}\gamma_{kl}}\cdot C_{01}$. We therefore have two expressions for $C_{ij}/C_{kl}$, and by setting these equal we obtain:\[ \frac{\gamma_{kl}\gamma_{ik}\tau_{ki}}{\gamma_{ki}\gamma_{ij}} = \left(\frac{\gamma_{01}\gamma_{i0}\tau_{0i}}{\gamma_{0i}\gamma_{ij}}\right)\cdot \left(\frac{\gamma_{01}\gamma_{k0}\tau_{0k}}{\gamma_{0k}\gamma_{kl}}\right)^{-1}= \frac{\gamma_{kl}}{\gamma_{ij}} \frac{\gamma_{i0}\tau_{0i}\gamma_{0k}}{\gamma_{k0}\tau_{0k}\gamma_{0i}},\] so $\frac{\gamma_{ik}\tau_{ki}}{\gamma_{ki}} = \frac{\gamma_{i0}\tau_{0i}\gamma_{0k}}{\gamma_{k0}\tau_{0k}\gamma_{0i}}$ and $\frac{\gamma_{i0}\tau_{0i}}{\gamma_{0i}} = \frac{\gamma_{ik}\tau_{ki}}{\gamma_{ki}} \frac{\gamma_{k0}\tau_{0k}}{\gamma_{0k}}$. This means that \begin{align} C_{ij} = \frac{\gamma_{01}}{\gamma_{ij}} \frac{\gamma_{ik}\tau_{ki}}{\gamma_{ki}} \frac{\gamma_{k0}\tau_{0k}}{\gamma_{0k}}\cdot C_{01} \label{Cij-param-equationn}\end{align} for any $k$. Set $C_{01}=\frac{\gamma_{0k}\tau_{k0}}{\gamma_{01}\gamma_{k0}}$; then by (\ref{Cij-param-equationn}),  $C_{ij}=\frac{\gamma_{ik}\tau_{ki}}{\gamma_{ij}\gamma_{ki}}$. Restricting to the 2-color subsystems, (\ref{two-color-parametrization}) shows that \[A_{i} = \alpha_{ij}C_{ij} = \alpha_{ij}\frac{\gamma_{ik}\tau_{ki}}{\gamma_{ij}\gamma_{ki}} \hspace{20pt} \text{and} \hspace{20pt} B_{ij} = \beta_{ij}C_{ij} = \beta_{ij}\frac{\gamma_{ik}\tau_{ki}}{\gamma_{ij}\gamma_{ki}},\] as desired. Furthermore, this solution is unique up to the choice of $C_{01}$, and since the system of equations is homogeneous, each solution is a scalar multiple of this one.
    \end{proof}
    
    For ease of use, we now restate Theorem \ref{nlabelSol} directly in terms of the $S$ and $T$ Boltzmann weights.

    \begin{cor} \label{rectExp}
        For any tuple of distinct labels $i$, $j$, $k$, any solution to the $n$-label Yang-Baxter equation can be written as a scalar multiple of the $d$-tuple with components:
        \begin{center}
            $A_{i}=\displaystyle \left(1 + \frac{a_i(S)a_j(T)b_{ij}(S) - a_i(S)a_j(S)b_{ij}(T)}{b_{ij}(T)c_{ij}(S)c_{ji}(S)}\right)\cdot \frac{b_{ik}(T)b_{ki}(S)}{b_{ik}(S)b_{ki}(T)}$
            
            $B_{ij}=\displaystyle \left(\frac{a_j(T)b_{ij}(S)-a_j(S)b_{ij}(T)}{c_{ij}(S)c_{ji}(S)}\right)\cdot \frac{b_{ij}(S)b_{ik}(T)b_{ki}(S)}{b_{ij}(T)b_{ik}(S)b_{ki}(T)}$ \hspace{2em}
            $C_{ij}=\displaystyle \frac{b_{ij}(S)b_{ik}(T)b_{ki}(S)c_{ji}(T)}{b_{ij}(T)b_{ik}(S)b_{ki}(T)c_{ji}(S)}$
        \end{center}
    \end{cor}

    \begin{proof}
        We repeat the definitions of the quantities $\alpha_{ij}$, $\beta_{ij}$, $\tau_{ij}$, and $\gamma_{ij}$:
        \begin{align*}
            \alpha_{ij} := \frac{\beta_{ij}a_{i}(S)c_{ji}(T) + b_{ij}(T)c_{ji}(S)}{b_{ij}(S)c_{ji}(T)} &\hspace{2em} \beta_{ij} := \frac{a_j(T)b_{ij}(S)-a_j(S)b_{ij}(T)}{c_{ij}(S)c_{ji}(T)} \\
            \tau_{ij} := \frac{c_{ij}(T)c_{ji}(S)}{c_{ij}(S)c_{ji}(T)} &\hspace{2em} \gamma_{ij} := \frac{b_{ij}(T)c_{ji}(S)}{b_{ij}(S)c_{ji}(T)}
        \end{align*}
        Substituting into the solution from Theorem \ref{nlabelSol} produces the desired expressions.
    \end{proof}

    At this point, the parametrization (\ref{n-color-parametrization}) is necessary but not sufficient, in that we don't actually know that (\ref{n-color-parametrization}) gives a bona fide solution. However, having the parametrization in hand will be enough for us to precisely establish the necessary and sufficient conditions for solvability. 
    
    \begin{thm} \label{nlabelCond}
        The $n$-label Yang-Baxter equation has a nontrivial solution if and only if, for each distinct tuple of labels $i,j,k$,
        \begin{equation} \label{n-color-conditions}
        \begin{aligned}
        \Delta_{ij}(S)=\Delta_{ij}(T), \ \hspace{3em}
        &\frac{\beta_{ij}}{\gamma_{ij}b_{ij}(S)} = \frac{\beta_{ik}}{\gamma_{ik}b_{ik}(S)}, \hspace{3em}
        \frac{b_{ik}(S)}{b_{ik}(T)}=\frac{b_{jk}(S)}{b_{jk}(T)},
            \\\gamma_{ik}c_{jk}(S)b_{ij}(T) + \beta_{ij}\gamma_{ik}c_{ik}(S)c_{ji}(T) &= \gamma_{ij}b_{ij}(S)c_{jk}(T), \\
	    	\gamma_{jk}b_{jk}(S)c_{ik}(T) &= \tau_{ij}\tau_{jk}\gamma_{ji}c_{ik}(S)b_{jk}(T) + \tau_{ij}\beta_{jk}\gamma_{ji}c_{ij}(S)c_{jk}(T), \\
	    	\gamma_{jk}c_{jk}(S)c_{ij}(T) + \beta_{ij}\gamma_{jk}c_{ik}(S)b_{ji}(T) &= \tau_{ij}\gamma_{ji}c_{ij}(S)c_{jk}(T) + \tau_{ij}\tau_{jk}\beta_{kj}\gamma_{ji}c_{ik}(S)b_{jk}(T).
        \end{aligned}
        \end{equation}
    \end{thm}
    
    Perhaps surprisingly, these conditions are exactly the same conditions that appeared in the 3-color case. Nominally, they ensure that all 3-color subsystems have a nonzero solution i.e. that $V_{\{i,j,k\}}$ has a nonzero element for all $i,j,k\in [0,n)$. But a priori it is not obvious that the solutions coincide, making the intersections $V_{\{i,j,k\}}\cap V_{\{i',j',k'\}}$ nonzero.

    \begin{proof}
        By Proposition \ref{3labelCond}, these are all necessary, since each condition is necessary in at least one three-label subsystem. We prove sufficiency by using the parametrization (\ref{n-color-parametrization}) in Theorem \ref{nlabelSol} and showing the all the Yang-Baxter polynomials vanish. By Corollary \ref{reduceCor}, we need only consider $Y_3(i,j)$ through $Y_6(i,j)$ and $Y_8(i,j,k)$ through $Y_{12}(i,j,k)$. As in Proposition \ref{deltaProp}, the use of $\alpha$, $\beta$, and $\tau$ permits some simplification. We rewrite the remaining necessary and sufficient equations from (\ref{Yang-Baxter-polynomials}):
        \begin{align*}
            0=Y_3(i,j) = & b_{ij}(S)c_{ji}(T)\cdot(\alpha_{ij}C_{ij} - A_{i})\\
            0=Y_4(i,j) = & c_{ij}(S)c_{ji}(T)\cdot(\tau_{ij}C_{ij} - C_{ji}) \\
            0=Y_5(i,j) = & c_{ij}(S)c_{ji}(T)\cdot(\beta_{ji}\tau_{ij}C_{ij} - B_{ji}) \\
            0=Y_6(i,j) = & c_{ij}(S)c_{ji}(T)\cdot(B_{ij} - \beta_{ij}C_{ij}) \\
            0=Y_8(i,j,k) = & B_{ij} b_{ik}(S) c_{jk}(T) - B_{ik} b_{ij}(S) c_{jk}(T) \\
            0=Y_9(i,j,k) = & C_{ij} b_{jk}(S) b_{ik}(T) - C_{ij} b_{ik}(S) b_{jk}(T) \\
            0=Y_{10}(i,j,k) = & C_{ij} c_{jk}(S) b_{ij}(T) + B_{ij} c_{ik}(S) c_{ji}(T) - C_{ik} b_{ij}(S) c_{jk}(T) \\
            0=Y_{11}(i,j,k) = & C_{ij} b_{jk}(S) c_{ik}(T) - C_{kj} c_{ik}(S) b_{jk}(T) - B_{jk} c_{ij}(S) c_{jk}(T) \\
            0=Y_{12}(i,j,k) = & C_{ij} c_{jk}(S) c_{ij}(T) + B_{ij} c_{ik}(S) b_{ji}(T) - C_{jk} c_{ij}(S) c_{jk}(T) - B_{kj} c_{ik}(S) b_{jk}(T)
        \end{align*}

        It is evident that the solution (\ref{n-color-parametrization}) causes $Y_3(i,j)$ and $Y_6(i,j)$ to vanish. To see that $Y_4(i,j)=0$, observe that $C_{ij} = \frac{b_{ij}(S)c_{ji}(T)}{b_{ij}(T)c_{ji}(S)} \frac{b_{ik}(T)b_{ki}(S)}{b_{ik}(S)b_{ki}(T)}$ by Corollary \ref{rectExp}. We calculate:
        $$\tau_{ij}C_{ij} = \frac{c_{ij}(T)c_{ji}(S)}{c_{ij}(S)c_{ji}(T)} \frac{b_{ij}(S)c_{ji}(T)}{b_{ij}(T)c_{ji}(S)} \frac{b_{ik}(T)b_{ki}(S)}{b_{ik}(S)b_{ki}(T)} = \frac{c_{ij}(T)}{c_{ij}(S)} \frac{b_{ij}(S)}{b_{ij}(T)} \frac{b_{ik}(T)b_{ji}(S)}{b_{ik}(S)b_{ji}(T)} = \frac{1}{\gamma_{ji}} \frac{b_{ij}(S)b_{ik}(T)}{b_{ij}(T)b_{ik}(S)}$$
        Transposing $i$ and $j$ in the expression for $C_{ij}$ from Corollary \ref{rectExp},
        $$C_{ji} = \frac{1}{\gamma_{ji}} \frac{b_{jk}(T)b_{kj}(S)}{b_{jk}(S)b_{kj}(T)} = \frac{1}{\gamma_{ji}} \frac{b_{ik}(T)b_{ij}(S)}{b_{ik}(S)b_{ij}(T)} = \tau_{ij}C_{ij}$$
        where we assume that $\frac{b_{ij}(S)}{b_{ij}(T)}=\frac{b_{jk}(S)}{b_{jk}(T)}$ and $\frac{b_{ik}(S)}{b_{ik}(T)}=\frac{b_{kj}(S)}{b_{kj}(T)}$. As a result, $C_{ji} = \tau_{ij}C_{ij}$ and $Y_4(i,j)$ vanishes. Further, $B_{ji}=\beta_{ji}C_{ij}=\beta_{ji}\tau_{ij}C_{ij}$ and $Y_5(i,j)=0$ as well.
    
        We must now verify the remaining equations. Recall that by Lemma \ref{gammaLemma}, $\gamma_{ij}C_{ij}=\gamma_{ik}C_{ik}$ for all $i,j,k$ in any solution. From the parametrization (\ref{n-color-parametrization}), $\gamma_{ij}C_{ij}=\gamma_{ik}\frac{\tau_{ki}}{\gamma_{ki}}=\gamma_{ik}C_{ik}$, so this equality does indeed hold. Moreover, by the definition of $B_{ij}$ in (\ref{n-color-parametrization}), the identity $B_{ij}=\beta_{ij}C_{ij}$ is true as well. We make extensive use of these in what follows. 

        First, we prove that $Y_8(i,j,k)=0$. We have that $\beta_{ij}\gamma_{ik}b_{ik}(S) = \beta_{ik}\gamma_{ij}b_{ij}(S)$. Given that $\gamma_{ij}C_{ij}=\gamma_{ik}C_{ik}$, we multiply the left side by $\gamma_{ij}C_{ij}$ and the right side by $\gamma_{ik}C_{ik}$:
        $$ \beta_{ij}\gamma_{ik} b_{ik}(S)\gamma_{ij}C_{ij} = \beta_{ik}\gamma_{ij}b_{ij}(S)\gamma_{ik}C_{ik} \text{ so } B_{ij}\gamma_{ik}\gamma_{ij}b_{ik}(S) = B_{ik}\gamma_{ij}\gamma_{ik}b_{ij}(S) $$
        Cancelling terms, we get $B_{ij}b_{ik}(S)=B_{ik}b_{ij}(S)$. From this, $Y_8(i,j,k)=0$.

        Second, we prove that $Y_9(i,j,k)=0$. We have that $\frac{b_{ik}(S)}{b_{ik}(T)}=\frac{b_{jk}(S)}{b_{jk}(T)}$. Multiplying by $C_{ij}$:
        $$ \frac{b_{ik}(S)}{b_{ik}(T)}C_{ij}=\frac{b_{jk}(S)}{b_{jk}(T)}C_{ij} \text{ so } b_{ik}(S)b_{jk}(T)C_{ij} = b_{ik}(T)b_{jk}(S)C_{ij}$$
        From this, $Y_{9}(i,j,k)=0$.

        Third, we prove that $Y_{10}(i,j,k)=0$. We have that $\gamma_{ik}c_{jk}(S)b_{ij}(T) + \beta_{ij}\gamma_{ik}c_{ik}(S)c_{ji}(T) = \gamma_{ij}b_{ij}(S)c_{jk}(T)$. Multiplying by $C_{ij}$:
        \begin{align*}
            C_{ij}\gamma_{ik}c_{jk}(S)b_{ij}(T) + C_{ij}\beta_{ij}\gamma_{ik}c_{ik}(S)c_{ji}(T) &= C_{ij}\gamma_{ij}b_{ij}(S)c_{jk}(T)\\
            C_{ij}\gamma_{ik}c_{jk}(S)b_{ij}(T) + B_{ij}\gamma_{ik}c_{ik}(S)c_{ji}(T) &= C_{ik}\gamma_{ik}b_{ij}(S)c_{jk}(T)\\
            C_{ij}c_{jk}(S)b_{ij}(T) + B_{ij}c_{ik}(S)c_{ji}(T) &= C_{ik}b_{ij}(S)c_{jk}(T)
        \end{align*}
        Again, this equality directly demonstrates that $Y_{10}(i,j,k)=0$.
        
        Persisting in a similar way, fourth, we prove that $Y_{11}(i,jk)=0$. We have $\gamma_{jk}b_{jk}(S)c_{ik}(T) = \tau_{ij}\tau_{jk}\gamma_{ji}c_{ik}(S)b_{jk}(T) + \tau_{ij}\beta_{jk}\gamma_{ji}c_{ij}(S)c_{jk}(T)$. Multiplying by $C_{jk}$:
        \begin{align*}
            C_{jk}\gamma_{jk}b_{jk}(S)c_{ik}(T) &= C_{jk}\tau_{ij}\tau_{jk}\gamma_{ji}c_{ik}(S)b_{jk}(T) + C_{jk}\tau_{ij}\beta_{jk}\gamma_{ji}c_{ij}(S)c_{jk}(T) \\
            C_{ji}\gamma_{ji}b_{jk}(S)c_{ik}(T) &= C_{kj}\tau_{ij}\gamma_{ji}c_{ik}(S)b_{jk}(T) + B_{jk}\tau_{ij}\gamma_{ji}c_{ij}(S)c_{jk}(T) \\
            C_{ij}\tau_{ij}b_{jk}(S)c_{ik}(T) &= C_{kj}\tau_{ij}c_{ik}(S)b_{jk}(T) + B_{jk}\tau_{ij}c_{ij}(S)c_{jk}(T) \\
            C_{ij}b_{jk}(S)c_{ik}(T) &= C_{kj}c_{ik}(S)b_{jk}(T) + B_{jk}c_{ij}(S)c_{jk}(T)
        \end{align*}
        We discover that this equality implies $Y_{11}(i,j,k)=0$, as before.

        Finally, we prove that $Y_{12}(i,j,k)=0$. We have that $\gamma_{jk}c_{jk}(S)c_{ij}(T) + \beta_{ij}\gamma_{jk}c_{ik}(S)b_{ji}(T) = \tau_{ij}\gamma_{ji}c_{ij}(S)c_{jk}(T) + \tau_{ij}\tau_{jk}\beta_{kj}\gamma_{ji}c_{ik}(S)b_{jk}(T)$. Multiplying by both $C_{ij}$ and $C_{kj}$:
        \begin{align*}
            &C_{ij}C_{kj}\gamma_{jk}c_{jk}(S)c_{ij}(T) + C_{ij}C_{kj}\beta_{ij}\gamma_{jk}c_{ik}(S)b_{ji}(T) \\
            &\hspace{1em}= C_{ij}C_{kj}\tau_{ij}\gamma_{ji}c_{ij}(S)c_{jk}(T) + C_{ij}C_{kj}\tau_{ij}\tau_{jk}\beta_{kj}\gamma_{ji}c_{ik}(S)b_{jk}(T) \\
            &C_{ij}C_{kj}\gamma_{jk}c_{jk}(S)c_{ij}(T) + B_{ij}C_{kj}\gamma_{jk}c_{ik}(S)b_{ji}(T) \\
            &\hspace{1em}= C_{ji}C_{kj}\gamma_{ji}c_{ij}(S)c_{jk}(T) + C_{ji}B_{kj}\tau_{jk}\gamma_{ji}c_{ik}(S)b_{jk}(T) \\
            &C_{ij}C_{kj}\gamma_{jk}c_{jk}(S)c_{ij}(T) + B_{ij}C_{kj}\gamma_{jk}c_{ik}(S)b_{ji}(T) \\
            &\hspace{1em}= C_{jk}C_{kj}\gamma_{jk}c_{ij}(S)c_{jk}(T) + C_{kj}B_{kj}\gamma_{jk}c_{ik}(S)b_{jk}(T) \\
            &C_{ij}c_{jk}(S)c_{ij}(T) + B_{ij}c_{ik}(S)b_{ji}(T) = C_{jk}c_{ij}(S)c_{jk}(T) + B_{kj}c_{ik}(S)b_{jk}(T)
        \end{align*}
        By this last equality, we have illustrated that $Y_{12}(i,j,k)=0$.
        
        Given our conditions, all of the requisite equations are satisfied. Hence, the $n$-label Yang-Baxter equation has a nontrivial solution if and only if they all hold.
    \end{proof}

\subsection{The case $\Delta_{ij}=0$}

One of the most important special cases of the six-vertex model is the \emph{free fermion point}, where Baxter's $\Delta=0$. This subsection concerns an $n$-color analogue of that condition.

For any $i,j\in [0,n)$, define \[\Delta_{ij}(R) := \frac{A_iA_j + B_{ij}B_{ji} - C_{ij}C_{ji}}{A_iB_{ij}}.\] This definition is analogous to that of $\Delta_{ij}(S)$ and $\Delta_{ij}(T)$, but since some $R$-weights may equal 0, even in a nonzero solution (see next subsection), $\Delta_{ij}(R)$ may not always be defined. Therefore, we also define \[\Delta_{ij}^{num}(R) := A_iA_j + B_{ij}B_{ji} - C_{ij}C_{ji},\] noting that $\Delta_{ij}^{num}(R) = \Delta_{ji}^{num}(R)$.

Let $S$ and $T$ be sets of $n$-color weights (which may or may not satisfy (\ref{n-color-conditions})), and let $R$ be a nonzero scalar multiple of (\ref{n-color-parametrization}). Recall that we say $R\in V(Q)$ whenever $R,S$, and $T$ satisfy the Yang-Baxter equation.

\begin{prop} \label{ff-prop}
Fix $i,j\in [0,n)$. If $\Delta_{ij}(S) = \Delta_{ij}(T) = 0$, then $\Delta_{ij}^{num}(R) = \Delta_{ji}^{num}(R) = 0$.
\end{prop}

\begin{proof}
This follows from \cite[Theorem~3]{BBF-Schur-polynomials}, noting that the six $R$-weights in (\ref{n-color-parametrization}) which have subscripts in $\{i,j\}$ are a common scalar multiple of the $R$-weights in \cite[Theorem~2]{BBF-Schur-polynomials}.
\end{proof}

It is noteworthy that Proposition \ref{ff-prop} holds whether or not $R\in V(Q)$, and for each pair of labels $i,j$ independent of other pairs.

In the six-vertex model, the free fermion point simultaneously describes:

\begin{itemize}
    \item The center of the ``disordered regime'', where the particle interactions can be considered ``maximally entropic'' \cite[pp.~151-152]{Baxter-book}.
    \item The set of six-vertex models which can be solved via determinants as in the Lindstr\"{o}m–Gessel–Viennot Lemma \cite[Proposition~2.3]{Naprienko-ff}.
    \item A class of lattice models which are pairwise solvable, and which under a natural composition law form a subgroup isomorphic to $GL_2\times GL_1$ of the Yang-Baxter groupoid \cite{KBI, BBF-Schur-polynomials, Naprienko-groupoid}.
    \item The set of six-vertex models which have a discrete-time Hamiltonian expressible as an exponential operator in a Heisenberg algebra \cite[Theorem~4.1]{Hardt-Hamiltonians}.
\end{itemize}

Analogues of each of these appear to be open questions in the generality of the $n$-color ice model, and we refrain from any conjectures. However, Proposition \ref{ff-prop} tells us that if $\Delta_{ij}(S) = \Delta_{ij}(T)$ for some $i,j\in [0,n)$ that the above properties hold when we restrict to states of the lattice model involving only $i$ and $j$. This opens up the possibility of building up $n$-color results from results on subsystems, much as we have done in this paper for solvability.

\subsection{When are all R-vertex weights nonzero?}

    It is noteworthy that, assuming the $S$ and $T$-vertex weights are nonzero, in any nontrivial solution $C_{ij}\ne 0$ for all $i$, $j$. Nonetheless, there are cases where some $R$-vertex weights can be 0. In this section, we explore some such cases. To begin, we analyze some of the conditions in Theorem \ref{nlabelCond}.

    \begin{lem} \label{condLem}
        The final three conditions in Theorem \ref{nlabelCond} are equivalent to:
        \begin{center}
            $\beta_{ij} = \left(\displaystyle\frac{\gamma_{ij}}{\gamma_{ik}} - \gamma_{kj}\right) \displaystyle\displaystyle\frac{b_{ij}(S)c_{jk}(T)}{c_{ik}(S)c_{ji}(T)}$ \hspace{2em}
            $\beta_{ij} = \left(\tau_{ik}\displaystyle\frac{\gamma_{ij}}{\gamma_{ik}} - \gamma_{kj}\displaystyle\frac{\tau_{ij}}{\tau_{kj}}\right) \frac{b_{ij}(S)c_{kj}(T)}{c_{ki}(S)c_{ij}(T)}$ \vspace{0.5em}
            
            $\beta_{ij}b_{ji}(T)\displaystyle\frac{\tau_{ji}}{\gamma_{ji}} - \beta_{kj}b_{jk}(T)\frac{\tau_{jk}}{\gamma_{jk}}  =  \left(\frac{1}{\gamma_{jk}} - \frac{\gamma_{kj}}{\gamma_{ij}\gamma_{ji}}\right) \frac{c_{ij}(S)c_{jk}(T)}{c_{ik}(S)}$
        \end{center}
    \end{lem}

    \begin{proof}
        Throughout this proof, we consistently refer to the identity $\frac{b_{ij}(T)}{b_{ij}(S)} = \frac{b_{kj}(T)}{b_{kj}(S)}$, one of the conditions in Theorem \ref{nlabelCond}. It is essential to expressing these conditions using $\gamma$ and $\tau$.
    
        The first condition from Theorem \ref{nlabelCond} that we simplify is $\gamma_{ik}c_{jk}(S)b_{ij}(T) + \beta_{ij}\gamma_{ik}c_{ik}(S)c_{ji}(T) = \gamma_{ij}b_{ij}(S)c_{jk}(T)$. We solve for $\beta_{ij}$ assuming it holds:
        \begin{align*}
            \beta_{ij}\gamma_{ik}c_{ik}(S)c_{ji}(T) &= \gamma_{ij}b_{ij}(S)c_{jk}(T) - \gamma_{ik}c_{jk}(S)b_{ij}(T) \\&= \left(\gamma_{ij} - \gamma_{ik}\frac{c_{jk}(S)b_{ij}(T)}{b_{ij}(S)c_{jk}(T)}\right)b_{ij}(S)c_{jk}(T) \\
            &= \left(\gamma_{ij} - \gamma_{ik}\frac{c_{jk}(S)b_{kj}(T)}{b_{kj}(S)c_{jk}(T)}\right)b_{ij}(S)c_{jk}(T) \\&= (\gamma_{ij} - \gamma_{ik}\gamma_{kj})b_{ij}(S)c_{jk}(T),
        \end{align*}
        assuming $\frac{b_{ij}(T)}{b_{ij}(S)} = \frac{b_{kj}(T)}{b_{kj}(S)}$. Dividing by $\gamma_{ik}c_{ik}(S)c_{ji}(T)$,
        $$\beta_{ij}=\left(\frac{\gamma_{ij}}{\gamma_{ik}} - \gamma_{kj}\right) \frac{b_{ij}(S)c_{jk}(T)}{c_{ik}(S)c_{ji}(T)}$$
        
        The second condition is $\gamma_{ij}b_{ij}(S)c_{kj}(T) = \tau_{ki}\tau_{ij}\gamma_{ik}c_{kj}(S)b_{ij}(T) + \tau_{ki}\beta_{ij}\gamma_{ik}c_{ki}(S)c_{ij}(T)$. Rearranging and factoring,
        \begin{align*}
            \beta_{ij}\gamma_{ik}c_{ki}(S)c_{ij}(T) &= \tau_{ik}\gamma_{ij}b_{ij}(S)c_{kj}(T) - \tau_{ij}\gamma_{ik}c_{kj}(S)b_{ij}(T) \\ 
            = \left(\tau_{ik}\gamma_{ij} - \tau_{ij}\gamma_{ik}\frac{b_{ij}(T)c_{kj}(S)}{b_{ij}(S)c_{kj}(T)}\right)b_{ij}(S)c_{kj}(T)
            &= \left(\tau_{ik}\gamma_{ij} - \tau_{ij}\gamma_{ik}\frac{b_{kj}(T)c_{kj}(S)}{b_{kj}(S)c_{kj}(T)}\right)b_{ij}(S)c_{kj}(T) \\
            &= (\tau_{ik}\gamma_{ij} - \tau_{ij}\gamma_{ik}\gamma_{kj}/\tau_{kj})b_{ij}(S)c_{kj}(T)
        \end{align*}
        where, as before, we assume $\frac{b_{ij}(T)}{b_{ij}(S)} = \frac{b_{kj}(T)}{b_{kj}(S)}$. Dividing by $\gamma_{ik}c_{ki}(S)c_{ij}(T)$, $$\beta_{ij} = \left(\tau_{ik}\frac{\gamma_{ij}}{\gamma_{ik}} - \gamma_{kj}\frac{\tau_{ij}}{\tau_{kj}}\right) \frac{b_{ij}(S)c_{kj}(T)}{c_{ki}(S)c_{ij}(T)}$$
        
        Last, we handle the condition $\gamma_{jk}c_{jk}(S)c_{ij}(T) + \beta_{ij}\gamma_{jk}c_{ik}(S)b_{ji}(T) = \tau_{ij}\gamma_{ji}c_{ij}(S)c_{jk}(T) + \tau_{ij}\tau_{jk}\beta_{kj}\gamma_{ji}c_{ik}(S)b_{jk}(T)$. Rearranging, we see that
        \begin{align*}
            &\beta_{ij}\gamma_{jk}c_{ik}(S)b_{ji}(T) - \tau_{ij}\tau_{jk}\beta_{kj}\gamma_{ji}c_{ik}(S)b_{jk}(T) = \tau_{ij}\gamma_{ji}c_{ij}(S)c_{jk}(T) - \gamma_{jk}c_{jk}(S)c_{ij}(T) \\
            &= \left(\tau_{ij}\gamma_{ji} - \gamma_{jk}\frac{c_{jk}(S)c_{ij}(T)}{c_{ij}(S)c_{jk}(T)}\right)c_{ij}(S)c_{jk}(T) \\
            &= \left(\tau_{ij}\gamma_{ji} - \gamma_{jk}\frac{c_{jk}(S)c_{ij}(T)b_{ij}(S)b_{ij}(T)}{c_{ij}(S)c_{jk}(T)b_{ij}(S)b_{ij}(T)}\right)c_{ij}(S)c_{jk}(T) \\
            &= \left(\tau_{ij}\gamma_{ji} - \gamma_{jk}\tau_{ij}/\gamma_{ij}\frac{c_{jk}(S)b_{kj}(T)}{c_{jk}(T)b_{kj}(S)}\right)c_{ij}(S)c_{jk}(T) = (\tau_{ij}\gamma_{ji} - \gamma_{jk}\gamma_{kj}\tau_{ij}/\gamma_{ij})c_{ij}(S)c_{jk}(T)
        \end{align*}
        Again, we used that $\frac{b_{ij}(T)}{b_{ij}(S)} = \frac{b_{kj}(T)}{b_{kj}(S)}$ to realize the presence of the factor $\gamma_{kj}$. Dividing by $\gamma_{ji}\gamma_{jk}\tau_{ij}c_{ik}(S)$ reveals
        $$\beta_{ij}b_{ji}(T)\frac{\tau_{ji}}{\gamma_{ji}} - \beta_{kj}b_{jk}(T)\frac{\tau_{jk}}{\gamma_{jk}} = \left(\frac{1}{\gamma_{jk}} - \frac{\gamma_{kj}}{\gamma_{ij}\gamma_{ji}}\right)\frac{c_{ij}(S)c_{jk}(T)}{c_{ik}(S)}$$

        Consequently, we have reformulated all three conditions as desired.
    \end{proof}

    This reformulation permits a precise characterization of the case where $B_{ij}=0$ in multiple ways. Recollect that $B_{ij}=0$ if and only if $\beta_{ij}=0$ so that we may work with $\beta_{ij}$ instead.

    It turns out that either all of the $\beta_{ij}$ are zero, or they are all nonzero:

    \begin{lem} \label{betaLem}
        $\beta_{ij}=0$ if and only if $\beta_{kl}=0$ for any labels $i$, $j$, $k$, $l$ with $i\ne j$ and $k\ne l$.
    \end{lem}
    \begin{proof}
        By Corollary \ref{rectExp}, $\beta_{ij}=0$ if and only if $\frac{a_j(T)}{a_j(S)} = \frac{b_{ij}(T)}{b_{ij}(S)}$. By Theorem \ref{nlabelCond}, $\frac{b_{ij}(T)}{b_{ij}(S)} = \frac{b_{kj}(T)}{b_{kj}(S)}$, so $\frac{a_j(T)}{a_j(S)} = \frac{b_{kj}(T)}{b_{kj}(S)}$ and $\beta_{kj}=0$. Further, by Theorem \ref{nlabelCond}, $\beta_{kj}\gamma_{kl} b_{kl}(S) = \beta_{kl}\gamma_{kj}b_{kj}(S)$ so $\beta_{kl}=0$. By symmetry, $\beta_{ij}=0$ if and only if $\beta_{kl}=0$.
    \end{proof}

    \begin{prop} \label{gamDecompProp}
        For any tuple of distinct labels $i$, $j$, $k$, the following are equivalent:
        \begin{enumerate}[(1)]
            \item $\beta_{ij}=0$
            \item $\gamma_{ij}=\gamma_{ik}\gamma_{kj}$
            \item $\displaystyle\frac{\gamma_{ij}}{\gamma_{ik}\gamma_{kj}} = \frac{\tau_{ij}}{\tau_{ik}\tau_{kj}}$
        \end{enumerate}
    \end{prop}
    \begin{proof}
        By Lemma \ref{condLem}, $\beta_{ij} = \left(\frac{\gamma_{ij}}{\gamma_{ik}} - \gamma_{kj}\right) \frac{b_{ij}(S)c_{jk}(T)}{c_{ik}(S)c_{ji}(T)}$. So $\beta_{ij}=0$ if and only if $\gamma_{ij}=\gamma_{ik}\gamma_{kj}$.

        By Lemma \ref{condLem}, $\beta_{ij} = \left(\tau_{ik}\frac{\gamma_{ij}}{\gamma_{ik}} - \gamma_{kj}\frac{\tau_{ij}}{\tau_{kj}}\right) \frac{b_{ij}(S)c_{kj}(T)}{c_{ki}(S)c_{ij}(T)}$. So $\beta_{ij}=0$ if and only if $\tau_{ik}\frac{\gamma_{ij}}{\gamma_{ik}} - \gamma_{kj}\frac{\tau_{ij}}{\tau_{kj}} = 0$ or, equivalently, $\frac{\gamma_{ij}}{\gamma_{ik}\gamma_{kj}} = \frac{\tau_{ij}}{\tau_{ik}\tau_{kj}}$.
    \end{proof}

    \begin{cor}
        If $\beta_{ij}=0$, then
        \begin{enumerate} [(1)]
            \item $\tau_{ij}=\tau_{ik}\tau_{kj}$.
            \item $\gamma_{ij}\gamma_{ji}=1$.
        \end{enumerate}
    \end{cor}

    \begin{proof}
        By Proposition \ref{gamDecompProp}, if $\beta_{ij}=0$, $\gamma_{ij}=\gamma_{ik}\gamma_{kj}$ and $\frac{\gamma_{ij}}{\gamma_{ik}\gamma_{kj}} = \frac{\tau_{ij}}{\tau_{ik}\tau_{kj}}$ so $\tau_{ij}=\tau_{ik}\tau_{kj}$.

        By Lemma \ref{betaLem}, $\beta_{ij}=0$ if and only if $\beta_{ij}=0$ and $\beta_{kj}=0$. Then by Lemma \ref{condLem}, $\beta_{ij}b_{ji}(T)\frac{\tau_{ji}}{\gamma_{ji}} - \beta_{kj}b_{jk}(T)\frac{\tau_{jk}}{\gamma_{jk}} = \left(\frac{1}{\gamma_{jk}} - \frac{\gamma_{kj}}{\gamma_{ij}\gamma_{ji}}\right)\frac{c_{ij}(S)c_{jk}(T)}{c_{ik}(S)}$ so $\beta_{ij}=0$ only if $\gamma_{ij}\gamma_{ji}=\gamma_{kj}\gamma_{jk}$. By Proposition \ref{gamDecompProp}, $\beta_{ij}=0$ is equivalent to $\gamma_{ij}=\gamma_{ik}\gamma_{kj}$, so we are able to write $\gamma_{ik}\gamma_{kj}\gamma_{jk}\gamma_{ki}=\gamma_{kj}\gamma_{jk}$ and $\gamma_{ik}\gamma_{ki}=1$. We are free to permute this by Lemma \ref{betaLem}, so $\gamma_{ij}\gamma_{ji}=1$.
    \end{proof}

    The extra conditions that arise from the 3-label case are what allow this characterization of when $B_{ij}=0$. Because $A_{i}$ appears in none of the Yang-Baxter polynomials $Y_8(i,j,k)$ -- $Y_{12}(i,j,k)$, it has no further constraints than appear in the 2-label case. In particular, from Corollary \ref{rectExp}, $A_{i}=0$ if and only if $\frac{a_{i}(S)}{b_{ij}(T)}(a_{j}(T)b_{ij}(S) - a_{j}(S)b_{ij}(T)) = c_{ij}(T)c_{ji}(S)$.
    
We conclude the current section by giving some simple examples where the $B_{ij}$ are all zero.

\begin{eg}
    Suppose that corresponding $S$ and $T$-vertex weights are equal i.e. $a_i(S)=a_i(T), b_{ij}(S) = b_{ij}(T), c_{ij}(S)=c_{ij}(T)$ for any distinct pair of labels $i\ne j$. Then the $R$-vertex weights are, up to scaling, given by
    $$\{A_{i}=1, B_{ij}=0, C_{ij}=1\ |\ i,j\in[0,n)\}$$
    An inspection of Corollary \ref{rectExp} gives these $R$-weights. Furthermore, one can see that $\beta_{ij}=0$ and $\alpha_{ij}=\gamma_{ij}=\tau_{ij}=1$, which satisfy the solvability conditions in Theorem \ref{nlabelCond}.
\end{eg}

\begin{eg}
    Let $a,b,c\in K^\times$, and consider parameters $z_i(x)\in K^\times$, $i\in [0,n)$, $x\in\{S,T\}$, satisfying $\frac{z_i(S)}{z_i(T)}=\frac{z_j(S)}{z_j(T)}$. Given weights
    $$\{a_{i}(x)=az_i(x), b_{ij}(x)=bz_i(x), c_{ij}(x)=cz_i(x)\ |\ i,j\in[0,n), x\in\{S,T\}\},$$
    the corresponding $R$-vertex weights are
    $$\{A_{i}=1, B_{ij}=0, C_{ij}=1\ |\ i,j\in[0,n)\}.$$

    To verify solvability, note that $\frac{z_j(S)z_i(T)}{z_i(S)z_j(T)}=1$. Then, $\Delta_{ij}(x) = \frac{a_{i}(x)a_j(x) + b_{ij}(x)b_{ji}(x) - c_{ij}(x)c_{ji}(x)}{a_{i}(x)b_{ij}(x)} = \frac{z_j(x)}{z_i(x)}\frac{a^2+b^2-c^2}{ab}$ does not depend on $S$ or $T$, and furthermore we compute $\beta_{ij} = \frac{ab}{c^2}(1 - \frac{z_j(S)z_i(T)}{z_i(S)z_j(T)}) = 0$, $\tau_{ij} = \gamma_{ij} = \frac{z_i(T)z_j(S)}{z_i(S)z_j(T)} = 1$. Plugging these quantities into (\ref{n-color-conditions}), we see that the conditions hold.
\end{eg}

    The $S$ and $T$ weights from the second example have a straightforward combinatorial interpretation. Setting $z_i(x)=1$ for all $i$ and $x$ gives identical $S$ and $T$ weights, so puts us back in the case of the first example. For general $z_i(x)$, one can start with identical $S$ and $T$ weights and then scale each rectangular vertex weight by a parameter, $z_i(x)$, depending on the color on the left edge of the vertex. It can be seen by analyzing the dynamics of the colored paths that is the same on any state of the lattice model with fixed boundary conditions, so factors out of the partition function. Therefore, in this very simple case, we see that such a scaling does not modify the $R$-weights. It is worth exploring to what extent this is a general phenomenon.

\section{Relationship to quantum group solutions} \label{sec:qgroups-solutions}

In this section, we show two symmetries of the solvability criterion in Theorem \ref{nlabelCond}. One of these involves modifications to the $b$-weights, and the other involves modifications to the $c$-weights. The former transformation is related to the phenomenon of Drinfeld-Reshetikhin twisting.

\subsection{Standard $R$-matrices of $U_q(\widehat{\mathfrak{gl}}_{n})$ and $U_q(\widehat{\mathfrak{gl}}_{m|n})$}

Let $\mathcal{H}$ be a Hopf algebra, with comultiplication $\Delta: \mathcal{H}\to \mathcal{H}\otimes\mathcal{H}$. $\mathcal{H}$ is called \emph{quasitriangular} if there exists an invertible element $R\in \mathcal{H}\otimes \mathcal{H}$ that satisfies \[(\Delta\otimes \text{id})R = R_{13}R_{23}, \hspace{20pt} (\text{id}\otimes\Delta)R = R_{13}R_{12}, \hspace{20pt} \tau\Delta(x) = R\Delta(x)R^{-1},\] where $R_{ij}$ refers to the embedding of $R$ in the $i,j$ tensor factors of $\mathcal{H}\otimes \mathcal{H}\otimes \mathcal{H}$ and $\tau:\mathcal{H}\otimes\mathcal{H}\to \mathcal{H}\otimes\mathcal{H}$ is the map $x\otimes y\mapsto y\otimes x$. $R$ is often called a \emph{universal $R$-matrix} for $\mathcal{H}$, and satisfies the Yang-Baxter equation \begin{align} R_{12}R_{13}R_{23} = R_{23}R_{13}R_{12}. \label{universal-YBE}\end{align} See, for instance, \cite{Kassel-quantum-groups} for more details on these objects.

This version of the Yang-Baxter equation is related to ours in the following sense. Given vector spaces $U,V,W$, define the operator $R\in\End(U\otimes V)$ by the formula \[R(u\otimes v) = \sum_{u',v'} R(u,v,u',v') u'\otimes v',\] where every vector appearing here is a basis vector. Here, $R(u,v,u',v')$ denotes the Boltzmann weight of the vertex

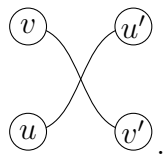
\begin{center}
\begin{tikzpicture}[scale=0.7]
\draw (0,0) to [out = 0, in = 180] (2,2);
\draw (0,2) to [out = 0, in = 180] (2,0);
\draw[fill=white] (0,0) circle (.35);
\draw[fill=white] (0,2) circle (.35);
\draw[fill=white] (2,2) circle (.35);
\draw[fill=white] (2,0) circle (.35);
\node at (0,0) {$u$};
\node at (0,2) {$v$};
\node at (2,2) {$u'$};
\node at (2,0) {$v'$};
\end{tikzpicture}.
\end{center}

Similarly, define $S\in\End(U\otimes W), T\in\End(V\otimes W)$ by \[S(u\otimes w) = \sum_{u',w'} S(u,w,u',w') u'\otimes w' \hspace{20pt} T(v\otimes w) = \sum_{v',w'} T(v,w,v',w') v'\otimes w',\] where  $S(u,w,u',w')$ and $T(v,w,v',w')$ refer to the appropriate Boltzmann weights.

In our setting, all of $U,V,W\cong\mathbb{C}^n$. The edge labels $1,\ldots,n$ for all our vertices are therefore in bijection with the standard bases of $U,V$, and $W$. With this viewpoint, the Yang-Baxter equation (\ref{YBE-diagram}) becomes equivalent to the equation $RST = TSR$ in $U\otimes V\otimes W$. Specializing $R,S$, and $T$ to $R_{12}, R_{13}$, and $R_{23}$ in (\ref{universal-YBE}) gives that version of the Yang-Baxter equation.

Now consider the case where $\mathcal{H}$ is the quantum affine superalgebra $U_q(\widehat{\mathfrak{gl}}_{m|n})$. The standard $R$-matrix\footnote{By this we mean the universal $R$-matrix of $U_q(\widehat{\mathfrak{gl}}_{m|n})$, applied to the module $V(z)\otimes V(1)$, where $V$ is the standard evaluation representation of $U_q(\widehat{\mathfrak{gl}}_{m|n})$.} is \begin{align*} R := R(z) = &\sum_{i\le m} (zq-q^{-1}) e_{ii}\otimes e_{ii} + \sum_{i>m} (q-zq^{-1}) e_{ii}\otimes e_{ii} + \sum_{i\ne j} (-1)^{\delta_{i\le m}\delta_{j\le m}} (z-1) e_{ii}\otimes e_{jj} \\&+ (q-q^{-1})\sum_{i>j} e_{ij}\otimes e_{ji} + z(q-q^{-1})\sum_{i<j} e_{ij}\otimes e_{ji},\end{align*} where all sums require $i,j\in [1,m+n]$. (Here, we are considering the ``ungraded'' Yang-Baxter equation \cite{Kojima-superalgebra}). The coefficients of this $R$-matrix give the Boltzmann weights in Figure \ref{std-R-matrix-superalgebra}.

\begin{figure}[h]
\centering
\scalebox{1.05}{$
\begin{array}{c@{\hspace{10pt}}c@{\hspace{10pt}}c}
\toprule
\vx{a_i} & \vx{b_{ij}} & \vx{c_{ij}} \\
\midrule
\begin{tikzpicture}
\coordinate (a) at (-.75, 0);
\coordinate (b) at (0, .75);
\coordinate (c) at (.75, 0);
\coordinate (d) at (0, -.75);
\coordinate (aa) at (-.75,.5);
\coordinate (cc) at (.75,.5);
\draw (a)--(0,0);
\draw (b)--(0,0);
\draw (c)--(0,0);
\draw (d)--(0,0);
\draw[fill=white] (a) circle (.25);
\draw[fill=white] (b) circle (.25);
\draw[fill=white] (c) circle (.25);
\draw[fill=white] (d) circle (.25);
\node at (0,1) { };
\node at (a) {$i$};
\node at (b) {$i$};
\node at (c) {$i$};
\node at (d) {$i$};
\end{tikzpicture}
&
\begin{tikzpicture}
\coordinate (a) at (-.75, 0);
\coordinate (b) at (0, .75);
\coordinate (c) at (.75, 0);
\coordinate (d) at (0, -.75);
\coordinate (aa) at (-.75,.5);
\coordinate (cc) at (.75,.5);
\draw (a)--(0,0);
\draw (b)--(0,0);
\draw (c)--(0,0);
\draw (d)--(0,0);
\draw[fill=white] (a) circle (.25);
\draw[fill=white] (b) circle (.25);
\draw[fill=white] (c) circle (.25);
\draw[fill=white] (d) circle (.25);
\node at (0,1) { };
\node at (a) {$i$};
\node at (b) {$j$};
\node at (c) {$i$};
\node at (d) {$j$};
\end{tikzpicture}
&
\begin{tikzpicture}
\coordinate (a) at (-.75, 0);
\coordinate (b) at (0, .75);
\coordinate (c) at (.75, 0);
\coordinate (d) at (0, -.75);
\coordinate (aa) at (-.75,.5);
\coordinate (cc) at (.75,.5);
\draw (a)--(0,0);
\draw (b)--(0,0);
\draw (c)--(0,0);
\draw (d)--(0,0);
\draw[fill=white] (a) circle (.25);
\draw[fill=white] (b) circle (.25);
\draw[fill=white] (c) circle (.25);
\draw[fill=white] (d) circle (.25);
\node at (0,1) { };
\node at (a) {$i$};
\node at (b) {$j$};
\node at (c) {$j$};
\node at (d) {$i$};
\end{tikzpicture}
\\
\midrule
\begin{array} {c@{\hspace{10pt}}c} zq-q^{-1}, & \text{if } i\le m \\ q-zq^{-1}, & \text{if } i>m \end{array} & \begin{array} {c@{\hspace{10pt}}c} 1-z, & \text{if } i,j\le m \\ z-1, & \text{otherwise} \end{array} & \begin{array} {c@{\hspace{10pt}}c} (q-q^{-1}), & \text{if } i>j \\ z(q-q^{-1}), & \text{if } i<j \end{array} \\
   \bottomrule
\end{array}$}
\caption{The admissible vertices and Boltzmann weights for the standard $U_q(\widehat{\mathfrak{gl}}_{m|n})$ $R$-matrix. Here, $i,j\in [1,m+n]$, and $z$ depends on the choice of $S$ or $T$.}
    \label{std-R-matrix-superalgebra}
\end{figure}

Note that these weights satisfy our conditions in Theorem \ref{nlabelCond}.

Setting $m=0$, we obtain the case of the quantum affine algebra $U_q(\widehat{\mathfrak{gl}}_n)$. The standard $R$-matrix is \begin{align*}R := R(z) = &\sum_i (q-zq^{-1}) e_{ii}\otimes e_{ii} + \sum_{i\ne j} (z-1) e_{ii}\otimes e_{jj} \\&+ (q-q^{-1})\sum_{i>j} e_{ij}\otimes e_{ji} + z(q-q^{-1})\sum_{i<j} e_{ij}\otimes e_{ji},\end{align*} and the corresponding Boltzmann weights are given in Figure \ref{std-R-matrix}.

\begin{figure}[h]
\centering
\scalebox{1.05}{$
\begin{array}{c@{\hspace{10pt}}c@{\hspace{10pt}}c}
\toprule
\vx{a_i} & \vx{b_{ij}} & \vx{c_{ij}} \\
\midrule
\begin{tikzpicture}
\coordinate (a) at (-.75, 0);
\coordinate (b) at (0, .75);
\coordinate (c) at (.75, 0);
\coordinate (d) at (0, -.75);
\coordinate (aa) at (-.75,.5);
\coordinate (cc) at (.75,.5);
\draw (a)--(0,0);
\draw (b)--(0,0);
\draw (c)--(0,0);
\draw (d)--(0,0);
\draw[fill=white] (a) circle (.25);
\draw[fill=white] (b) circle (.25);
\draw[fill=white] (c) circle (.25);
\draw[fill=white] (d) circle (.25);
\node at (0,1) { };
\node at (a) {$i$};
\node at (b) {$i$};
\node at (c) {$i$};
\node at (d) {$i$};
\end{tikzpicture}
&
\begin{tikzpicture}
\coordinate (a) at (-.75, 0);
\coordinate (b) at (0, .75);
\coordinate (c) at (.75, 0);
\coordinate (d) at (0, -.75);
\coordinate (aa) at (-.75,.5);
\coordinate (cc) at (.75,.5);
\draw (a)--(0,0);
\draw (b)--(0,0);
\draw (c)--(0,0);
\draw (d)--(0,0);
\draw[fill=white] (a) circle (.25);
\draw[fill=white] (b) circle (.25);
\draw[fill=white] (c) circle (.25);
\draw[fill=white] (d) circle (.25);
\node at (0,1) { };
\node at (a) {$i$};
\node at (b) {$j$};
\node at (c) {$i$};
\node at (d) {$j$};
\end{tikzpicture}
&
\begin{tikzpicture}
\coordinate (a) at (-.75, 0);
\coordinate (b) at (0, .75);
\coordinate (c) at (.75, 0);
\coordinate (d) at (0, -.75);
\coordinate (aa) at (-.75,.5);
\coordinate (cc) at (.75,.5);
\draw (a)--(0,0);
\draw (b)--(0,0);
\draw (c)--(0,0);
\draw (d)--(0,0);
\draw[fill=white] (a) circle (.25);
\draw[fill=white] (b) circle (.25);
\draw[fill=white] (c) circle (.25);
\draw[fill=white] (d) circle (.25);
\node at (0,1) { };
\node at (a) {$i$};
\node at (b) {$j$};
\node at (c) {$j$};
\node at (d) {$i$};
\end{tikzpicture}
\\
\midrule
q-zq^{-1} & z-1 & \begin{array} {c@{\hspace{10pt}}c} q-q^{-1}, & \text{if } i>j \\ z(q-q^{-1}), & \text{if } i<j \end{array} \\
   \bottomrule
\end{array}$}
\caption{The admissible vertices and Boltzmann weights for the standard $U_q(\mathfrak{gl}_n)$ $R$-matrix. Here, $i,j\in [1,n]$, and $z$ depends on the choice of $S$ or $T$.}
    \label{std-R-matrix}
\end{figure}

\subsection{Drinfeld-Reshetikhin Twisting}

Drinfeld-Reshetikhin twisting is a method to produce another quasitriangular Hopf algebra that is isomorphic to $\mathcal{H}$ as an algebra, but has a different comultiplication and $R$-matrix, resulting in new solutions of the Yang-Baxter equation. This flexibility is important in lattice model many contexts. For instance, a Drinfeld-Reshetikhin twist can allow us to work with \emph{stochastic} Boltzmann weights, which are important in integrable probability \cite{BorodinWheeler-bosonic}, or match a desired special function.

Reshetikhin \cite{Reshetikhin-twist} showed that if $F\in \mathcal{H}\otimes\mathcal{H}$ satisfies the conditions \begin{align} (\Delta\otimes \text{id})F = F_{13}F_{23}, \hspace{10pt} (\text{id}\otimes\Delta)F = F_{13}F_{12}, \hspace{10pt} F_{12}F_{13}F_{23} = F_{23}F_{13}F_{12}, \hspace{10pt} FF_{21} = 1, \label{twist-conditions}\end{align} then there exists a unique quasitriangular Hopf algebra $R^{(F)}$ with the same multiplication as $\mathcal{H}$, comultiplication given by $\Delta^{(F)}(a) := F\Delta(a)F^{-1}$, and universal $R$-matrix given by $R^{(F)} = F_{21}RF^{-1}$.

Let $U_q := U_q(\mathfrak{gl}_n)$. Reshetikhin showed that for for any set of nonzero complex numbers $\{f_{ij} |1\le i<j\le n\}$ the element \[F = \exp\left(\sum_{i<j} (H_i\otimes H_j - H_j\otimes H_i)f_{ij}\right) \in U_q\otimes U_q\] satisfies (\ref{twist-conditions}). Here, the $H_i$ are Cartan-like generators of $U_q$ \cite[\S~2]{Reshetikhin-twist}.

The standard $R$-matrix $R^{(F)}$ of the resulting twist $U_q^{(F)}$ is given by the following formula \cite[\S~4]{BBBF-Hecke-modules}: \begin{align*}R^{(F)}(z) = &\sum_i (q-zq^{-1}) e_{ii}\otimes e_{ii} + \sum_{i<j} \rho_{ij}(z-1) e_{ii}\otimes e_{jj} \\&+ \sum_{i>j} \rho_{ji}^{-1}(z-1)e_{ii}\otimes e_{jj} + (q-q^{-1})\sum_{i>j} e_{ij}\otimes e_{ji} + z(q-q^{-1})\sum_{i<j} e_{ij}\otimes e_{ji}.\end{align*} where $\rho_{ij} = (-1)^{\delta_{i\le m} +\delta_{j\le m}} \exp(-2f_{ij} + 2f_{i,j-1} + 2f_{i-1,j} - 2f_{i-1,j-1})$, and $f_{ii}$ is taken to be 0. When $i>j$, define $\rho_{ij} = \rho_{ji}^{-1}$. Note that the only coefficients changed by this twist are those corresponding to vectors of the form $e_{ii}\otimes e_{jj}$. The resulting Boltzmann weights are given in Figure \ref{twisted-std-R-matrix}.

\begin{figure}[h]
\centering
\scalebox{1.05}{$
\begin{array}{c@{\hspace{10pt}}c@{\hspace{10pt}}c}
\toprule
\vx{a_i} & \vx{b_{ij}} & \vx{c_{ij}} \\
\midrule
\begin{tikzpicture}
\coordinate (a) at (-.75, 0);
\coordinate (b) at (0, .75);
\coordinate (c) at (.75, 0);
\coordinate (d) at (0, -.75);
\coordinate (aa) at (-.75,.5);
\coordinate (cc) at (.75,.5);
\draw (a)--(0,0);
\draw (b)--(0,0);
\draw (c)--(0,0);
\draw (d)--(0,0);
\draw[fill=white] (a) circle (.25);
\draw[fill=white] (b) circle (.25);
\draw[fill=white] (c) circle (.25);
\draw[fill=white] (d) circle (.25);
\node at (0,1) { };
\node at (a) {$i$};
\node at (b) {$i$};
\node at (c) {$i$};
\node at (d) {$i$};
\end{tikzpicture}
&
\begin{tikzpicture}
\coordinate (a) at (-.75, 0);
\coordinate (b) at (0, .75);
\coordinate (c) at (.75, 0);
\coordinate (d) at (0, -.75);
\coordinate (aa) at (-.75,.5);
\coordinate (cc) at (.75,.5);
\draw (a)--(0,0);
\draw (b)--(0,0);
\draw (c)--(0,0);
\draw (d)--(0,0);
\draw[fill=white] (a) circle (.25);
\draw[fill=white] (b) circle (.25);
\draw[fill=white] (c) circle (.25);
\draw[fill=white] (d) circle (.25);
\node at (0,1) { };
\node at (a) {$i$};
\node at (b) {$j$};
\node at (c) {$i$};
\node at (d) {$j$};
\end{tikzpicture}
&
\begin{tikzpicture}
\coordinate (a) at (-.75, 0);
\coordinate (b) at (0, .75);
\coordinate (c) at (.75, 0);
\coordinate (d) at (0, -.75);
\coordinate (aa) at (-.75,.5);
\coordinate (cc) at (.75,.5);
\draw (a)--(0,0);
\draw (b)--(0,0);
\draw (c)--(0,0);
\draw (d)--(0,0);
\draw[fill=white] (a) circle (.25);
\draw[fill=white] (b) circle (.25);
\draw[fill=white] (c) circle (.25);
\draw[fill=white] (d) circle (.25);
\node at (0,1) { };
\node at (a) {$i$};
\node at (b) {$j$};
\node at (c) {$j$};
\node at (d) {$i$};
\end{tikzpicture}
\\
\midrule
q-zq^{-1} & \rho_{ij}(z-1) & \begin{array} {c@{\hspace{10pt}}c} q-q^{-1}, & \text{if } i>j \\ z(q-q^{-1}), & \text{if } i<j \end{array} \\
   \bottomrule
\end{array}$}
\caption{The admissible vertices and Boltzmann weights for the Drinfeld-Reshetikhin twist of the standard $U_q(\widehat{\mathfrak{gl}}_n)$ $R$-matrix.}
    \label{twisted-std-R-matrix}
\end{figure}

The upshot of this particular implementation of the Drinfeld-Reshetikhin twist is for the standard $R$-matrix for $U_q$, solvability is preserved when the weights $b_{ij}$ are replaced by $\rho_{ij}b_{ij}$, for any scalars $\rho_{ij}\in\mathbb{C}^\times$ such that $\rho_{ij}\rho_{ji}=1$.

\subsection{Generalized Twists}

Our main result of this section is that this fact is far more general than the case of the standard $R$-matrix for $U_q$. In fact, we can do a similar transformation to \emph{any} set of solvable $n$-color weights. It is notable that our result is purely combinatorial; a solution of the Yang-Baxter equation does not need to have an associated quantum group representation in order for twisting to work. In keeping with the theme of this paper, this allows us to construct and compare large classes of solutions to the Yang-Baxter equation without needing to use deep results in representation theory.

Tantalizingly, this result suggests that a sophisticated algebraic structure could potentially underlie our solutions to the $n$-color Yang-Baxter equation. If this turns out to be the case, it would be interesting to explore this object.

Fix $S$ and $T$ Boltzmann weights \[\{a_i(x), b_{ij}(x), c_{ij}(x) | i,j\in [0,n), x\in\{S,T\}\},\] satisfying the conditions in Theorem \ref{nlabelCond}. By that result, this defines a solvable lattice model with R-weights given by (\ref{n-color-parametrization}). Fix a set of parameters $D = \{\rho_{ij} \in \mathbb C^{\times}, i,j\in [0,n)\}$ such that \( \rho_{ij} \rho_{ji} = 1 \) for all $i$ and $j$. The twisted weights are then given by: \[\{a_i^{(D)}(x) := a_i(x), b_{ij}^{(D)}(x) := \rho_{ij}b_{ij}(x), c_{ij}^{(D)}(x) := c_{ij}(x)\}.\]

\begin{cor} \label{twist-corollary}
The twisted weights also satisfy the conditions in Theorem \ref{nlabelCond}.
\end{cor}

\begin{proof}
For any quantity $z$ involving the $S$ and $T$ Boltzmann weights, let $z^{(D)}$ denote $z$ with the $S$ and $T$ weights replaced by their twisted analogues. With this notation, one can check that: \begin{align} \beta_{ij}^{(D)} = \rho_{ij}\beta_{ij}, \hspace{10pt} \tau_{ij}^{(D)} = \tau_{ij}, \hspace{10pt} \alpha_{ij}^{(D)} = \alpha_{ij}, \hspace{10pt} \gamma_{ij}^{(D)} = \gamma_{ij}, \hspace{10pt} \Delta_{ij}^{(D)}(x) = \rho_{ji}\Delta_{ij}(x).\end{align} 

One can then use these expressions to show that the conditions in Theorem (\ref{nlabelCond}) are satisfied by the twisted weights. This is a straightforward check. For instance, after twisting, the condition \[\gamma_{jk}c_{jk}(S)c_{ij}(T) + \beta_{ij}\gamma_{jk}c_{ik}(S)b_{ji}(T) = \tau_{ij}\gamma_{ji}c_{ij}(S)c_{jk}(T) + \tau_{ij}\tau_{jk}\beta_{kj}\gamma_{ji}c_{ik}(S)b_{jk}(T)\] becomes \[\gamma_{jk}c_{jk}(S)c_{ij}(T) + \rho_{ij}\beta_{ij}\gamma_{jk}c_{ik}(S)\rho_{ji}b_{ji}(T) = \tau_{ij}\gamma_{ji}c_{ij}(S)c_{jk}(T) + \tau_{ij}\tau_{jk}\rho_{kj}\beta_{kj}\gamma_{ji}c_{ik}(S)\rho_{jk}b_{jk}(T),\] and the fact that $\rho_{ij}\rho_{ji}=1$, returns us to the original condition.
\end{proof}

There is also a similar transformation involving the $c_{ij}$ which also preserves solvability. Again, assume that the $S$ and $T$ weights satisfy the conditions in Theorem \ref{nlabelCond}, and fix another set of parameters $\mathcal{Z} = \{\zeta_{ij} \in \mathbb C^{\times}, i,j\in [0,n)\}$ such that \( \zeta_{ij} \zeta_{ji} = 1 \) for all $i,j\in [0,n)$, and \(\zeta_{ij}\zeta_{jk}\zeta_{ki} = 1\) for all $i,j,k\in [0,n)$.

Define a new set of Boltzmann weights by: \[\{a_i^{(\mathcal{Z})}(x) := a_i(x), b_{ij}^{(\mathcal{Z})}(x) := b_{ij}(x), c_{ij}^{(\mathcal{Z})}(x) := \zeta_{ij}c_{ij}(x)\}.\] It is then straightforward to show:

\begin{cor}
These weights also satisfy the conditions in Theorem \ref{nlabelCond}.
\end{cor}

We believe this latter transformation may be related to a change of basis formula (see \cite[Change~of~Basis~section]{BBBG-metaplectic-YBE}).

\bibliographystyle{siam}
\bibliography{bibliography.bib}

\newpage

\section*{Appendix: Enumeration of Yang-Baxter polynomials}\label{Appendix}
Here we enumerate all nontrivial (i.e. with admissible states) Yang-Baxter diagrams, along with their corresponding equations and 
Yang-Baxter polynomials (if nonzero). We separate the diagrams into two groupings: those that involve at most two distinct labels, and therefore appear in the two-color lattice model, and those that involve three distinct labels. Color-coding is used to indicate if a Yang-Baxter equation is vacuously true for all possible $R, S, T$ weights (red), if there are further constraints but an equal number of admissible states present on both sides (yellow), or if there are an unequal number of admissible states on each side of the equation (green).

\subsection*{A.1: One- and Two-Color Cases}
        \centering
        \begin{longtable}{|c|c|l|}
		\hline
		\LARGE $X^{iii}_{iii}$
		&\begin{minipage}{0.2\textwidth}
			\centering\vspace{0.5em}\scalebox{0.75}{\begin{tikzpicture}[scale=0.725,baseline=0.7cm]
            \lybd{i}{i}{i}{i}{i}{i}{i}{i}{i}
        \end{tikzpicture}}\vspace{0.5em}
		\end{minipage}
		{\LARGE $=$}
		\begin{minipage}{0.2\textwidth}
			\centering\vspace{0.5em}\scalebox{0.75}{\begin{tikzpicture}[scale=0.725,baseline=0.7cm]
             \rybd{i}{i}{i}{i}{i}{i}{i}{i}{i}
         \end{tikzpicture}}\vspace{0.5em}
		\end{minipage}
		& \cellcolor{red!50}
		\\
		\hline
		\multicolumn{3}{|c|}{\large\rule{0pt}{1.5em}\rule[-0.5em]{0pt}{1.5em}$A_{i}a_{i}(S)a_{i}(T)=A_{i}a_{i}(S)a_{i}(T)$} 
		\\
  \hline
    \end{longtable}
\begin{longtable}{|c|c|l|}
		\hline
		\LARGE $X^{iij}_{iij}$
		&\begin{minipage}{0.2\textwidth}
			\centering\vspace{0.5em}\scalebox{0.75}{\begin{tikzpicture}[scale=0.725,baseline=0.7cm]
 \lybd{i}{i}{j}{i}{i}{j}{j}{i}{i}
 \end{tikzpicture}}\vspace{0.5em}
		\end{minipage}
		{\LARGE $=$}
             \begin{minipage}{0.2\textwidth}
			\centering\vspace{0.5em}\scalebox{0.75}{\begin{tikzpicture}[scale=0.725,baseline=0.7cm]
   \rybd{i}{i}{j}{i}{i}{j}{j}{i}{i}
   \end{tikzpicture}}\vspace{0.5em}
		\end{minipage}
		& \cellcolor{red!50}
		\\
		\hline
		\multicolumn{3}{|c|}{\large\rule{0pt}{1.5em}\rule[-0.5em]{0pt}{1.5em}$A_{i}b_{ij}(S)b_{ij}(T)=A_{i}b_{ij}(S)b_{ij}(T)$} 
		\\
    \hline
    \end{longtable}
    \begin{longtable}{|c|c|l|}
		\hline
		\LARGE $X^{iij}_{iji}$
		&\begin{minipage}{0.2\textwidth}
			\centering\vspace{0.5em}\scalebox{0.75}{\begin{tikzpicture}[scale=0.725,baseline=0.7cm]
            \lybd{i}{i}{j}{i}{j}{i}{j}{i}{i}
            \end{tikzpicture}}\vspace{0.5em}
		\end{minipage}
		{\LARGE $=$}
		\begin{minipage}{0.2\textwidth}
			\centering\vspace{0.5em}\scalebox{0.75}{\begin{tikzpicture}[scale=0.725,baseline=0.7cm]
            \rybd{i}{i}{j}{i}{j}{i}{j}{j}{i}
            \end{tikzpicture}}\vspace{0.5em}
		\end{minipage}
		{\LARGE $+$}
		\begin{minipage}{0.2\textwidth}
			\centering\vspace{0.5em}\scalebox{0.75}{\begin{tikzpicture}[scale=0.725,baseline=0.7cm]
            \rybd{i}{i}{j}{i}{j}{i}{j}{j}{i}
            \end{tikzpicture}}\vspace{0.5em}
		\end{minipage}
		& \cellcolor{green!50}
		\\
		\hline
		\multicolumn{3}{|c|}{\large\rule{0pt}{1.5em}\rule[-0.5em]{0pt}{1.5em}$A_{i}b_{ij}(S)c_{ij}(T)=B_{ij}a_{i}(S)c_{ij}(T)+C_{ji}c_{ij}(S)b_{ij}(T)$} \\
   \multicolumn{3}{|c|}{\large\rule{0pt}{1.5em}\rule[-0.5em]{0pt}{1.5em}$Y_1(i,j)\coloneqq A_{i}b_{ij}(S)c_{ij}(T)- B_{ij}a_{i}(S)c_{ij}(T)-C_{ji}c_{ij}(S)b_{ij}(T)$} 
		\\
    \hline
    \end{longtable}

\newpage

\begin{longtable}{|c|c|l|}
		\hline
		\LARGE $X^{iij}_{jii}$
		&\begin{minipage}{0.2\textwidth}
			\centering\vspace{0.5em}\scalebox{0.75}{\begin{tikzpicture}[scale=0.725,baseline=0.7cm]
    \lybd{i}{i}{j}{j}{i}{i}{i}{i}{i}
    \end{tikzpicture}}\vspace{0.5em}
		\end{minipage}
		{\LARGE $=$}
		\begin{minipage}{0.2\textwidth}
			\centering\vspace{0.5em}\scalebox{0.75}{\begin{tikzpicture}[scale=0.725,baseline=0.7cm]
   \rybd{i}{i}{j}{j}{i}{i}{j}{j}{i}
   \end{tikzpicture}}\vspace{0.5em}
		\end{minipage}
		{\LARGE $+$}
		\begin{minipage}{0.2\textwidth}
			\centering\vspace{0.5em}\scalebox{0.75}{\begin{tikzpicture}[scale=0.725,baseline=0.7cm]
 \rybd{i}{i}{j}{j}{i}{i}{i}{i}{j}
 \end{tikzpicture}}\vspace{0.5em}
		\end{minipage}
		& \cellcolor{green!50}
		\\
		\hline
		\multicolumn{3}{|c|}{\large\rule{0pt}{1.5em}\rule[-0.5em]{0pt}{1.5em}$A_{i}c_{ij}(S)a_{i}(T)=B_{ji}c_{ij}(S)b_{ij}(T)+C_{ij}a_0(S)c_{ij}(T)$} \\
             \multicolumn{3}{|c|}{\large\rule{0pt}{1.5em}\rule[-0.5em]{0pt}{1.5em}$Y_2(i,j) \coloneqq A_{i}c_{ij}(S)a_{i}(T)-  B_{ji}c_{ij}(S)b_{ij}(T)-C_{ij}a_0(S)c_{ij}(T)$} 
		\\
  \hline
\end{longtable}
\begin{longtable}{|c|c|l|}
		\hline
		\LARGE $X^{iji}_{iij}$
		&\begin{minipage}{0.2\textwidth}
			\centering\vspace{0.5em}\scalebox{0.75}{\begin{tikzpicture}[scale=0.725,baseline=0.7cm]\lybd{i}{j}{i}{i}{i}{j}{i}{i}{j}\end{tikzpicture}}\vspace{0.5em}
		\end{minipage}
		{\LARGE $+$}
		\begin{minipage}{0.2\textwidth}
			\centering\vspace{0.5em}\scalebox{0.75}{\begin{tikzpicture}[scale=0.725,baseline=0.7cm]\lybd{i}{j}{i}{i}{i}{j}{j}{j}{i}\end{tikzpicture}}\vspace{0.5em}
		\end{minipage}
		{\LARGE $=$}
		\begin{minipage}{0.2\textwidth}
			\centering\vspace{0.5em}\scalebox{0.75}{\begin{tikzpicture}[scale=0.725,baseline=0.7cm]\rybd{i}{j}{i}{i}{i}{j}{j}{i}{i}\end{tikzpicture}}\vspace{0.5em}
		\end{minipage}
		& \cellcolor{green!50}
		\\
		\hline
		\multicolumn{3}{|c|}{\large\rule{0pt}{1.5em}\rule[-0.5em]{0pt}{1.5em}$B_{ij}a_{i}(S)c_{ji}(T)+C_{ij}c_{ji}(S)b_{ij}(T)=A_{i}b_{ij}(S)c_{ji}(T)$} 
		\\
  \multicolumn{3}{|c|}{\large\rule{0pt}{1.5em}\rule[-0.5em]{0pt}{1.5em}$Y_3(i,j) \coloneqq B_{ij}a_{i}(S)c_{ji}(T)+C_{ij}c_{ji}(S)b_{ij}(T)-A_{i}b_{ij}(S)c_{ji}(T)$}  \\
  \hline
       \end{longtable}
\begin{longtable}{|c|c|l|}
		\hline
		\LARGE $X^{iji}_{iji}$
		&\begin{minipage}{0.16\textwidth}
			\centering\vspace{0.5em}\scalebox{0.75}{\begin{tikzpicture}[scale=0.625,baseline=0.7cm]\lybd{i}{j}{i}{i}{j}{i}{i}{i}{j}\end{tikzpicture}}\vspace{0.5em}
		\end{minipage}
		{\Large $+$}
        \begin{minipage}{0.16\textwidth}
			\centering\vspace{0.5em}\scalebox{0.75}{\begin{tikzpicture}[scale=0.625,baseline=0.7cm]\lybd{i}{j}{i}{i}{j}{i}{j}{j}{i}\end{tikzpicture}}\vspace{0.5em}
		\end{minipage}
		{\Large $=$}
		\begin{minipage}{0.16\textwidth}
			\centering\vspace{0.5em}\scalebox{0.75}{\begin{tikzpicture}[scale=0.625,baseline=0.7cm]\rybd{i}{j}{i}{i}{j}{i}{i}{i}{j}\end{tikzpicture}}\vspace{0.5em}
		\end{minipage}
		{\Large $+$}
		\begin{minipage}{0.16\textwidth}
			\centering\vspace{0.5em}\scalebox{0.75}{\begin{tikzpicture}[scale=0.625,baseline=0.7cm]\rybd{i}{j}{i}{i}{j}{i}{j}{j}{i}\end{tikzpicture}}\vspace{0.5em}
    	\end{minipage}
		& \cellcolor{yellow!50}
		\\
		\hline
    	\multicolumn{3}{|c|}{\large\rule{0pt}{1.5em}\rule[-0.5em]{0pt}{1.5em}$B_{ij}a_{i}(S)b_{ji}(T)+C_{ij}c_{ji}(S)c_{ij}(T)=B_{ij}a_{i}(S)b_{ji}(T)+C_{ji}c_{ij}(S)c_{ji}(T)$} 
		\\
            \multicolumn{3}{|c|}{\large\rule{0pt}{1.5em}\rule[-0.5em]{0pt}{1.5em}$Y_4(i,j)\coloneqq C_{ij}c_{ji}(S)c_{ij}(T)-C_{ji}c_{ij}(S)c_{ji}(T)$} 
		\\
  \hline
  \end{longtable}
  \newpage
\begin{longtable}{|c|c|l|}
		\hline
		\LARGE $X^{iji}_{jii}$
		&\begin{minipage}{0.2\textwidth}
			\centering\vspace{0.5em}\scalebox{0.75}{\begin{tikzpicture}[scale=0.725,baseline=0.7cm]\lybd{i}{j}{i}{j}{i}{i}{i}{j}{i}\end{tikzpicture}}\vspace{0.5em}
		\end{minipage}
		{\LARGE $=$}
		\begin{minipage}{0.2\textwidth}
			\centering\vspace{0.5em}\scalebox{0.75}{\begin{tikzpicture}[scale=0.725,baseline=0.7cm]\rybd{i}{j}{i}{j}{i}{i}{i}{i}{j}\end{tikzpicture}}\vspace{0.5em}
		\end{minipage}
		{\LARGE $+$}
		\begin{minipage}{0.2\textwidth}
			\centering\vspace{0.5em}\scalebox{0.75}{\begin{tikzpicture}[scale=0.725,baseline=0.7cm]\rybd{i}{j}{i}{j}{i}{i}{j}{j}{i}\end{tikzpicture}}\vspace{0.5em}
		\end{minipage}
		& \cellcolor{green!50}
		\\
		\hline
    	\multicolumn{3}{|c|}{\large\rule{0pt}{1.5em}\rule[-0.5em]{0pt}{1.5em}$C_{ij}b_{ji}(S)a_{i}(T)=B_{ji}c_{ij}(S)c_{ji}(T)+C_{ij}a_{i}(S)b_{ji}(T)$} 
		\\
             \multicolumn{3}{|c|}{\large\rule{0pt}{1.5em}\rule[-0.5em]{0pt}{1.5em}$Y_5(i,j) \coloneqq C_{ij}b_{ji}(S)a_{i}(T)-B_{ji}c_{ij}(S)c_{ji}(T)-C_{ij}a_{i}(S)b_{ji}(T)$} 
		\\
 \hline
 \end{longtable}
\begin{longtable}{|c|c|l|}
		\hline
	\LARGE $X^{ijj}_{ijj}$ 
		&\begin{minipage}{0.2\textwidth}
			\centering\vspace{0.5em}\scalebox{0.75}{\begin{tikzpicture}[scale=0.725,baseline=0.7cm]\lybd{i}{j}{j}{i}{j}{j}{j}{i}{j}\end{tikzpicture}}\vspace{0.5em}
		\end{minipage}
		{\LARGE $=$}
		\begin{minipage}{0.2\textwidth}
			\centering\vspace{0.5em}\scalebox{0.75}{\begin{tikzpicture}[scale=0.725,baseline=0.7cm]\rybd{i}{j}{j}{i}{j}{j}{j}{i}{j}\end{tikzpicture}}\vspace{0.5em}
		\end{minipage}
		& \cellcolor{red!50}
		\\
		\hline
		\multicolumn{3}{|c|}{\large\rule{0pt}{1.5em}\rule[-0.5em]{0pt}{1.5em}$B_{ij}b_{ij}(S)a_{j}(T)=B_{ij}b_{ij}(S)a_{j}(T)$} 
		\\
 \hline
 \end{longtable}
\begin{longtable}{|c|c|l|}
		\hline
		\LARGE $X^{ijj}_{jij}$ 
		&\begin{minipage}{0.2\textwidth}
			\centering\vspace{0.5em}\scalebox{0.75}{\begin{tikzpicture}[scale=0.725,baseline=0.7cm]\lybd{i}{j}{j}{j}{i}{j}{i}{i}{j}\end{tikzpicture}}\vspace{0.5em}
		\end{minipage}
		{\LARGE $+$}
		\begin{minipage}{0.2\textwidth}
			\centering\vspace{0.5em}\scalebox{0.75}{\begin{tikzpicture}[scale=0.725,baseline=0.7cm]\lybd{i}{j}{j}{j}{i}{j}{j}{j}{i}\end{tikzpicture}}\vspace{0.5em}
		\end{minipage}
		{\LARGE $=$}
		\begin{minipage}{0.2\textwidth}
			\centering\vspace{0.5em}\scalebox{0.75}{\begin{tikzpicture}[scale=0.725,baseline=0.7cm]\rybd{i}{j}{j}{j}{i}{j}{j}{i}{j}\end{tikzpicture}}\vspace{0.5em}
		\end{minipage}
		& \cellcolor{green!50}
		\\
		\hline
		\multicolumn{3}{|c|}{\large\rule{0pt}{1.5em}\rule[-0.5em]{0pt}{1.5em}$B_{ij}c_{ij}(S)c_{ji}(T)+C_{ij}a_{j}(S)b_{ij}(T)=C_{ij}b_{ij}(S)a_{j}(T)$} 
		\\
  \multicolumn{3}{|c|}{\large\rule{0pt}{1.5em}\rule[-0.5em]{0pt}{1.5em}$Y_6(i,j) \coloneqq B_{ij}c_{ij}(S)c_{ji}(T)+C_{ij}a_{j}(S)b_{ij}(T)-C_{ij}b_{ij}(S)a_{j}(T)$} 
		\\
 \hline
 \end{longtable}
\begin{longtable}{|c|c|l|}
		\hline
		\LARGE $X^{ijj}_{jji}$ 
		&\begin{minipage}{0.2\textwidth}
			\centering\vspace{0.5em}\scalebox{0.75}{\begin{tikzpicture}[scale=0.725,baseline=0.7cm]\lybd{i}{j}{j}{j}{j}{i}{i}{i}{j}\end{tikzpicture}}\vspace{0.5em}
		\end{minipage}
		{\LARGE $+$}
		\begin{minipage}{0.2\textwidth}
			\centering\vspace{0.5em}\scalebox{0.75}{\begin{tikzpicture}[scale=0.725,baseline=0.7cm]\lybd{i}{j}{j}{j}{j}{i}{j}{j}{i}\end{tikzpicture}}\vspace{0.5em}
		\end{minipage}
		{\LARGE $=$}
		\begin{minipage}{0.2\textwidth}
			\centering\vspace{0.5em}\scalebox{0.75}{\begin{tikzpicture}[scale=0.725,baseline=0.7cm]\rybd{i}{j}{j}{j}{j}{i}{j}{j}{j}\end{tikzpicture}}\vspace{0.5em}
		\end{minipage}
		& \cellcolor{green!50}
		\\
		\hline
		\multicolumn{3}{|c|}{\large\rule{0pt}{1.5em}\rule[-0.5em]{0pt}{1.5em}$B_{ij}c_{ij}(S)b_{ji}(T)+C_{ij}a_{j}(S)c_{ij}(T)=A_{j}c_{ij}(S)a_{j}(T)$} 
		\\
        \multicolumn{3}{|c|}{\large\rule{0pt}{1.5em}\rule[-0.5em]{0pt}{1.5em}$Y_7(i,j) \coloneqq B_{ij}c_{ij}(S)b_{ji}(T)+C_{ij}a_{j}(S)c_{ij}(T)-A_{j}c_{ij}(S)a_{j}(T)$} 
		\\
 \hline
 \end{longtable}
 \newpage
\subsection*{A.2: Three-Color Cases}

\begin{longtable}{|c|c|l|}
		\hline
		\LARGE $X^{ijk}_{ijk}$
		&\begin{minipage}{0.2\textwidth}
			\centering\vspace{0.5em}\scalebox{0.75}{\begin{tikzpicture}[scale=0.725,baseline=0.7cm]\lybd{i}{j}{k}{i}{j}{k}{k}{i}{j}\end{tikzpicture}}\vspace{0.5em}
		\end{minipage}
		{\LARGE $=$}
		\begin{minipage}{0.2\textwidth}
			\centering\vspace{0.5em}\scalebox{0.75}{\begin{tikzpicture}[scale=0.725,baseline=0.7cm]\rybd{i}{j}{k}{i}{j}{k}{k}{i}{j}\end{tikzpicture}}\vspace{0.5em}
		\end{minipage}
		& \cellcolor{red!50}
		\\
		\hline
		\multicolumn{3}{|c|}{\large\rule{0pt}{1.5em}\rule[-0.5em]{0pt}{1.5em}$B_{ij}b_{ik}(S)b_{jk}(T)=B_{ij}b_{ik}(S)b_{jk}(T)$}
		\\
 \hline
 \end{longtable}
\begin{longtable}{|c|c|l|}
		\hline
		\LARGE $X^{ijk}_{ikj}$
		&\begin{minipage}{0.2\textwidth}
			\centering\vspace{0.5em}\scalebox{0.75}{\begin{tikzpicture}[scale=0.725,baseline=0.7cm]\lybd{i}{j}{k}{i}{k}{j}{k}{i}{j}\end{tikzpicture}}\vspace{0.5em}
		\end{minipage}
		{\LARGE $=$}
		\begin{minipage}{0.2\textwidth}
			\centering\vspace{0.5em}\scalebox{0.75}{\begin{tikzpicture}[scale=0.725,baseline=0.7cm]\rybd{i}{j}{k}{i}{k}{j}{j}{i}{k}\end{tikzpicture}}\vspace{0.5em}
		\end{minipage}
		& \cellcolor{yellow!50}
		\\
		\hline
		\multicolumn{3}{|c|}{\large\rule{0pt}{1.5em}\rule[-0.5em]{0pt}{1.5em}$B_{ij}b_{ik}(S)c_{jk}(T)=B_{ik}b_{ij}(S)c_{jk}(T)$}
		\\
  \multicolumn{3}{|c|}{\large\rule{0pt}{1.5em}\rule[-0.5em]{0pt}{1.5em}$Y_8(i,j,k) \coloneqq B_{ij}b_{ik}(S)c_{jk}(T)-B_{ik}b_{ij}(S)c_{jk}(T)$}
		\\
 \hline
 \end{longtable}
\begin{longtable}{|c|c|l|}
		\hline
		\LARGE $X^{ijk}_{jik}$
		&\begin{minipage}{0.2\textwidth}
			\centering\vspace{0.5em}\scalebox{0.75}{\begin{tikzpicture}[scale=0.725,baseline=0.7cm]\lybd{i}{j}{k}{j}{i}{k}{k}{j}{i}\end{tikzpicture}}\vspace{0.5em}
		\end{minipage}
		{\LARGE $=$}
		\begin{minipage}{0.2\textwidth}
			\centering\vspace{0.5em}\scalebox{0.75}{\begin{tikzpicture}[scale=0.725,baseline=0.7cm]\rybd{i}{j}{k}{j}{i}{k}{k}{i}{j}\end{tikzpicture}}\vspace{0.5em}
		\end{minipage}
		& \cellcolor{yellow!50}
		\\
		\hline
		\multicolumn{3}{|c|}{\large\rule{0pt}{1.5em}\rule[-0.5em]{0pt}{1.5em}$C_{ij}b_{jk}(S)b_{ik}(T)=C_{ij}b_{ik}(S)b_{jk}(T)$} 
		\\
  		\multicolumn{3}{|c|}{\large\rule{0pt}{1.5em}\rule[-0.5em]{0pt}{1.5em}$Y_9(i,j,k) \coloneqq C_{ij}b_{jk}(S)b_{ik}(T) - C_{ij}b_{ik}(S)b_{jk}(T)$} 
		\\
 \hline
 \end{longtable}
\begin{longtable}{|c|c|l|}
		\hline
		\LARGE $X^{ijk}_{kij}$
		&\begin{minipage}{0.2\textwidth}
			\centering\vspace{0.5em}\scalebox{0.75}{\begin{tikzpicture}[scale=0.725,baseline=0.7cm]\lybd{i}{j}{k}{k}{i}{j}{i}{i}{j}\end{tikzpicture}}\vspace{0.5em}
		\end{minipage}
		{\LARGE $+$}
		\begin{minipage}{0.2\textwidth}
			\centering\vspace{0.5em}\scalebox{0.75}{\begin{tikzpicture}[scale=0.725,baseline=0.7cm]\lybd{i}{j}{k}{k}{i}{j}{j}{j}{i}\end{tikzpicture}}\vspace{0.5em}
		\end{minipage}
		{\LARGE $=$}
		\begin{minipage}{0.2\textwidth}
			\centering\vspace{0.5em}\scalebox{0.75}{\begin{tikzpicture}[scale=0.725,baseline=0.7cm]\rybd{i}{j}{k}{k}{i}{j}{j}{j}{k}\end{tikzpicture}}\vspace{0.5em}
		\end{minipage}
		& \cellcolor{green!50}
		\\
		\hline
		\multicolumn{3}{|c|}{\large\rule{0pt}{1.5em}\rule[-0.5em]{0pt}{1.5em}$B_{ij}c_{ik}(S)b_{ji}(T)+C_{ij}c_{jk}(S)b_{ij}(T)=C_{ik}b_{ij}(S)c_{jk}(T)$} 
		\\
  		\multicolumn{3}{|c|}{\large\rule{0pt}{1.5em}\rule[-0.5em]{0pt}{1.5em}$Y_{10}(i,j,k) \coloneqq B_{ij}c_{ik}(S)b_{ji}(T)+C_{ij}c_{jk}(S)b_{ij}(T)-C_{ik}b_{ij}(S)c_{jk}(T)$} 
		\\
 \hline
 \end{longtable}

 \newpage
\begin{longtable}{|c|c|l|}
		\hline
		\LARGE $X^{ijk}_{jki}$
		&\begin{minipage}{0.2\textwidth}
			\centering\vspace{0.5em}\scalebox{0.75}{\begin{tikzpicture}[scale=0.725,baseline=0.7cm]\lybd{i}{j}{k}{j}{k}{i}{k}{j}{i}\end{tikzpicture}}\vspace{0.5em}
		\end{minipage}
		{\LARGE $=$}
		\begin{minipage}{0.2\textwidth}
			\centering\vspace{0.5em}\scalebox{0.75}{\begin{tikzpicture}[scale=0.725,baseline=0.7cm]\rybd{i}{j}{k}{j}{k}{i}{j}{j}{k}\end{tikzpicture}}\vspace{0.5em}
		\end{minipage}
		{\LARGE $+$}
		\begin{minipage}{0.2\textwidth}
			\centering\vspace{0.5em}\scalebox{0.75}{\begin{tikzpicture}[scale=0.725,baseline=0.7cm]\rybd{i}{j}{k}{j}{k}{i}{k}{k}{j}\end{tikzpicture}}\vspace{0.5em}
		\end{minipage}
		& \cellcolor{green!50}
		\\
		\hline
		\multicolumn{3}{|c|}{\large\rule{0pt}{1.5em}\rule[-0.5em]{0pt}{1.5em}$C_{ij}b_{jk}(S)c_{ik}(T)=B_{jk}c_{ij}(S)c_{jk}(T) + C_{kj}c_{ik}(S)b_{jk}(T)$} 
		\\
  		\multicolumn{3}{|c|}{\large\rule{0pt}{1.5em}\rule[-0.5em]{0pt}{1.5em}$Y_{11}(i,j,k) \coloneqq C_{ij}b_{jk}(S)c_{ik}(T)- B_{jk}c_{ij}(S)c_{jk}(T) - C_{kj}c_{ik}(S)b_{jk}(T)$} 
		\\
 \hline
 \end{longtable}
\begin{longtable}{|c|c|l|}
		\hline
		\LARGE $X^{ijk}_{kji}$
		&\begin{minipage}{0.16\textwidth}
			\centering\vspace{0.5em}\scalebox{0.75}{\begin{tikzpicture}[scale=0.625,baseline=0.7cm]\lybd{i}{j}{k}{k}{j}{i}{i}{i}{j}\end{tikzpicture}}\vspace{0.5em}
		\end{minipage}
		{\Large $+$}
		\begin{minipage}{0.16\textwidth}
			\centering\vspace{0.5em}\scalebox{0.75}{\begin{tikzpicture}[scale=0.625,baseline=0.7cm]\lybd{i}{j}{k}{k}{j}{i}{j}{j}{i}\end{tikzpicture}}\vspace{0.5em}
		\end{minipage}
		{\Large $=$}
		\begin{minipage}{0.16\textwidth}
			\centering\vspace{0.5em}\scalebox{0.75}{\begin{tikzpicture}[scale=0.625,baseline=0.7cm]\rybd{i}{j}{k}{k}{j}{i}{j}{j}{k}\end{tikzpicture}}\vspace{0.5em}
		\end{minipage}
		{\Large $+$}
		\begin{minipage}{0.16\textwidth}
			\centering\vspace{0.5em}\scalebox{0.75}{\begin{tikzpicture}[scale=0.625,baseline=0.7cm]\rybd{i}{j}{k}{k}{j}{i}{k}{k}{j}\end{tikzpicture}}\vspace{0.5em}
		\end{minipage}
		& \cellcolor{yellow!50}
		\\
		\hline
		\multicolumn{3}{|c|}{\large\rule{0pt}{1.5em}\rule[-0.5em]{0pt}{1.5em}$B_{ij}c_{ik}(S)b_{ji}(T)+C_{ij}c_{jk}(S)c_{ij}(T)=B_{kj}c_{ik}(S)b_{jk}(T) + C_{jk}c_{ij}(S)b_{jk}(T)$} 
		\\
  		\multicolumn{3}{|c|}{\large\rule{0pt}{1.5em}\rule[-0.5em]{0pt}{1.5em}{\small $Y_{12}(i,j,k) \coloneqq B_{ij}c_{ik}(S)b_{ji}(T)+C_{ij}c_{jk}(S)c_{ij}(T)- B_{kj}c_{ik}(S)b_{jk}(T) - C_{jk}c_{ij}(S)b_{jk}(T)$}} 
		\\
		\hline
	\end{longtable}

\end{document}


\section*{Appendix A.1: General YBD equivalence classes}
Here we enumerate all non-trivial (i.e. with admissible states) three color Yang-Baxter diagrams and their corresponding equations and polynomials, for any three labels or colors $i, j, k \in [1, n)$. Color coding is given to indicate if a Yang-Baxter equation is vacuously true for all possible $R, S, T$ weights (red), if there are further constraints but an equal number of admissible states is present on both sides (yellow), or if there are an unequal number of admissible states on each side of the equation (green).

\textcolor{blue}{add note here about corresponding YBE's, YB polynomials if applicable (not vacuously zero), color labeling}

\subsection*{Degenerate (One and Two Color) Cases}
	\centering
	\begin{longtable}{|c|c|l|}
		\hline
		\LARGE $X^{iii}_{iii}$
		&\begin{minipage}{0.2\textwidth}
			\centering\vspace{0.5em}\scalebox{0.75}{\begin{tikzpicture}[scale=0.725,baseline=0.7cm]
            \lybd{i}{i}{i}{i}{i}{i}{i}{i}{i}
        \end{tikzpicture}}\vspace{0.5em}
		\end{minipage}
		{\LARGE $=$}
		\begin{minipage}{0.2\textwidth}
			\centering\vspace{0.5em}\scalebox{0.75}{\begin{tikzpicture}[scale=0.725,baseline=0.7cm]
             \rybd{i}{i}{i}{i}{i}{i}{i}{i}{i}
         \end{tikzpicture}}\vspace{0.5em}
		\end{minipage}
		& \cellcolor{red!50}
		\\
		\hline
		\multicolumn{3}{|c|}{\large\rule{0pt}{1.5em}\rule[-0.5em]{0pt}{1.5em}$A_{i}a_{i}(S)a_{i}(T)=A_{i}a_{i}(S)a_{i}(T)$} 
		\\
  \hline
    \end{longtable}
\begin{longtable}{|c|c|l|}
		\hline
		\LARGE $X^{iij}_{iij}$
		&\begin{minipage}{0.2\textwidth}
			\centering\vspace{0.5em}\scalebox{0.75}{\begin{tikzpicture}[scale=0.725,baseline=0.7cm]
 \lybd{i}{i}{j}{i}{i}{j}{j}{i}{i}
 \end{tikzpicture}}\vspace{0.5em}
		\end{minipage}
		{\LARGE $=$}
             \begin{minipage}{0.2\textwidth}
			\centering\vspace{0.5em}\scalebox{0.75}{\begin{tikzpicture}[scale=0.725,baseline=0.7cm]
   \rybd{i}{i}{j}{i}{i}{j}{j}{i}{i}
   \end{tikzpicture}}\vspace{0.5em}
		\end{minipage}
		& \cellcolor{red!50}
		\\
		\hline
		\multicolumn{3}{|c|}{\large\rule{0pt}{1.5em}\rule[-0.5em]{0pt}{1.5em}$A_{i}b_{ij}(S)b_{ij}(T)=A_{i}b_{ij}(S)b_{ij}(T)$} 
		\\
    \hline
    \end{longtable}
    \begin{longtable}{|c|c|l|}
		\hline
		\LARGE $X^{iij}_{iji}$ 
		&\begin{minipage}{0.2\textwidth}
			\centering\vspace{0.5em}\scalebox{0.75}{\begin{tikzpicture}[scale=0.725,baseline=0.7cm]
            \lybd{i}{i}{j}{i}{j}{i}{j}{i}{i}
            \end{tikzpicture}}\vspace{0.5em}
		\end{minipage}
		{\LARGE $=$}
		\begin{minipage}{0.2\textwidth}
			\centering\vspace{0.5em}\scalebox{0.75}{\begin{tikzpicture}[scale=0.725,baseline=0.7cm]
            \rybd{i}{i}{j}{i}{j}{i}{j}{j}{i}
            \end{tikzpicture}}\vspace{0.5em}
		\end{minipage}
		{\LARGE $+$}
		\begin{minipage}{0.2\textwidth}
			\centering\vspace{0.5em}\scalebox{0.75}{\begin{tikzpicture}[scale=0.725,baseline=0.7cm]
            \rybd{i}{i}{j}{i}{j}{i}{j}{j}{i}
            \end{tikzpicture}}\vspace{0.5em}
		\end{minipage}
		& \cellcolor{green!50}
		\\
		\hline
		\multicolumn{3}{|c|}{\large\rule{0pt}{1.5em}\rule[-0.5em]{0pt}{1.5em}$A_{i}b_{ij}(S)c_{ij}(T)=B_{ij}a_{i}(S)c_{ij}(T)+C_{ji}c_{ij}(S)b_{ij}(T)$} \\
   \multicolumn{3}{|c|}{\large\rule{0pt}{1.5em}\rule[-0.5em]{0pt}{1.5em}$Y_1(i,j)\coloneqq A_{i}b_{ij}(S)c_{ij}(T)- B_{ij}a_{i}(S)c_{ij}(T)-C_{ji}c_{ij}(S)b_{ij}(T)$} 
		\\
    \hline
    \end{longtable}

\begin{longtable}{|c|c|l|}
		\hline
		\LARGE $X^{iij}_{jii}$
		&\begin{minipage}{0.2\textwidth}
			\centering\vspace{0.5em}\scalebox{0.75}{\begin{tikzpicture}[scale=0.725,baseline=0.7cm]
    \lybd{i}{i}{j}{j}{i}{i}{i}{i}{i}
    \end{tikzpicture}}\vspace{0.5em}
		\end{minipage}
		{\LARGE $=$}
		\begin{minipage}{0.2\textwidth}
			\centering\vspace{0.5em}\scalebox{0.75}{\begin{tikzpicture}[scale=0.725,baseline=0.7cm]
   \rybd{i}{i}{j}{j}{i}{i}{j}{j}{i}
   \end{tikzpicture}}\vspace{0.5em}
		\end{minipage}
		{\LARGE $+$}
		\begin{minipage}{0.2\textwidth}
			\centering\vspace{0.5em}\scalebox{0.75}{\begin{tikzpicture}[scale=0.725,baseline=0.7cm]
 \rybd{i}{i}{j}{j}{i}{i}{i}{i}{j}
 \end{tikzpicture}}\vspace{0.5em}
		\end{minipage}
		& \cellcolor{green!50}
		\\
		\hline
		\multicolumn{3}{|c|}{\large\rule{0pt}{1.5em}\rule[-0.5em]{0pt}{1.5em}$A_{i}c_{ij}(S)a_{i}(T)=B_{ji}c_{ij}(S)b_{ij}(T)+C_{ij}a_0(S)c_{ij}(T)$} \\
             \multicolumn{3}{|c|}{\large\rule{0pt}{1.5em}\rule[-0.5em]{0pt}{1.5em}$Y_2(i,j) \coloneqq A_{i}c_{ij}(S)a_{i}(T)-  B_{ji}c_{ij}(S)b_{ij}(T)-C_{ij}a_0(S)c_{ij}(T)$} 
		\\
  \hline
\end{longtable}
\begin{longtable}{|c|c|l|}
		\hline
		\LARGE $X^{iji}_{iij}$
		&\begin{minipage}{0.2\textwidth}
			\centering\vspace{0.5em}\scalebox{0.75}{\begin{tikzpicture}[scale=0.725,baseline=0.7cm]\lybd{i}{j}{i}{i}{i}{j}{i}{i}{j}\end{tikzpicture}}\vspace{0.5em}
		\end{minipage}
		{\LARGE $+$}
		\begin{minipage}{0.2\textwidth}
			\centering\vspace{0.5em}\scalebox{0.75}{\begin{tikzpicture}[scale=0.725,baseline=0.7cm]\lybd{i}{j}{i}{i}{i}{j}{j}{j}{i}\end{tikzpicture}}\vspace{0.5em}
		\end{minipage}
		{\LARGE $=$}
		\begin{minipage}{0.2\textwidth}
			\centering\vspace{0.5em}\scalebox{0.75}{\begin{tikzpicture}[scale=0.725,baseline=0.7cm]\rybd{i}{j}{i}{i}{i}{j}{j}{i}{i}\end{tikzpicture}}\vspace{0.5em}
		\end{minipage}
		& \cellcolor{green!50}
		\\
		\hline
		\multicolumn{3}{|c|}{\large\rule{0pt}{1.5em}\rule[-0.5em]{0pt}{1.5em}$B_{ij}a_{i}(S)c_{ji}(T)+C_{ij}c_{ji}(S)b_{ij}(T)=A_{i}b_{ij}(S)c_{ji}(T)$} 
		\\
  \multicolumn{3}{|c|}{\large\rule{0pt}{1.5em}\rule[-0.5em]{0pt}{1.5em}$Y_3(i,j) \coloneqq B_{ij}a_{i}(S)c_{ji}(T)+C_{ij}c_{ji}(S)b_{ij}(T)-A_{i}b_{ij}(S)c_{ji}(T)$}  \\
  \hline
       \end{longtable}
\begin{longtable}{|c|c|l|}
		\hline
		\LARGE $X^{iji}_{iji}$
		&\begin{minipage}{0.18\textwidth}
			\centering\vspace{0.5em}\scalebox{0.75}{\begin{tikzpicture}[scale=0.725,baseline=0.7cm]\lybd{i}{j}{i}{i}{j}{i}{i}{i}{j}\end{tikzpicture}}\vspace{0.5em}
		\end{minipage}
		{\LARGE $+$}
        \begin{minipage}{0.18\textwidth}
			\centering\vspace{0.5em}\scalebox{0.75}{\begin{tikzpicture}[scale=0.725,baseline=0.7cm]\lybd{i}{j}{i}{i}{j}{i}{j}{j}{i}\end{tikzpicture}}\vspace{0.5em}
		\end{minipage}
		{\LARGE $=$}
		\begin{minipage}{0.18\textwidth}
			\centering\vspace{0.5em}\scalebox{0.75}{\begin{tikzpicture}[scale=0.725,baseline=0.7cm]\rybd{i}{j}{i}{i}{j}{i}{i}{i}{j}\end{tikzpicture}}\vspace{0.5em}
		\end{minipage}
		{\LARGE $+$}
		\begin{minipage}{0.18\textwidth}
			\centering\vspace{0.5em}\scalebox{0.75}{\begin{tikzpicture}[scale=0.725,baseline=0.7cm]\rybd{i}{j}{i}{i}{j}{i}{j}{j}{i}\end{tikzpicture}}\vspace{0.5em}
    	\end{minipage}
		& \cellcolor{yellow!50}
		\\
		\hline
    	\multicolumn{3}{|c|}{\large\rule{0pt}{1.5em}\rule[-0.5em]{0pt}{1.5em}$B_{ij}a_{i}(S)b_{ji}(T)+C_{ij}c_{ji}(S)c_{ij}(T)=B_{ij}a_{i}(S)b_{ji}(T)+C_{ji}c_{ij}(S)c_{ji}(T)$} 
		\\
            \multicolumn{3}{|c|}{\large\rule{0pt}{1.5em}\rule[-0.5em]{0pt}{1.5em}$Y_4(i,j)\coloneqq C_{ij}c_{ji}(S)c_{ij}(T)-C_{ji}c_{ij}(S)c_{ji}(T)$} 
		\\
  \hline
  \end{longtable}
\begin{longtable}{|c|c|l|}
		\hline
		\LARGE $X^{iji}_{jii}$
		&\begin{minipage}{0.2\textwidth}
			\centering\vspace{0.5em}\scalebox{0.75}{\begin{tikzpicture}[scale=0.725,baseline=0.7cm]\lybd{i}{j}{i}{j}{i}{i}{i}{j}{i}\end{tikzpicture}}\vspace{0.5em}
		\end{minipage}
		{\LARGE $=$}
		\begin{minipage}{0.2\textwidth}
			\centering\vspace{0.5em}\scalebox{0.75}{\begin{tikzpicture}[scale=0.725,baseline=0.7cm]\rybd{i}{j}{i}{j}{i}{i}{i}{i}{j}\end{tikzpicture}}\vspace{0.5em}
		\end{minipage}
		{\LARGE $+$}
		\begin{minipage}{0.2\textwidth}
			\centering\vspace{0.5em}\scalebox{0.75}{\begin{tikzpicture}[scale=0.725,baseline=0.7cm]\rybd{i}{j}{i}{j}{i}{i}{j}{j}{i}\end{tikzpicture}}\vspace{0.5em}
		\end{minipage}
		& \cellcolor{green!50}
		\\
		\hline
    	\multicolumn{3}{|c|}{\large\rule{0pt}{1.5em}\rule[-0.5em]{0pt}{1.5em}$C_{ij}b_{ji}(S)a_{i}(T)=B_{ji}c_{ij}(S)c_{ji}(T)+C_{ij}a_{i}(S)b_{ji}(T)$} 
		\\
             \multicolumn{3}{|c|}{\large\rule{0pt}{1.5em}\rule[-0.5em]{0pt}{1.5em}$Y_5(i,j) \coloneqq C_{ij}b_{ji}(S)a_{i}(T)-B_{ji}c_{ij}(S)c_{ji}(T)-C_{ij}a_{i}(S)b_{ji}(T)$} 
		\\
 \hline
 \end{longtable}
\begin{longtable}{|c|c|l|}
		\hline
	\LARGE $X^{ijj}_{ijj}$ 
		&\begin{minipage}{0.2\textwidth}
			\centering\vspace{0.5em}\scalebox{0.75}{\begin{tikzpicture}[scale=0.725,baseline=0.7cm]\lybd{i}{j}{j}{i}{j}{j}{j}{i}{j}\end{tikzpicture}}\vspace{0.5em}
		\end{minipage}
		{\LARGE $=$}
		\begin{minipage}{0.2\textwidth}
			\centering\vspace{0.5em}\scalebox{0.75}{\begin{tikzpicture}[scale=0.725,baseline=0.7cm]\rybd{i}{j}{j}{i}{j}{j}{j}{i}{j}\end{tikzpicture}}\vspace{0.5em}
		\end{minipage}
		& \cellcolor{red!50}
		\\
		\hline
		\multicolumn{3}{|c|}{\large\rule{0pt}{1.5em}\rule[-0.5em]{0pt}{1.5em}$B_{ij}b_{ij}(S)a_{j}(T)=B_{ij}b_{ij}(S)a_{j}(T)$} 
		\\
 \hline
 \end{longtable}
\begin{longtable}{|c|c|l|}
		\hline
		\LARGE $X^{ijj}_{jij}$ 
		&\begin{minipage}{0.2\textwidth}
			\centering\vspace{0.5em}\scalebox{0.75}{\begin{tikzpicture}[scale=0.725,baseline=0.7cm]\lybd{i}{j}{j}{j}{i}{j}{i}{i}{j}\end{tikzpicture}}\vspace{0.5em}
		\end{minipage}
		{\LARGE $+$}
		\begin{minipage}{0.2\textwidth}
			\centering\vspace{0.5em}\scalebox{0.75}{\begin{tikzpicture}[scale=0.725,baseline=0.7cm]\lybd{i}{j}{j}{j}{i}{j}{j}{j}{i}\end{tikzpicture}}\vspace{0.5em}
		\end{minipage}
		{\LARGE $=$}
		\begin{minipage}{0.2\textwidth}
			\centering\vspace{0.5em}\scalebox{0.75}{\begin{tikzpicture}[scale=0.725,baseline=0.7cm]\rybd{i}{j}{j}{j}{i}{j}{j}{i}{j}\end{tikzpicture}}\vspace{0.5em}
		\end{minipage}
		& \cellcolor{green!50}
		\\
		\hline
		\multicolumn{3}{|c|}{\large\rule{0pt}{1.5em}\rule[-0.5em]{0pt}{1.5em}$B_{ij}c_{ij}(S)c_{ji}(T)+C_{ij}a_{j}(S)b_{ij}(T)=C_{ij}b_{ij}(S)a_{j}(T)$} 
		\\
  \multicolumn{3}{|c|}{\large\rule{0pt}{1.5em}\rule[-0.5em]{0pt}{1.5em}$Y_6(i,j) \coloneqq B_{ij}c_{ij}(S)c_{ji}(T)+C_{ij}a_{j}(S)b_{ij}(T)=C_{ij}b_{ij}(S)a_{j}(T)$} 
		\\
 \hline
 \end{longtable}
\begin{longtable}{|c|c|l|}
		\hline
		\LARGE $X^{ijj}_{jji}$ 
		&\begin{minipage}{0.2\textwidth}
			\centering\vspace{0.5em}\scalebox{0.75}{\begin{tikzpicture}[scale=0.725,baseline=0.7cm]\lybd{i}{j}{j}{j}{j}{i}{i}{i}{j}\end{tikzpicture}}\vspace{0.5em}
		\end{minipage}
		{\LARGE $+$}
		\begin{minipage}{0.2\textwidth}
			\centering\vspace{0.5em}\scalebox{0.75}{\begin{tikzpicture}[scale=0.725,baseline=0.7cm]\lybd{i}{j}{j}{j}{j}{i}{j}{j}{i}\end{tikzpicture}}\vspace{0.5em}
		\end{minipage}
		{\LARGE $=$}
		\begin{minipage}{0.2\textwidth}
			\centering\vspace{0.5em}\scalebox{0.75}{\begin{tikzpicture}[scale=0.725,baseline=0.7cm]\rybd{i}{j}{j}{j}{j}{i}{j}{j}{j}\end{tikzpicture}}\vspace{0.5em}
		\end{minipage}
		& \cellcolor{green!50}
		\\
		\hline
		\multicolumn{3}{|c|}{\large\rule{0pt}{1.5em}\rule[-0.5em]{0pt}{1.5em}$B_{ij}c_{ij}(S)b_{ji}(T)+C_{ij}a_{j}(S)c_{ij}(T)=A_{j}c_{ij}(S)a_{j}(T)$} 
		\\
        \multicolumn{3}{|c|}{\large\rule{0pt}{1.5em}\rule[-0.5em]{0pt}{1.5em}$Y_7(i,j) \coloneqq B_{ij}c_{ij}(S)b_{ji}(T)+C_{ij}a_{j}(S)c_{ij}(T)=A_{j}c_{ij}(S)a_{j}(T)$} 
		\\
 \hline
 \end{longtable}
 \subsection*{3 color cases}

\begin{longtable}{|c|c|l|}
		\hline
		\LARGE $X^{ijk}_{ijk}$
		&\begin{minipage}{0.2\textwidth}
			\centering\vspace{0.5em}\scalebox{0.75}{\begin{tikzpicture}[scale=0.725,baseline=0.7cm]\lybd{i}{j}{k}{i}{j}{k}{k}{i}{j}\end{tikzpicture}}\vspace{0.5em}
		\end{minipage}
		{\LARGE $=$}
		\begin{minipage}{0.2\textwidth}
			\centering\vspace{0.5em}\scalebox{0.75}{\begin{tikzpicture}[scale=0.725,baseline=0.7cm]\rybd{i}{j}{k}{i}{j}{k}{k}{i}{j}\end{tikzpicture}}\vspace{0.5em}
		\end{minipage}
		& \cellcolor{red!50}
		\\
		\hline
		\multicolumn{3}{|c|}{\large\rule{0pt}{1.5em}\rule[-0.5em]{0pt}{1.5em}$B_{ij}b_{ik}(S)b_{jk}(T)=B_{ij}b_{ik}(S)b_{jk}(T)$}
		\\
 \hline
 \end{longtable}
\begin{longtable}{|c|c|l|}
		\hline
		\LARGE $X^{ijk}_{ikj}$
		&\begin{minipage}{0.2\textwidth}
			\centering\vspace{0.5em}\scalebox{0.75}{\begin{tikzpicture}[scale=0.725,baseline=0.7cm]\lybd{i}{j}{k}{i}{k}{j}{k}{i}{j}\end{tikzpicture}}\vspace{0.5em}
		\end{minipage}
		{\LARGE $=$}
		\begin{minipage}{0.2\textwidth}
			\centering\vspace{0.5em}\scalebox{0.75}{\begin{tikzpicture}[scale=0.725,baseline=0.7cm]\rybd{i}{j}{k}{i}{k}{j}{j}{i}{k}\end{tikzpicture}}\vspace{0.5em}
		\end{minipage}
		& \cellcolor{yellow!50}
		\\
		\hline
		\multicolumn{3}{|c|}{\large\rule{0pt}{1.5em}\rule[-0.5em]{0pt}{1.5em}$B_{ij}b_{ik}(S)c_{jk}(T)=B_{ik}b_{ij}(S)c_{jk}(T)$}
		\\
  \multicolumn{3}{|c|}{\large\rule{0pt}{1.5em}\rule[-0.5em]{0pt}{1.5em}$Y_8(i,j,k) \coloneqq B_{ij}b_{ik}(S)c_{jk}(T)-B_{ik}b_{ij}(S)c_{jk}(T)$}
		\\
 \hline
 \end{longtable}
\begin{longtable}{|c|c|l|}
		\hline
		\LARGE $X^{ijk}_{jik}$
		&\begin{minipage}{0.2\textwidth}
			\centering\vspace{0.5em}\scalebox{0.75}{\begin{tikzpicture}[scale=0.725,baseline=0.7cm]\lybd{i}{j}{k}{j}{i}{k}{k}{j}{i}\end{tikzpicture}}\vspace{0.5em}
		\end{minipage}
		{\LARGE $=$}
		\begin{minipage}{0.2\textwidth}
			\centering\vspace{0.5em}\scalebox{0.75}{\begin{tikzpicture}[scale=0.725,baseline=0.7cm]\rybd{i}{j}{k}{j}{i}{k}{k}{i}{j}\end{tikzpicture}}\vspace{0.5em}
		\end{minipage}
		& \cellcolor{yellow!50}
		\\
		\hline
		\multicolumn{3}{|c|}{\large\rule{0pt}{1.5em}\rule[-0.5em]{0pt}{1.5em}$C_{ij}b_{jk}(S)b_{ik}(T)=C_{ij}b_{ik}(S)b_{jk}(T)$} 
		\\
  		\multicolumn{3}{|c|}{\large\rule{0pt}{1.5em}\rule[-0.5em]{0pt}{1.5em}$Y_9(i,j,k) \coloneqq C_{ij}b_{jk}(S)b_{ik}(T) - C_{ij}b_{ik}(S)b_{jk}(T)$} 
		\\
 \hline
 \end{longtable}
\begin{longtable}{|c|c|l|}
		\hline
		\LARGE $X^{ijk}_{kij}$
		&\begin{minipage}{0.2\textwidth}
			\centering\vspace{0.5em}\scalebox{0.75}{\begin{tikzpicture}[scale=0.725,baseline=0.7cm]\lybd{i}{j}{k}{k}{i}{j}{i}{i}{j}\end{tikzpicture}}\vspace{0.5em}
		\end{minipage}
		{\LARGE $+$}
		\begin{minipage}{0.2\textwidth}
			\centering\vspace{0.5em}\scalebox{0.75}{\begin{tikzpicture}[scale=0.725,baseline=0.7cm]\lybd{i}{j}{k}{k}{i}{j}{j}{j}{i}\end{tikzpicture}}\vspace{0.5em}
		\end{minipage}
		{\LARGE $=$}
		\begin{minipage}{0.2\textwidth}
			\centering\vspace{0.5em}\scalebox{0.75}{\begin{tikzpicture}[scale=0.725,baseline=0.7cm]\rybd{i}{j}{k}{k}{i}{j}{j}{j}{k}\end{tikzpicture}}\vspace{0.5em}
		\end{minipage}
		& \cellcolor{green!50}
		\\
		\hline
		\multicolumn{3}{|c|}{\large\rule{0pt}{1.5em}\rule[-0.5em]{0pt}{1.5em}$B_{ij}c_{ik}(S)b_{ji}(T)+C_{ij}c_{jk}(S)b_{ij}(T)=C_{ik}b_{ij}(S)c_{jk}(T)$} 
		\\
  		\multicolumn{3}{|c|}{\large\rule{0pt}{1.5em}\rule[-0.5em]{0pt}{1.5em}$Y_{10}(i,j,k) \coloneqq B_{ij}c_{ik}(S)b_{ji}(T)+C_{ij}c_{jk}(S)b_{ij}(T)-C_{ik}b_{ij}(S)c_{jk}(T)$} 
		\\
 \hline
 \end{longtable}
\begin{longtable}{|c|c|l|}
		\hline
		\LARGE $X^{ijk}_{jki}$
		&\begin{minipage}{0.2\textwidth}
			\centering\vspace{0.5em}\scalebox{0.75}{\begin{tikzpicture}[scale=0.725,baseline=0.7cm]\lybd{i}{j}{k}{j}{k}{i}{k}{j}{i}\end{tikzpicture}}\vspace{0.5em}
		\end{minipage}
		{\LARGE $=$}
		\begin{minipage}{0.2\textwidth}
			\centering\vspace{0.5em}\scalebox{0.75}{\begin{tikzpicture}[scale=0.725,baseline=0.7cm]\rybd{i}{j}{k}{j}{k}{i}{j}{j}{k}\end{tikzpicture}}\vspace{0.5em}
		\end{minipage}
		{\LARGE $+$}
		\begin{minipage}{0.2\textwidth}
			\centering\vspace{0.5em}\scalebox{0.75}{\begin{tikzpicture}[scale=0.725,baseline=0.7cm]\rybd{i}{j}{k}{j}{k}{i}{k}{k}{j}\end{tikzpicture}}\vspace{0.5em}
		\end{minipage}
		& \cellcolor{green!50}
		\\
		\hline
		\multicolumn{3}{|c|}{\large\rule{0pt}{1.5em}\rule[-0.5em]{0pt}{1.5em}$C_{ij}b_{jk}(S)c_{ik}(T)=B_{jk}c_{ij}(S)c_{jk}(T) + C_{kj}c_{ik}(S)b_{jk}(T)$} 
		\\
  		\multicolumn{3}{|c|}{\large\rule{0pt}{1.5em}\rule[-0.5em]{0pt}{1.5em}$Y_{11}(i,j,k) \coloneqq C_{ij}b_{jk}(S)c_{ik}(T)- B_{jk}c_{ij}(S)c_{jk}(T) - C_{kj}c_{ik}(S)b_{jk}(T)$} 
		\\
 \hline
 \end{longtable}
\begin{longtable}{|c|c|l|}
		\hline
		\LARGE $X^{ijk}_{kji}$
		&\begin{minipage}{0.18\textwidth}
			\centering\vspace{0.5em}\scalebox{0.75}{\begin{tikzpicture}[scale=0.725,baseline=0.7cm]\lybd{i}{j}{k}{k}{j}{i}{i}{i}{j}\end{tikzpicture}}\vspace{0.5em}
		\end{minipage}
		{\LARGE $+$}
		\begin{minipage}{0.18\textwidth}
			\centering\vspace{0.5em}\scalebox{0.75}{\begin{tikzpicture}[scale=0.725,baseline=0.7cm]\lybd{i}{j}{k}{k}{j}{i}{j}{j}{i}\end{tikzpicture}}\vspace{0.5em}
		\end{minipage}
		{\LARGE $=$}
		\begin{minipage}{0.18\textwidth}
			\centering\vspace{0.5em}\scalebox{0.75}{\begin{tikzpicture}[scale=0.725,baseline=0.7cm]\rybd{i}{j}{k}{k}{j}{i}{j}{j}{k}\end{tikzpicture}}\vspace{0.5em}
		\end{minipage}
		{\LARGE $+$}
		\begin{minipage}{0.18\textwidth}
			\centering\vspace{0.5em}\scalebox{0.75}{\begin{tikzpicture}[scale=0.725,baseline=0.7cm]\rybd{i}{j}{k}{k}{j}{i}{k}{k}{j}\end{tikzpicture}}\vspace{0.5em}
		\end{minipage}
		& \cellcolor{yellow!50}
		\\
		\hline
		\multicolumn{3}{|c|}{\large\rule{0pt}{1.5em}\rule[-0.5em]{0pt}{1.5em}$B_{ij}c_{ik}(S)b_{ji}(T)+C_{ij}c_{jk}(S)c_{ij}(T)=B_{kj}c_{ik}(S)b_{jk}(T) + C_{jk}c_{ij}(S)b_{jk}(T)$} 
		\\
  		\multicolumn{3}{|c|}{\large\rule{0pt}{1.5em}\rule[-0.5em]{0pt}{1.5em}$Y_{12}(i,j,k) \coloneqq B_{ij}c_{ik}(S)b_{ji}(T)+C_{ij}c_{jk}(S)c_{ij}(T)- B_{kj}c_{ik}(S)b_{jk}(T) - C_{jk}c_{ij}(S)b_{jk}(T)$} 
		\\
		\hline
	\end{longtable}
\newpage
\centering
\section{Appendix $A.2$: Sample YBEs from each equivalence class}

\begin{align*}
X^{000}_{000} \quad & A_{0} a_0(S) a_0(T) = A_{0} a_0(S) a_0(T) \tag{$Z^0(0)$} \\
X^{001}_{001} \quad & A_{0} b_{01}(S) b_{01}(T) = A_{0} b_{01}(S) b_{01}(T) \tag{$Z^1(0,1)$} \\
X^{001}_{010} \quad & A_{0} b_{01}(S) c_{01}(T) = B_{01} a_{0}(S) c_{01}(T) + C_{10} b_{01}(T) c_{01}(S) \tag{$Z^3(0,1)$} \\
X^{001}_{100} \quad & A_{0} a_{0}(T) c_{01}(S) = B_{10} b_{01}(T) c_{01}(S) + C_{01} a_{0}(S) c_{01}(T) \tag{$Z^4(0,1)$} \\
X^{010}_{001} \quad & B_{01} a_{0}(S) c_{10}(T) + C_{01} b_{01}(T) c_{10}(S) = A_{0} b_{01}(S) c_{10}(T) \tag{$Z^5(0,1)$} \\
X^{010}_{010} \quad & B_{01} a_{0}(S) b_{10}(T) + C_{01} c_{01}(T) c_{10}(S) = B_{01} a_{0}(S) b_{10}(T) + C_{10} c_{01}(S) c_{10}(T) \tag{$Z^6(0,1)$} \\
X^{010}_{100} \quad & C_{01} a_{0}(T) b_{10}(S) = B_{10} c_{01}(S) c_{10}(T) + C_{01} a_{0}(S) b_{10}(T) \tag{$Z^7(0,1)$} \\
X^{011}_{011} \quad & B_{01} b_{01}(S) a_1(T) = B_{01} b_{01}(S) a_0(T) \tag{$Z^8(0,1)$} \\
X^{011}_{101} \quad & B_{01} c_{01}(S) c_{10}(T) + C_{01} a_{1}(S) b_{01}(T) = C_{01} a_{1}(T) b_{01}(S) \tag{$Z^9(0,1)$} \\
X^{011}_{110} \quad & B_{01} b_{10}(T) c_{01}(S) + C_{01} a_{1}(S) c_{01}(T) = A_{1} a_{1}(T) c_{01}(S) \tag{$Z^{10}(0,1)$} \\ 
X^{001}_{001} \quad & B_{01}b_{03}(S)b_{13}(T) = B_{01}b_{03}(S)b_{13}(T) \tag{$Z^{11}(0,1,3)$}\\
X^{001}_{010} \quad & B_{00} b_{01}(S) c_{01}(T) = B_{01} b_{00}(S) c_{01}(T) \tag{$Z^{12}(0,1,3)$}\\
X^{001}_{001} \quad & C_{00} b_{01}(S) b_{01}(T) = C_{00} b_{01}(S) b_{01}(T) \tag{$Z^{13}(0,1,3)$}\\
X^{001}_{100} \quad & C_{00} c_{01}(S) b_{00}(T) + B_{00} c_{01}(S) c_{00}(T) = C_{01} b_{00}(S) c_{01}(T) \tag{$Z^{14}(0,1,3)$}\\
X^{001}_{010} \quad & C_{00} b_{01}(S) c_{01}(T) = C_{10} c_{01}(S) b_{01}(T) + B_{01} c_{00}(S) c_{01}(T) \tag{$Z^{15}(0,1,3)$}\\
X^{001}_{100} \quad & C_{00} c_{01}(S) c_{00}(T) + B_{00} c_{01}(S) b_{00}(T) = C_{01} c_{00}(S) c_{01}(T) + B_{10} c_{01}(S) b_{01}(T) \tag{$Z^{16}(0,1,3)$}\\
\end{align*}

	\bigskip
	
	Appendix B.1
	\centering
	
	Case 1 $(+,+,-,+,-,+)$:
    \begin{center}
        \begin{tikzpicture}
        \vertex{0}{0}{}{$-$}{$+$}{}{}
        \vertex{0}{-1}{}{$-$}{$-$}{$+$}{}
        \rvertex{-1}{-1}{$+$}{$+$}{$+$}{$+$}{}
    \end{tikzpicture}
    \begin{tikzpicture}
        \vertex{3}{0}{$+$}{$-$}{}{}{}
        \vertex{3}{-1}{$+$}{$+$}{}{$+$}{}
        \rvertex{4}{-1}{$+$}{$-$}{$+$}{$-$}{}
    \end{tikzpicture}
    \begin{tikzpicture}
        \vertex{3}{0}{$+$}{$-$}{}{}{}
        \vertex{3}{-1}{$+$}{$-$}{}{$+$}{}
        \rvertex{4}{-1}{$-$}{$+$}{$+$}{$-$}{}
    \end{tikzpicture}
    \end{center}
    
Here is the abstraction:
    \begin{center}
    \begin{tikzpicture}
            \node (a) at (0,0){$+$};
            
            \node[below=1cm of a] (b){$+$};
            
            \node[below=1cm of b] (c){$-$};
            
            \node[right=1cm of a] (d){$+$};
            
            \node[right=1cm of b] (e){$-$};
            
            \node[right=1cm of c] (f){$+$};
            
            \draw[->] (c) -- (e);
        \end{tikzpicture}
    \end{center}
    Case 2 $(+,+,-,-,+,+)$:

    \begin{center}
        \begin{tikzpicture}
        \vertex{0}{0}{}{$-$}{$-$}{}{}
        \vertex{0}{-1}{}{$+$}{$+$}{$+$}{}
        \rvertex{-1}{-1}{$+$}{$+$}{$+$}{$+$}{}
    \end{tikzpicture}
    \begin{tikzpicture}
        \vertex{3}{0}{$+$}{$-$}{}{}{}
        \vertex{3}{-1}{$+$}{$+$}{}{$+$}{}
        \rvertex{4}{-1}{$+$}{$-$}{$-$}{$+$}{}
    \end{tikzpicture}
    \begin{tikzpicture}
        \vertex{3}{0}{$+$}{$-$}{}{}{}
        \vertex{3}{-1}{$+$}{$-$}{}{$+$}{}
        \rvertex{4}{-1}{$-$}{$+$}{$-$}{$+$}{}
    \end{tikzpicture}
    \end{center}
    Furthermore, observe the abstraction of the boundary condition:
    \begin{center}
        \begin{tikzpicture}
            \node (a) at (0,0){$+$};
            
            \node[below=1cm of a] (b){$+$};
            
            \node[below=1cm of b] (c){$-$};
            
            \node[right=1cm of a] (d){$-$};
            
            \node[right=1cm of b] (e){$+$};
            
            \node[right=1cm of c] (f){$+$};
            
            \draw[->] (c) -- (d);
        \end{tikzpicture}
    \end{center}
    Again this provides an interesting result as the right model has one admissible state whereas the left has two admissible states.
    
    Case 3 $(+,-,+,+,+,-)$:
    \begin{center}
        \begin{tikzpicture}
        \vertex{0}{0}{}{$+$}{$+$}{}{}
        \vertex{0}{-1}{}{$-$}{$+$}{$-$}{}
        \rvertex{-1}{-1}{$+$}{$-$}{$-$}{$+$}{}
    \end{tikzpicture}
     \begin{tikzpicture}
        \vertex{0}{0}{}{$+$}{$+$}{}{}
        \vertex{0}{-1}{}{$+$}{$+$}{$-$}{}
        \rvertex{-1}{-1}{$+$}{$-$}{$+$}{$-$}{}
    \end{tikzpicture}
    \begin{tikzpicture}
        \vertex{3}{0}{$-$}{$+$}{}{}{}
        \vertex{3}{-1}{$+$}{$-$}{}{$-$}{}
        \rvertex{4}{-1}{$+$}{$+$}{$+$}{$+$}{}
    \end{tikzpicture}
    \end{center}
    Observe the abstraction of the boundary condition:
    \begin{center}
        \begin{tikzpicture}
            \node (a) at (0,0){$+$};
            
            \node[below=1cm of a] (b){$-$};
            
            \node[below=1cm of b] (c){$+$};
            
            \node[right=1cm of a] (d){$+$};
            
            \node[right=1cm of b] (e){$+$};
            
            \node[right=1cm of c] (f){$-$};
            
            \draw[->] (b) -- (f);
        \end{tikzpicture}
    \end{center}
    Again this boundary condition provides an interesting result as the right model has two admissible states whereas the left has one admissible state.
    
    Case 4 $(+,-,+,-,+,+)$:

    \begin{center}
        \begin{tikzpicture}
        \vertex{0}{0}{$-$}{$+$}{$-$}{}{}
        \vertex{0}{-1}{}{$+$}{$+$}{$+$}{}
        \rvertex{-1}{-1}{$+$}{$-$}{$-$}{$+$}{}
    \end{tikzpicture}
    \begin{tikzpicture}
        \vertex{3}{0}{$-$}{$+$}{}{}{}
        \vertex{3}{-1}{$+$}{$+$}{$+$}{$+$}{}
        \rvertex{4}{-1}{$+$}{$-$}{$-$}{$+$}{}
    \end{tikzpicture}
    \begin{tikzpicture}
        \vertex{3}{0}{$-$}{$+$}{}{}{}
        \vertex{3}{-1}{$+$}{$-$}{}{$+$}{}
        \rvertex{4}{-1}{$-$}{$+$}{$-$}{$+$}{}
    \end{tikzpicture}
    \end{center}
    Observe the abstraction of the boundary condition:
    \begin{center}
        \begin{tikzpicture}
            \node (a) at (0,0){$+$};
            
            \node[below=1cm of a] (b){$-$};
            
            \node[below=1cm of b] (c){$+$};
            
            \node[right=1cm of a] (d){$-$};
            
            \node[right=1cm of b] (e){$+$};
            
            \node[right=1cm of c] (f){$+$};
            
            \draw[->] (b) -- (d);
        \end{tikzpicture}
    \end{center}
    Again this boundary condition provides an interesting result as the left model has two admissible states whereas the right has one admissible state.
    
    Case 5 $(-,+,+,+,-,+)$:
    
    \begin{center}
        \begin{tikzpicture}
            \vertex{0}{0}{}{$+$}{$+$}{}{}
            \vertex{0}{-1}{}{$-$}{$-$}{$+$}{}
            \rvertex{-1}{-1}{$-$}{$+$}{$-$}{$+$}{}
        \end{tikzpicture}
        \begin{tikzpicture}
            \vertex{0}{0}{}{$+$}{$+$}{}{}
            \vertex{0}{-1}{}{$+$}{$-$}{$+$}{}
            \rvertex{-1}{-1}{$-$}{$+$}{$+$}{$-$}{}
        \end{tikzpicture}
        \begin{tikzpicture}
            \vertex{3}{0}{$+$}{$+$}{$+$}{$+$}{}
            \vertex{3}{-1}{$-$}{$+$}{}{$+$}{}
            \rvertex{4}{-1}{$-$}{$+$}{$+$}{$-$}{}
        \end{tikzpicture}
    \end{center}
    Observe the boundary condition:
    \begin{center}
        \begin{tikzpicture}
            \node (a) at (0,0){$-$};
            
            \node[below=1cm of a] (b){$+$};
            
            \node[below=1cm of b] (c){$+$};
            
            \node[right=1cm of a] (d){$+$};
            
            \node[right=1cm of b] (e){$-$};
            
            \node[right=1cm of c] (f){$+$};
            
            \draw[->] (a) -- (e);
        \end{tikzpicture}
    \end{center}
    Again, this boundary condition provides an interesting result as the left model has two admissible states whereas the right has one admissible state. 
    
    Case 6 $(-,+,+,+,+,-)$:
    
    \begin{center}
        \begin{tikzpicture}
            \vertex{0}{0}{}{$+$}{$+$}{}{}
            \vertex{0}{-1}{}{$+$}{$+$}{$-$}{}
            \rvertex{-1}{-1}{$-$}{$+$}{$+$}{$-$}{}
        \end{tikzpicture}
        \begin{tikzpicture}
            \vertex{0}{0}{}{$+$}{$+$}{}{}
            \vertex{0}{-1}{}{$-$}{$+$}{$-$}{}
            \rvertex{-1}{-1}{$-$}{$+$}{$-$}{$+$}{}
        \end{tikzpicture}
        \begin{tikzpicture}
            \vertex{3}{0}{$+$}{$+$}{}{}{}
            \vertex{3}{-1}{$-$}{$+$}{}{$-$}{}
            \rvertex{4}{-1}{$+$}{$+$}{$+$}{$+$}{}
        \end{tikzpicture}
    \end{center}
    Observe the abstraction of the boundary condition:
    \begin{center}
        \begin{tikzpicture}
            \node (a) at (0,0){$-$};
            
            \node[below=1cm of a] (b){$+$};
            
            \node[below=1cm of b] (c){$+$};
            
            \node[right=1cm of a] (d){$+$};
            
            \node[right=1cm of b] (e){$+$};
            
            \node[right=1cm of c] (f){$-$};
            
            \draw[->] (a) -- (f);
        \end{tikzpicture}
    \end{center}  being added to appendix
    Again, this boundary condition provides an interesting result as the left model has two admissible states whereas the right has one.
    Thus in conclusion, any model in which the boundary conditions are $\alpha\not\mapsto\delta$, $\beta\not\mapsto\epsilon$ and $\gamma\not\mapsto\eta$ we know that the result will be interesting. 
    
    \begin{center}
    Appendix B.2
    \end{center}
    
    Assume that there is a path connecting $R \mapsto R$, a path connecting $B \mapsto B$, and a path $+ \mapsto +$. We will now go through all the possible three color case such that we do not repeat cases that are 1 color model, i.e. a boundary condition like this (R,+,+,+,+,R) as this is equivalent to a previous boundary condition (-,+,+,+,+,-).

    Case 1 (R,B,+,+,R,B):
    \begin{center}
        \begin{tikzpicture}
            \vertex{0}{0}{}{$+$}{$+$}{}{}
            \vertex{0}{-1}{}{}{$\color{red}R$}{$\color{blue}B$}{}
            \rvertex{-1}{-1}{$\color{red}R$}{$\color{blue}B$}{}{}{} 
        \end{tikzpicture}
        \begin{tikzpicture}
            \vertex{3}{0}{$\color{blue}B$}{$+$}{}{}{}
            \vertex{3}{-1}{$\color{red}R$}{}{}{$\color{blue}B$}{}
            \rvertex{4}{-1}{}{}{$+$}{$\color{red}R$}{}
        \end{tikzpicture}
    \end{center}
    Notice, the left hand Yang-Baxter Equation has two admissible states, whereas the left has only one. Hence:
    
    \begin{center}
        \begin{tikzpicture}
            \vertex{0}{0}{}{$+$}{$+$}{}{}
            \vertex{0}{-1}{}{\color{red}$R$}{$\color{red}R$}{$\color{blue}B$}{}
            \rvertex{-1}{-1}{$\color{red}R$}{$\color{blue}B$}{$\color{red}R$}{$\color{blue}B$}{}
        \end{tikzpicture}
        \begin{tikzpicture}
            \vertex{0}{0}{}{$+$}{$+$}{}{}
            \vertex{0}{-1}{}{\color{blue}$B$}{$\color{red}R$}{$\color{blue}B$}{}
            \rvertex{-1}{-1}{$\color{red}R$}{$\color{blue}B$}{\color{blue}$B$}{\color{red}$R$}{}
        \end{tikzpicture}
        \begin{tikzpicture}
            \vertex{3}{0}{$\color{blue}B$}{$+$}{}{}{}
            \vertex{3}{-1}{$\color{red}R$}{\color{blue}$B$}{}{$\color{blue}B$}{}
            \rvertex{4}{-1}{\color{red}$R$}{$+$}{$+$}{$\color{red}R$}{}
        \end{tikzpicture}
    \end{center}
    Now, let us look at the abstraction for the boundary ins and outs:
    \begin{center}
        \begin{tikzpicture}
            \node (a) at (0,0){$R$};
            
            \node[below=1cm of a] (b){$B$};
            
            \node[below=1cm of b] (c){$+$};
            
            \node[right=1cm of a] (d){$+$};
            
            \node[right=1cm of b] (e){$R$};
            
            \node[right=1cm of c] (f){$B$};
            
            \draw[->] (a) -- (e);
            
            \draw[->] (b) -- (f);
            
            \draw[->] (c) -- (d);
        \end{tikzpicture}
    \end{center}
    Thus the abstraction has a desirable mapping as each in maps diagonally to one out. 
    
    Case 2 (+,R,B,R,B,+):
    \begin{center}
        \begin{tikzpicture}
            \vertex{0}{0}{}{\color{blue}$B$}{\color{red}$R$}{}{}
            \vertex{0}{-1}{}{}{$\color{blue}B$}{$+$}{}
            \rvertex{-1}{-1}{$+$}{$\color{red}R$}{}{}{} 
        \end{tikzpicture}
        \begin{tikzpicture}
            \vertex{3}{0}{$\color{red}R$}{\color{blue}$B$}{}{}{}
            \vertex{3}{-1}{$+$}{}{}{$+$}{}
            \rvertex{4}{-1}{}{}{\color{red}$R$}{$\color{blue}B$}{}
        \end{tikzpicture}
    \end{center}
    Notice that the model on the right has two admissible states whereas the left has one. 
    \begin{center}
        \begin{tikzpicture}
            \vertex{0}{0}{}{\color{blue}$B$}{\color{red}$R$}{}{}
            \vertex{0}{-1}{}{\color{blue}$B$}{$\color{blue}B$}{$+$}{}
            \rvertex{-1}{-1}{$+$}{$\color{red}R$}{\color{red}$R$}{$+$}{} 
        \end{tikzpicture}
         \begin{tikzpicture}
            \vertex{3}{0}{$\color{red}R$}{\color{blue}$B$}{}{}{}
            \vertex{3}{-1}{$+$}{\color{blue}$B$}{}{$+$}{}
            \rvertex{4}{-1}{\color{blue}$B$}{\color{red}$R$}{\color{red}$R$}{$\color{blue}B$}{}
        \end{tikzpicture}
        \begin{tikzpicture}
            \vertex{3}{0}{$\color{red}R$}{\color{blue}$B$}{}{}{}
            \vertex{3}{-1}{$+$}{\color{red}$R$}{}{$+$}{}
            \rvertex{4}{-1}{\color{red}$R$}{\color{blue}$B$}{\color{red}$R$}{$\color{blue}B$}{}
        \end{tikzpicture}
    \end{center}
    Once again, let us look at the abstraction:
    
    \begin{center}
        \begin{tikzpicture}
            \node (a) at (0,0){$+$};
            
            \node[below=1cm of a] (b){$R$};
            
            \node[below=1cm of b] (c){$B$};
            
            \node[right=1cm of a] (d){$R$};
            
            \node[right=1cm of b] (e){$B$};
            
            \node[right=1cm of c] (f){$+$};
            
            \draw[->] (a) -- (f);
            
            \draw[->] (b) -- (d);
            
            \draw[->] (c) -- (e);
        \end{tikzpicture}
    \end{center}
    We again observe a result which is desirable as $R \mapsto R$ diagonally, $B \mapsto B$ diagonally, and $+\mapsto+$ diagonally.
    
    Case 3 (R,B,+,B,+,R):
    
    \begin{center}
        \begin{tikzpicture}
            \vertex{0}{0}{}{$+$}{\color{blue}$B$}{}{}
            \vertex{0}{-1}{}{}{$+$}{\color{red}$R$}{}
            \rvertex{-1}{-1}{$\color{red}R$}{$\color{blue}B$}{}{}{} 
        \end{tikzpicture}
        \begin{tikzpicture}
            \vertex{3}{0}{$\color{blue}B$}{$+$}{}{}{}
            \vertex{3}{-1}{$\color{red}R$}{}{}{\color{red}$R$}{}
            \rvertex{4}{-1}{}{}{\color{blue}$B$}{$+$}{}
        \end{tikzpicture}
    \end{center}
    Notice that the right hand model has two admissible states whereas the left only has one state.
    \begin{center}
        \begin{tikzpicture}
            \vertex{0}{0}{}{$+$}{\color{blue}$B$}{}{}
            \vertex{0}{-1}{}{$+$}{$+$}{\color{red}$R$}{}
            \rvertex{-1}{-1}{$\color{red}R$}{$\color{blue}B$}{\color{blue}$B$}{\color{red}$R$}{} 
        \end{tikzpicture}
        \begin{tikzpicture}
            \vertex{3}{0}{$\color{blue}B$}{$+$}{}{}{}
            \vertex{3}{-1}{$\color{red}R$}{\color{blue}$B$}{}{\color{red}$R$}{}
            \rvertex{4}{-1}{\color{blue}$B$}{$+$}{\color{blue}$B$}{$+$}{}
        \end{tikzpicture}
        \begin{tikzpicture}
            \vertex{3}{0}{$\color{blue}B$}{$+$}{}{}{}
            \vertex{3}{-1}{$\color{red}R$}{$+$}{}{\color{red}$R$}{}
            \rvertex{4}{-1}{$+$}{\color{blue}$B$}{\color{blue}$B$}{$+$}{}
        \end{tikzpicture}
    \end{center}
    Let us now look at the abstraction: 
    \begin{center}
        \begin{tikzpicture}
            \node (a) at (0,0){$R$};
            
            \node[below=1cm of a] (b){$B$};
            
            \node[below=1cm of b] (c){$+$};
            
            \node[right=1cm of a] (d){$B$};
            
            \node[right=1cm of b] (e){$+$};
            
            \node[right=1cm of c] (f){$R$};
            
            \draw[->] (a) -- (f);
            
            \draw[->] (b) -- (d);
            
            \draw[->] (c) -- (e);
        \end{tikzpicture}
    \end{center}
    
    Hence, we observe a result which is desirable as $R \mapsto R$ diagonally, $B \mapsto B$ diagonally, and $+\mapsto+$ diagonally.
    
    Case 4 (B,+,R,+,R,B):
    \begin{center}
        \begin{tikzpicture}
            \vertex{0}{0}{}{\color{red}$R$}{$+$}{}{}
            \vertex{0}{-1}{}{}{\color{red}$R$}{\color{blue}$B$}{}
            \rvertex{-1}{-1}{$\color{blue}B$}{$+$}{}{}{} 
        \end{tikzpicture}
        \begin{tikzpicture}
            \vertex{3}{0}{$+$}{\color{red}$R$}{}{}{}
            \vertex{3}{-1}{\color{blue}$B$}{}{}{\color{blue}$B$}{}
            \rvertex{4}{-1}{}{}{$+$}{\color{red}$R$}{}
        \end{tikzpicture}
    \end{center}
    Again we find that the right model has two admissible states whereas the left model has only one admissible state. 
    \begin{center}
        \begin{tikzpicture}
            \vertex{0}{0}{}{\color{red}$R$}{$+$}{}{}
            \vertex{0}{-1}{}{\color{red}$R$}{\color{red}$R$}{\color{blue}$B$}{}
            \rvertex{-1}{-1}{$\color{blue}B$}{$+$}{$+$}{\color{blue}$B$}{} 
        \end{tikzpicture}
        \begin{tikzpicture}
            \vertex{3}{0}{$+$}{\color{red}$R$}{}{}{}
            \vertex{3}{-1}{\color{blue}$B$}{$+$}{}{\color{blue}$B$}{}
            \rvertex{4}{-1}{$+$}{\color{red}$R$}{$+$}{\color{red}$R$}{}
        \end{tikzpicture}
        \begin{tikzpicture}
            \vertex{3}{0}{$+$}{\color{red}$R$}{}{}{}
            \vertex{3}{-1}{\color{blue}$B$}{\color{red}$R$}{}{\color{blue}$B$}{}
            \rvertex{4}{-1}{\color{red}$R$}{$+$}{$+$}{\color{red}$R$}{}
        \end{tikzpicture}
    \end{center}
    We will now look at the abstraction of the ins and outs:
    \begin{center}
        \begin{tikzpicture}
            \node (a) at (0,0){$B$};
            
            \node[below=1cm of a] (b){$+$};
            
            \node[below=1cm of b] (c){$R$};
            
            \node[right=1cm of a] (d){$+$};
            
            \node[right=1cm of b] (e){$R$};
            
            \node[right=1cm of c] (f){$B$};
            
            \draw[->] (a) -- (f);
            
            \draw[->] (b) -- (d);
            
            \draw[->] (c) -- (e);
        \end{tikzpicture}
    \end{center}
    Hence, we observe a result which is desirable as $R \mapsto R$ diagonally, $B \mapsto B$ diagonally, and $+\mapsto+$ diagonally.
    
    Case 5 (B,+,R,R,B,+):
    \begin{center}
        \begin{tikzpicture}
            \vertex{0}{0}{}{\color{red}$R$}{\color{red}$R$}{}{}
            \vertex{0}{-1}{}{}{\color{blue}$B$}{$+$}{}
            \rvertex{-1}{-1}{\color{blue}$B$}{$+$}{}{}{} 
        \end{tikzpicture}
        \begin{tikzpicture}
            \vertex{3}{0}{$+$}{\color{red}$R$}{}{}{}
            \vertex{3}{-1}{\color{blue}$B$}{}{}{$+$}{}
            \rvertex{4}{-1}{}{}{\color{red}$R$}{\color{blue}$B$}{}
        \end{tikzpicture}
    \end{center}
    Observe that there are two admissible states for the left model whereas there is only one admissible state for the right model. 
    \begin{center}
        \begin{tikzpicture}
            \vertex{0}{0}{}{\color{red}$R$}{\color{red}$R$}{}{}
            \vertex{0}{-1}{}{$+$}{\color{blue}$B$}{$+$}{}
            \rvertex{-1}{-1}{\color{blue}$B$}{$+$}{$+$}{\color{blue}$B$}{} 
        \end{tikzpicture}
        \begin{tikzpicture}
            \vertex{0}{0}{}{\color{red}$R$}{\color{red}$R$}{}{}
            \vertex{0}{-1}{}{\color{blue}$B$}{\color{blue}$B$}{$+$}{}
            \rvertex{-1}{-1}{\color{blue}$B$}{$+$}{\color{blue}$B$}{$+$}{} 
        \end{tikzpicture}
        \begin{tikzpicture}
            \vertex{3}{0}{$+$}{\color{red}$R$}{}{}{}
            \vertex{3}{-1}{\color{blue}$B$}{$+$}{}{$+$}{}
            \rvertex{4}{-1}{\color{blue}$B$}{\color{red}$R$}{\color{red}$R$}{\color{blue}$B$}{}
        \end{tikzpicture}
    \end{center}
    We will now look at the abstraction of the ins and outs:
    \begin{center}
        \begin{tikzpicture}
            \node (a) at (0,0){$B$};
            
            \node[below=1cm of a] (b){$+$};
            
            \node[below=1cm of b] (c){$R$};
            
            \node[right=1cm of a] (d){$R$};
            
            \node[right=1cm of b] (e){$B$};
            
            \node[right=1cm of c] (f){$+$};
            
            \draw[->] (a) -- (e);
            
            \draw[->] (b) -- (f);
            
            \draw[->] (c) -- (d);
        \end{tikzpicture}
    \end{center}
    Hence, we observe a result which is desirable as $R \mapsto R$ diagonally, $B \mapsto B$ diagonally, and $+\mapsto+$ diagonally.
    
    Case 6 (R,+,B,+,B,R):
    
    \begin{center}
        \begin{tikzpicture}
            \vertex{0}{0}{}{\color{blue}$B$}{$+$}{}{}
            \vertex{0}{-1}{}{}{\color{blue}$B$}{\color{red}$R$}{}
            \rvertex{-1}{-1}{\color{red}$R$}{$+$}{}{}{} 
        \end{tikzpicture}
        \begin{tikzpicture}
            \vertex{3}{0}{$+$}{\color{blue}$B$}{}{}{}
            \vertex{3}{-1}{\color{red}$R$}{}{}{\color{red}$R$}{}
            \rvertex{4}{-1}{}{}{$+$}{\color{blue}$B$}{}
        \end{tikzpicture}
    \end{center}
    Notice the left model has one admissible state whereas, the left model has two admissible states.
    \begin{center}
        \begin{tikzpicture}
            \vertex{0}{0}{}{\color{blue}$B$}{$+$}{}{}
            \vertex{0}{-1}{}{\color{blue}$B$}{\color{blue}$B$}{\color{red}$R$}{}
            \rvertex{-1}{-1}{\color{red}$R$}{$+$}{$+$}{\color{red}$R$}{} 
        \end{tikzpicture}
        \begin{tikzpicture}
            \vertex{3}{0}{$+$}{\color{blue}$B$}{}{}{}
            \vertex{3}{-1}{\color{red}$R$}{\color{blue}$B$}{}{\color{red}$R$}{}
            \rvertex{4}{-1}{\color{blue}$B$}{$+$}{$+$}{\color{blue}$B$}{}
        \end{tikzpicture}
        \begin{tikzpicture}
            \vertex{3}{0}{$+$}{\color{blue}$B$}{}{}{}
            \vertex{3}{-1}{\color{red}$R$}{$+$}{}{\color{red}$R$}{}
            \rvertex{4}{-1}{$+$}{\color{blue}$B$}{$+$}{\color{blue}$B$}{}
        \end{tikzpicture}
    \end{center}
    We will now look at the abstraction of the ins and outs:
    \begin{center}
        \begin{tikzpicture}
            \node (a) at (0,0){$R$};
            
            \node[below=1cm of a] (b){$+$};
            
            \node[below=1cm of b] (c){$B$};
            
            \node[right=1cm of a] (d){$+$};
            
            \node[right=1cm of b] (e){$B$};
            
            \node[right=1cm of c] (f){$R$};
            
            \draw[->] (a) -- (f);
            
            \draw[->] (b) -- (d);
            
            \draw[->] (c) -- (e);
        \end{tikzpicture}
    \end{center}
    Hence, we observe a result which is desirable as $R \mapsto R$ diagonally, $B \mapsto B$ diagonally, and $+\mapsto+$ diagonally.
    
    Case 7 (R,+,B,B,R,+):
    \begin{center}
        \begin{tikzpicture}
            \vertex{0}{0}{}{\color{blue}$B$}{\color{blue}$B$}{}{}
            \vertex{0}{-1}{}{}{\color{red}$R$}{$+$}{}
            \rvertex{-1}{-1}{\color{red}$R$}{$+$}{}{}{} 
        \end{tikzpicture}
        \begin{tikzpicture}
            \vertex{3}{0}{$+$}{\color{blue}$B$}{}{}{}
            \vertex{3}{-1}{\color{red}$R$}{}{}{$+$}{}
            \rvertex{4}{-1}{}{}{\color{blue}$B$}{\color{red}$R$}{}
        \end{tikzpicture}
    \end{center}
    Observe that the left model has two states whereas the right model only has one.
    \begin{center}
        \begin{tikzpicture}
            \vertex{0}{0}{}{\color{blue}$B$}{\color{blue}$B$}{}{}
            \vertex{0}{-1}{}{\color{red}$R$}{\color{red}$R$}{$+$}{}
            \rvertex{-1}{-1}{\color{red}$R$}{$+$}{\color{red}$R$}{$+$}{} 
        \end{tikzpicture}
        \begin{tikzpicture}
            \vertex{0}{0}{}{\color{blue}$B$}{\color{blue}$B$}{}{}
            \vertex{0}{-1}{}{$+$}{\color{red}$R$}{$+$}{}
            \rvertex{-1}{-1}{\color{red}$R$}{$+$}{$+$}{\color{red}$R$}{} 
        \end{tikzpicture}
        \begin{tikzpicture}
            \vertex{3}{0}{$+$}{\color{blue}$B$}{}{}{}
            \vertex{3}{-1}{\color{red}$R$}{$+$}{}{$+$}{}
            \rvertex{4}{-1}{\color{red}$R$}{\color{blue}$B$}{\color{blue}$B$}{\color{red}$R$}{}
        \end{tikzpicture}
    \end{center}
    We will now look at the abstraction of the ins and outs:
    \begin{center}
        \begin{tikzpicture}
            \node (a) at (0,0){$R$};
            
            \node[below=1cm of a] (b){$+$};
            
            \node[below=1cm of b] (c){$B$};
            
            \node[right=1cm of a] (d){$B$};
            
            \node[right=1cm of b] (e){$R$};
            
            \node[right=1cm of c] (f){$+$};
            
            \draw[->] (a) -- (e);
            
            \draw[->] (b) -- (f);
            
            \draw[->] (c) -- (d);
        \end{tikzpicture}
    \end{center}
    Hence, we observe a result which is desirable as $R \mapsto R$ diagonally, $B \mapsto B$ diagonally, and $+\mapsto+$ diagonally.
    
    Case 8 (+,R,B,B,+,R):
    
    \begin{center}
        \begin{tikzpicture}
            \vertex{0}{0}{}{\color{blue}$B$}{\color{blue}$B$}{}{}
            \vertex{0}{-1}{}{}{$+$}{\color{red}$R$}{}
            \rvertex{-1}{-1}{$+$}{\color{red}$R$}{}{}{} 
        \end{tikzpicture}
        \begin{tikzpicture}
            \vertex{3}{0}{\color{red}$R$}{\color{blue}$B$}{}{}{}
            \vertex{3}{-1}{$+$}{}{}{\color{red}$R$}{}
            \rvertex{4}{-1}{}{}{\color{blue}$B$}{$+$}{}
        \end{tikzpicture}
    \end{center}
    Notice that the model on the left has two admissible states whereas the right has one. 
    \begin{center}
        \begin{tikzpicture}
            \vertex{0}{0}{}{\color{blue}$B$}{\color{blue}$B$}{}{}
            \vertex{0}{-1}{}{\color{red}$R$}{$+$}{\color{red}$R$}{}
            \rvertex{-1}{-1}{$+$}{\color{red}$R$}{\color{red}$R$}{$+$}{} 
        \end{tikzpicture}
         \begin{tikzpicture}
            \vertex{0}{0}{}{\color{blue}$B$}{\color{blue}$B$}{}{}
            \vertex{0}{-1}{}{$+$}{$+$}{\color{red}$R$}{}
            \rvertex{-1}{-1}{$+$}{\color{red}$R$}{$+$}{\color{red}$R$}{} 
        \end{tikzpicture}
        \begin{tikzpicture}
            \vertex{3}{0}{\color{red}$R$}{\color{blue}$B$}{}{}{}
            \vertex{3}{-1}{$+$}{\color{red}$R$}{}{\color{red}$R$}{}
            \rvertex{4}{-1}{$+$}{\color{blue}$B$}{\color{blue}$B$}{$+$}{}
        \end{tikzpicture}
    \end{center}
    We will now look at the abstraction of the ins and outs:
    \begin{center}
        \begin{tikzpicture}
            \node (a) at (0,0){$+$};
            
            \node[below=1cm of a] (b){$R$};
            
            \node[below=1cm of b] (c){$B$};
            
            \node[right=1cm of a] (d){$B$};
            
            \node[right=1cm of b] (e){$+$};
            
            \node[right=1cm of c] (f){$R$};
            
            \draw[->] (a) -- (e);
            
            \draw[->] (b) -- (f);
            
            \draw[->] (c) -- (d);
        \end{tikzpicture}
    \end{center}
    Hence, we observe a result which is desirable as $R \mapsto R$ diagonally, $B \mapsto B$ diagonally, and $+\mapsto+$ diagonally.
    
    Case 9 (+,B,R,B,R,+):
    \begin{center}
        \begin{tikzpicture}
            \vertex{0}{0}{}{\color{red}$R$}{\color{blue}$B$}{}{}
            \vertex{0}{-1}{}{}{\color{red}$R$}{$+$}{}
            \rvertex{-1}{-1}{$+$}{\color{blue}$B$}{}{}{} 
        \end{tikzpicture}
        \begin{tikzpicture}
            \vertex{3}{0}{\color{blue}$B$}{\color{red}$R$}{}{}{}
            \vertex{3}{-1}{$+$}{}{}{$+$}{}
            \rvertex{4}{-1}{}{}{\color{blue}$B$}{\color{red}$R$}{}
        \end{tikzpicture}
    \end{center}
    Notice that the model on the right has two admissible states whereas the left has one.
    
    \begin{center}
        \begin{tikzpicture}
            \vertex{0}{0}{}{\color{red}$R$}{\color{blue}$B$}{}{}
            \vertex{0}{-1}{}{\color{red}$R$}{\color{red}$R$}{$+$}{}
            \rvertex{-1}{-1}{$+$}{\color{blue}$B$}{\color{blue}$B$}{$+$}{} 
        \end{tikzpicture}
        \begin{tikzpicture}
            \vertex{3}{0}{\color{blue}$B$}{\color{red}$R$}{}{}{}
            \vertex{3}{-1}{$+$}{\color{blue}$B$}{}{$+$}{}
            \rvertex{4}{-1}{\color{blue}$B$}{\color{red}$R$}{\color{blue}$B$}{\color{red}$R$}{}
        \end{tikzpicture}\begin{tikzpicture}
            \vertex{3}{0}{\color{blue}$B$}{\color{red}$R$}{}{}{}
            \vertex{3}{-1}{$+$}{\color{red}$R$}{}{$+$}{}
            \rvertex{4}{-1}{\color{red}$R$}{\color{blue}$B$}{\color{blue}$B$}{\color{red}$R$}{}
        \end{tikzpicture}
    \end{center}
    
    We will now look at the abstraction of the ins and outs:
    \begin{center}
        \begin{tikzpicture}
            \node (a) at (0,0){$+$};
            
            \node[below=1cm of a] (b){$B$};
            
            \node[below=1cm of b] (c){$R$};
            
            \node[right=1cm of a] (d){$B$};
            
            \node[right=1cm of b] (e){$R$};
            
            \node[right=1cm of c] (f){$+$};
            
            \draw[->] (a) -- (f);
            
            \draw[->] (b) -- (d);
            
            \draw[->] (c) -- (e);
        \end{tikzpicture}
    \end{center}
    Hence, we observe a result which is desirable as $R \mapsto R$ diagonally, $B \mapsto B$ diagonally, and $+\mapsto+$ diagonally.
    
    Case 10 (+,B,R,R,+,B):
    \begin{center}
        \begin{tikzpicture}
            \vertex{0}{0}{}{\color{red}$R$}{\color{red}$R$}{}{}
            \vertex{0}{-1}{}{}{$+$}{\color{blue}$B$}{}
            \rvertex{-1}{-1}{$+$}{\color{blue}$B$}{}{}{} 
        \end{tikzpicture}
        \begin{tikzpicture}
            \vertex{3}{0}{\color{blue}$B$}{\color{red}$R$}{}{}{}
            \vertex{3}{-1}{$+$}{}{}{\color{blue}$B$}{}
            \rvertex{4}{-1}{}{}{\color{red}$R$}{$+$}{}
        \end{tikzpicture}
    \end{center}
    Notice that the left model has two admissible states, whereas the right model has only one.
    \begin{center}
        \begin{tikzpicture}
            \vertex{0}{0}{}{\color{red}$R$}{\color{red}$R$}{}{}
            \vertex{0}{-1}{}{\color{blue}$B$}{$+$}{\color{blue}$B$}{}
            \rvertex{-1}{-1}{$+$}{\color{blue}$B$}{\color{blue}$B$}{$+$}{} 
        \end{tikzpicture}
        \begin{tikzpicture}
            \vertex{0}{0}{}{\color{red}$R$}{\color{red}$R$}{}{}
            \vertex{0}{-1}{}{$+$}{$+$}{\color{blue}$B$}{}
            \rvertex{-1}{-1}{$+$}{\color{blue}$B$}{$+$}{\color{blue}$B$}{} 
        \end{tikzpicture}
        \begin{tikzpicture}
            \vertex{3}{0}{\color{blue}$B$}{\color{red}$R$}{}{}{}
            \vertex{3}{-1}{$+$}{\color{blue}$B$}{}{\color{blue}$B$}{}
            \rvertex{4}{-1}{$+$}{\color{red}$R$}{\color{red}$R$}{$+$}{}
        \end{tikzpicture}
    \end{center}
    We will now look at the abstraction of the ins and outs:
    \begin{center}
        \begin{tikzpicture}
            \node (a) at (0,0){$+$};
            
            \node[below=1cm of a] (b){$B$};
            
            \node[below=1cm of b] (c){$R$};
            
            \node[right=1cm of a] (d){$R$};
            
            \node[right=1cm of b] (e){$+$};
            
            \node[right=1cm of c] (f){$B$};
            
            \draw[->] (a) -- (e);
            
            \draw[->] (b) -- (f);
            
            \draw[->] (c) -- (d);
        \end{tikzpicture}
    \end{center}
    
    Hence, we observe a result which is desirable as $R \mapsto R$ diagonally, $B \mapsto B$ diagonally, and $+\mapsto+$ diagonally.
    
    Case 11 (B,R,+,+,B,R):
    \begin{center}
        \begin{tikzpicture}
            \vertex{0}{0}{}{$+$}{$+$}{}{}
            \vertex{0}{-1}{}{}{\color{blue}$B$}{\color{red}$R$}{}
            \rvertex{-1}{-1}{\color{blue}$B$}{\color{red}$R$}{}{}{} 
        \end{tikzpicture}
        \begin{tikzpicture}
            \vertex{3}{0}{\color{red}$R$}{$+$}{}{}{}
            \vertex{3}{-1}{\color{blue}$B$}{}{}{$\color{red}R$}{}
            \rvertex{4}{-1}{}{}{$+$}{\color{blue}$B$}{}
        \end{tikzpicture}
    \end{center}
    
    The left hand Yang-Baxter equation has two admissible states, whereas the right has only one. Hence:
    
    \begin{center}
        \begin{tikzpicture}
            \vertex{0}{0}{}{$+$}{$+$}{}{}
            \vertex{0}{-1}{}{\color{red}$R$}{\color{blue}$B$}{\color{red}$R$}{}
            \rvertex{-1}{-1}{\color{blue}$B$}{\color{red}$R$}{\color{red}$R$}{\color{blue}$B$}{} 
        \end{tikzpicture}
         \begin{tikzpicture}
            \vertex{0}{0}{}{$+$}{$+$}{}{}
            \vertex{0}{-1}{}{\color{blue}$B$}{\color{blue}$B$}{\color{red}$R$}{}
            \rvertex{-1}{-1}{\color{blue}$B$}{\color{red}$R$}{\color{blue}$B$}{\color{red}$R$}{} 
        \end{tikzpicture}
        \begin{tikzpicture}
            \vertex{3}{0}{\color{red}$R$}{$+$}{}{}{}
            \vertex{3}{-1}{\color{blue}$B$}{\color{red}$R$}{}{$\color{red}R$}{}
            \rvertex{4}{-1}{\color{blue}$B$}{$+$}{$+$}{\color{blue}$B$}{}
        \end{tikzpicture}
    \end{center}
    
    We will now look at the abstraction of the ins and outs:
    \begin{center}
        \begin{tikzpicture}
            \node (a) at (0,0){$B$};
            
            \node[below=1cm of a] (b){$R$};
            
            \node[below=1cm of b] (c){$+$};
            
            \node[right=1cm of a] (d){$+$};
            
            \node[right=1cm of b] (e){$B$};
            
            \node[right=1cm of c] (f){$R$};
            
            \draw[->] (a) -- (e);
            
            \draw[->] (b) -- (f);
            
            \draw[->] (c) -- (d);
        \end{tikzpicture}
    \end{center}
    
    Hence, we observe a result which is desirable as $R \mapsto R$ diagonally, $B \mapsto B$ diagonally, and $+\mapsto+$ diagonally.
    
    Case 12 (B,R,+,R,+,B):
    \begin{center}
        \begin{tikzpicture}
            \vertex{0}{0}{}{$+$}{\color{red}$R$}{}{}
            \vertex{0}{-1}{}{}{$+$}{\color{blue}$B$}{}
            \rvertex{-1}{-1}{\color{blue}$B$}{\color{red}$R$}{}{}{} 
        \end{tikzpicture}
        \begin{tikzpicture}
            \vertex{3}{0}{\color{red}$R$}{$+$}{}{}{}
            \vertex{3}{-1}{\color{blue}$B$}{}{}{\color{blue}$B$}{}
            \rvertex{4}{-1}{}{}{\color{red}$R$}{$+$}{}
        \end{tikzpicture}
    \end{center}
    The right hand Yang-Baxter equation has two admissible states, whereas the left has only one. Hence:
    
    \begin{center}
        \begin{tikzpicture}
            \vertex{0}{0}{}{$+$}{\color{red}$R$}{}{}
            \vertex{0}{-1}{}{$+$}{$+$}{\color{blue}$B$}{}
            \rvertex{-1}{-1}{\color{blue}$B$}{\color{red}$R$}{\color{red}$R$}{\color{blue}$B$}{} 
        \end{tikzpicture}
        \begin{tikzpicture}
            \vertex{3}{0}{\color{red}$R$}{$+$}{}{}{}
            \vertex{3}{-1}{\color{blue}$B$}{$+$}{}{\color{blue}$B$}{}
            \rvertex{4}{-1}{$+$}{\color{red}$R$}{\color{red}$R$}{$+$}{}
        \end{tikzpicture}
        \begin{tikzpicture}
            \vertex{3}{0}{\color{red}$R$}{$+$}{}{}{}
            \vertex{3}{-1}{\color{blue}$B$}{\color{red}$R$}{}{\color{blue}$B$}{}
            \rvertex{4}{-1}{\color{red}$R$}{$+$}{\color{red}$R$}{$+$}{}
        \end{tikzpicture}
    \end{center}
    We will now look at the abstraction of the ins and outs:
    \begin{center}
        \begin{tikzpicture}
            \node (a) at (0,0){$B$};
            
            \node[below=1cm of a] (b){$R$};
            
            \node[below=1cm of b] (c){$+$};
            
            \node[right=1cm of a] (d){$R$};
            
            \node[right=1cm of b] (e){$+$};
            
            \node[right=1cm of c] (f){$B$};
            
            \draw[->] (a) -- (f);
            
            \draw[->] (b) -- (d);
            
            \draw[->] (c) -- (e);
        \end{tikzpicture}
    \end{center}
    Hence, we observe a result which is desirable as $R \mapsto R$ diagonally, $B \mapsto B$ diagonally, and $+\mapsto+$ diagonally.
\newpage